\documentclass[11pt,a4paper,reqno]{amsart}%
\usepackage[utf8]{inputenc}
\usepackage{mathrsfs}
\usepackage{dsfont}
\usepackage{hyperref}
\usepackage{amsmath}
\usepackage{amssymb}
\usepackage{amsthm}
\usepackage{amsfonts}
\usepackage{amstext}
\usepackage{amsopn}
\usepackage{amsxtra}
\usepackage{mathrsfs}
\usepackage{dsfont}
\usepackage{esint}
\usepackage{graphicx}
\newtheorem{theorem}{Theorem}
\newtheorem{proposition}{Proposition}
\newtheorem{remark}{Remark}

\newtheorem{lemma}{Lemma}
\newtheorem{corollary}{Corollary}

\newcommand{\dps}{\displaystyle}
\newcommand{\ii}{\infty}
\newcommand\R{{\ensuremath {\mathbb R} }}

\newcommand\Z{{\ensuremath {\mathbb Z} }}

\newcommand\1{{\ensuremath {\mathds 1} }}
\renewcommand\phi{\varphi}
\newcommand{\gH}{\mathfrak{H}}
\newcommand{\gS}{\mathfrak{S}}
\newcommand{\gK}{\mathfrak{K}}

\newcommand{\cS}{\mathcal{S}}

\newcommand{\cY}{\mathcal{Y}}

\newcommand{\cC}{\mathcal{C}}

\newcommand{\cX}{\mathcal{X}}
\newcommand{\cB}{\mathcal{B}}
\newcommand{\cK}{\mathcal{K}}
\newcommand{\cE}{\mathcal{E}}
\newcommand{\cW}{\mathcal{W}}
\newcommand{\cF}{\mathcal{F}}

\newcommand{\cH}{\mathcal{H}}

\renewcommand{\epsilon}{\varepsilon}
\newcommand\pscal[1]{{\ensuremath{\left\langle #1 \right\rangle}}}
\newcommand{\norm}[1]{ \left| \! \left| #1 \right| \! \right| }
\newcommand{\tr}{{\rm Tr}}

\renewcommand{\ge}{\geqslant}
\renewcommand{\le}{\leqslant}
\renewcommand{\geq}{\geqslant}
\renewcommand{\leq}{\leqslant}
\renewcommand{\hat}{\widehat}
\renewcommand{\tilde}{\widetilde}

\title[Hartree equation for infinitely many particles I]{The Hartree equation for infinitely many particles\\ I. Well-posedness theory}

\author[M. Lewin]{Mathieu LEWIN}
\address{CNRS \& Universit\'e de Cergy-Pontoise, Mathematics Department (UMR 8088), F-95000 Cergy-Pontoise, France} 
\email{mathieu.lewin@math.cnrs.fr}

\author[J. Sabin]{Julien SABIN}
\address{Universit\'e de Cergy-Pontoise, Mathematics Department (UMR 8088), F-95000 Cergy-Pontoise, France} 
\email{julien.sabin@u-cergy.fr}

\date{\today}

\begin{document}

\begin{abstract}
  We show local and global well-posedness results for the Hartree equation
  $$i\partial_t\gamma=[-\Delta+w*\rho_\gamma,\gamma],$$
  where $\gamma$ is a bounded self-adjoint operator on $L^2(\R^d)$, $\rho_\gamma(x)=\gamma(x,x)$ and $w$ is a smooth short-range interaction potential. The initial datum $\gamma(0)$ is assumed to be a perturbation of a translation-invariant state $\gamma_f=f(-\Delta)$ which describes a quantum system with an infinite number of particles, such as the Fermi sea at zero temperature, or the Fermi-Dirac and Bose-Einstein gases at positive temperature. Global well-posedness follows from the conservation of the relative (free) energy of the state $\gamma(t)$, counted relatively to the stationary state $\gamma_f$. We indeed use a general notion of relative entropy, which allows to treat a wide class of stationary states $f(-\Delta)$. Our results are based on a Lieb-Thirring inequality at positive density and on a recent Strichartz inequality for orthonormal functions, which are both due to Frank, Lieb, Seiringer and the first author of this article.
\end{abstract}

\thanks{\copyright\,2013 by the authors. This paper may be reproduced, in its entirety, for non-commercial purposes. To appear in \emph{Comm. Math. Phys.}}

\maketitle

\tableofcontents

\section{Introduction}

The time-dependent Hartree equation 
\begin{equation}
\left\{\begin{array}{rcl}
i\, \partial_t u&=&\big(-\Delta_x+w\ast|u|^2\big)u,\qquad (t,x)\in\R\times\R^d\\
u(0,x)&=&u_0(x).
\end{array}\right.
\label{eq:Hartree}
\end{equation}
describes the dynamics of a Bose-Einstein condensate in $\R^d$, in which all the quantum particles are in the same state $u(t,x)$ (a normalized square integrable function, $\int_{\R^d}|u(t,x)|^2\,dx=1$). The equation~\eqref{eq:Hartree} can be derived from many-body quantum mechanics in a mean-field limit~\cite{Hepp-74,GinVel-79,Spohn-80,ErdYau-01,RodSch-09,KnoPic-10,Pickl-11,AmmNie-08,LewNamSch-13}. The nonlinear term $w\ast|u|^2$ describes the interactions between the particles, where $w:\R^d\to\R$ is the pair interaction potential and $\ast$ is a notation for the convolution of two functions in $\R^d$. Physically, $w$ is a short range potential and it is sometimes even taken to be a Dirac $\delta$, leading to the so-called Gross-Pitaevskii or cubic nonlinear Schr\"odinger (NLS) equation.

There is a similar theory for fermions, which is called \emph{reduced Hartree-Fock} in the literature~\cite{Solovej-91}, or often simply \emph{Hartree} (it corresponds to ignoring the exchange term in the Hartree-Fock model). On the contrary to bosons, fermions cannot occupy the same state and a system of $N$ particles is then described by a set of $N$ orthonormal functions $u_1,...,u_N\in L^2(\R^d)$. Their dynamics is modelled by a system of $N$ coupled Hartree equations of the previous form:
\begin{equation}
\left\{\begin{array}{rcl}
i\, \partial_t u_1&=&\dps\left(-\Delta_x+w\ast\left(\sum_{k=1}^N|u_k|^2\right)\right)u_1,\\
&\vdots&\\
i\, \partial_t u_N&=&\dps\left(-\Delta_x+w\ast\left(\sum_{k=1}^N|u_k|^2\right)\right)u_N,\\
u_j(0,x)&=&u_{0,j}(x),\qquad j=1,...,N.\phantom{\int}
\end{array}\right.
\label{eq:rHF}
\end{equation}
The function 
$$\rho(t,x):=\sum_{j=1}^N|u_j(t,x)|^2$$
is the total density of particles in the system at time $t$. Since all the equations~\eqref{eq:rHF} involve the same (so-called \emph{mean-field}) operator $H(t)=-\Delta +w\ast\rho(t)$, it is clear that the system $(u_1(t),...,u_N(t))$ remains orthonormal for every time $t$. Indeed, we have $u_j(t)=U(t,0)u_j(0)$ where $U(t,t')$ is the unitary propagator associated with the time-dependent Hamiltonian $H(t)$. In particular, the number of particles $\int_{\R^d}\rho(t)=N$ is conserved for all times. 
Another conserved quantity along the flow is the nonlinear Hartree energy
\begin{equation*}
\cE(u_1,...,u_N):=\sum_{j=1}^N\int_{\R^d}|\nabla u_j|^2\\+\frac12 \int_{\R^d}\int_{\R^d}w(x-y)\rho(t,x)\rho(t,y)\,dx\,dy. 
\end{equation*}
These two conservation laws can be used to prove global well-posedness of~\eqref{eq:rHF} under suitable assumptions on $w$ (for instance that $w\geq0$). The equations~\eqref{eq:rHF} can also be derived from many-body quantum mechanics~\cite{BarGolGotMau-03,BarErdGolMauYau-02,ElgErdSchYau-04,FroKno-11,BenPorSch-13}, in a mean-field or semi-classical limit.

The purpose of this paper is to study the well-posedness of the equation~\eqref{eq:rHF} in the case where $N=\ii$, that is, for a gas containing infinitely many particles. In this situation, the total Hartree energy of the system is also infinite. We consider both the zero and positive temperature cases. To our knowledge, the results contained in this article are the first of this kind for the Hartree model with infinitely many particles. A similar problem has been considered before at zero temperature for Dirac particles in~\cite{HaiLewSpa-05} and for crystals in~\cite{CanSto-11}, but there the mean-field operator $H(t)$ has a gap in its spectrum, which dramatically simplifies the study.

In order to put the Hartree equation for infinitely many particles on a solid ground, it is easier to use the formalism of \emph{density matrices}. We first explain this for a finite $N$, before turning to $N=\ii$. The idea is to introduce the operator
$$\gamma(t):=\sum_{j=1}^N|u_j(t)\rangle\langle u_j(t)|,$$
which is the rank-$N$ orthogonal projection\footnote{Here and everywhere, we use Dirac's notation $|u\rangle\langle v|$ for the operator $f\mapsto \pscal{v,f}u$. Our scalar product is always anti-linear with respect to the left argument.} onto the $N$-dimensional space spanned by the functions $u_1(t),...,u_N(t)$. The operator $\gamma(t)$ is our new variable and it is called the \emph{one-particle density matrix}. The equation~\eqref{eq:rHF} can be equivalently re-written in terms of $\gamma(t)$ as
\begin{equation}
\boxed{\phantom{\int}i\, \partial_t \gamma=\big[-\Delta+w\ast\rho_\gamma\,,\, \gamma\big]\phantom{\int}}
\label{eq:rHF_DM}
\end{equation}
with the initial datum
\begin{equation}
\gamma(0)=\sum_{k=1}^N|u_k(0)\rangle\langle u_k(0)|.
\label{eq:finite_rank_initial}
\end{equation}
Here $\rho_\gamma(t,x)=\gamma(t,x,x)$ is the density associated with the operator $\gamma$, which coincides with the density of particles $\rho(t,x)$ introduced earlier ($\gamma(t,x,y)$ is a notation for the integral kernel of $\gamma(t)$). The average particle number and the average energy of the system can now be written in terms of $\gamma$ only as $N=\tr(\gamma)$ and
$$\cE(\gamma)=\tr\big((-\Delta)\gamma\big)+\frac12\int_{\R^d}\int_{\R^d}w(x-y)\rho_\gamma(x)\rho_\gamma(y)\,dx\,dy.$$

The formulation~\eqref{eq:rHF_DM} is clearly more adapted to the study of infinite systems, in which case we only have to consider a more general initial datum than the finite rank $\gamma(0)$ in~\eqref{eq:finite_rank_initial}. In principle any non-negative operator $\gamma(0)$ can be considered, except that for fermions we should not forget the Pauli principle which requires that $0\leq\gamma(0)\leq1$.

Several authors have already studied the Hartree equation~\eqref{eq:rHF_DM} with $\gamma(t)$ an \emph{infinite-rank} operator. When $\gamma(0)$ is a trace-class operator with finite kinetic energy, $\tr(1-\Delta)\gamma(0)<\ii$, Bove, Da Prato and Fano~\cite{BovPraFan-74,BovPraFan-76} and Chadam~\cite{Chadam-76} have simultaneously proved the well-posedness of Equation~\eqref{eq:rHF_DM} (for the more precise Hartree-Fock model which also includes an exchange term). Later, Zagatti obtained the same result by a slightly different method~\cite{Zagatti-92}. 

Let us emphasize that, even if $\gamma(t)$ can have an infinite rank, the trace-class condition in these works means that the total average number of particles $\tr\,\gamma(t)$ is indeed finite for all times. In the present article, we consider an infinite system, for which the average number of particles is infinite:
$$\tr\, \gamma(t)=+\ii.$$
In the cases considered in this paper, $\gamma(t)$ is not even a compact operator and none of the methods used in the previous works~\cite{BovPraFan-74,BovPraFan-76,Chadam-76,Zagatti-92} is applicable.

The equation~\eqref{eq:rHF_DM} has many well-known stationary states with infinitely many particles and we discuss this now. Consider the operator
\begin{equation}
\gamma_f=f(-\Delta)
\label{eq:gamma_f_intro}
\end{equation}
with $f:\R^+\to\R^+$, which acts in the Fourier variable as a multiplication by the function $g(p):=f(|p|^2)$. The integral kernel of $\gamma_f$ is $\gamma_f(x,y)=(2\pi)^{-d/2}\check{g}(x-y)$ and hence the corresponding density is uniform in space:
$$\rho_{\gamma_f}(x)=(2\pi)^{-d/2}\check{g}(0)=(2\pi)^{-d} \int_{\R^d}f(|p|^2)\,dp,\qquad \forall x\in\R^d.$$
This integral is finite when $\int_{\R^d}|f(|p|^2)|\,dp<\ii$.
As a consequence, we find that the mean-field potential is also uniform,
$$w\ast\rho_{\gamma_f}(x) = (2\pi)^{-d} \left(\int_{\R^d}f(|p|^2)\,dp\right)\left(\int_{\R^d}w(x)\,dx\right),$$
and we get
$[-\Delta+w\ast\rho_{\gamma_f}\,,\, {\gamma_f}]=[-\Delta\,,\, f(-\Delta)]=0.$
We conclude that any function $f\geq0$ such that $\int_{\R^d}f(|p|^2)\,dp<\ii$ provides a stationary state $\gamma_f$ for the Hartree equation. Four important physical examples which are covered by our results, are the

\smallskip

\noindent$\bullet$ \emph{Fermi gas at zero temperature and chemical potential $\mu>0$}: 
\begin{equation}
f(r)=\1(0\leq r\leq\mu)\quad\text{and}\quad \gamma_f=\1(-\Delta\leq \mu);
\label{eq:Fermi-gas-zero-temp} 
\end{equation}

\medskip

\noindent$\bullet$ \emph{Fermi gas at positive temperature $T>0$ and chemical potential $\mu\in\R$}: 
\begin{equation}
f(r)=\frac{1}{e^{(r-\mu)/T}+1}\quad\text{and}\quad \gamma_f=\frac{1}{e^{(-\Delta-\mu)/T}+1};
\label{eq:Fermi-gas-positive-temp} 
\end{equation}

\medskip

\noindent$\bullet$ \emph{Bose gas at positive temperature $T>0$ and chemical potential $\mu<0$}: 
\begin{equation}
f(r)=\frac{1}{e^{(r-\mu)/T}-1}\quad\text{and}\quad \gamma_f=\frac{1}{e^{(-\Delta-\mu)/T}-1};
\label{eq:Bose-gas-positive-temp} 
\end{equation}

\medskip

\noindent$\bullet$ \emph{Boltzmann gas at positive temperature $T>0$ and chemical potential $\mu\in\R$}: 
\begin{equation}
f(r)=e^{-(r-\mu)/T}\quad\text{and}\quad \gamma_f=e^{(\Delta+\mu)/T}.
\label{eq:Boltzmann-gas-positive-temp} 
\end{equation}

Note that, in order to properly define the potential $w\ast \rho_{\gamma_f}$, we have used here both that $g\in L^1(\R^d)$ and that $w\in L^1(\R^d)$. In this paper, we always assume that $w$ is a (smooth enough) short range potential and we do not discuss possible extensions to, say, the NLS case $w=a\delta$.

The main purpose of this article is to prove the global well-posedness of the Hartree equation~\eqref{eq:rHF_DM}, for initial data $\gamma(0)$ which are `nice' perturbations of a reference stationary state $\gamma_f$, with $f$ as before. As usual for Hamiltonian PDEs, after having shown local well-posedness, we use conserved quantities in order to control the possible growth of $\gamma(t)$ in a suitable space along the flow and prove that the solution is global. We systematically work on the difference $\gamma(t)-\gamma_f$, which we think of as the time dependent perturbation of the reference state $\gamma_f$. So we have to control $\gamma(t)-\gamma_f$ in suitable norms.

The situation is much more complicated than for a finite system. For instance, the number of particle is infinite but one can (under suitable assumptions on $\gamma(0)$) give a meaning to the \emph{relative number of particles}
$\delta N(t)=\tr\big(\gamma(t)-\gamma_f\big).$
We can show that $\delta N(t)$ is conserved but, since $\gamma(t)-\gamma_f$ has no sign, this does not seem to yield any bound on $\gamma(t)-\gamma_f$.

The energy helps us more, when $\gamma_f$ is one of the particular Gibbs states in~\eqref{eq:Fermi-gas-zero-temp}--\eqref{eq:Boltzmann-gas-positive-temp}. In our approach, $\gamma(0)$ is assumed to have a finite relative entropy with respect to $\gamma_f$. We then prove that the relative free energy is conserved and use it as a Lyapunov function, under suitable assumptions on the potential $w$. This is the equivalent of the usual energy methods for finite systems, as we explain in detail in the rest of the paper. The main difference to the usual case is that we only have one conservation law.

Let us remark that if $\gamma_f\equiv0$ (vacuum case), then any unitarily invariant norm for $\gamma(t)$ is also conserved. This includes for instance Schatten norms $\norm{\gamma}_{\gS^p}=(\tr\,|\gamma|^p)^{1/p}$ ($p=2$ for Hilbert-Schmidt operators). These conservation laws can be used to prove well-posedness for~\eqref{eq:rHF_DM} with $\gamma_f=0$ in Schatten spaces with $p>1$ (see~\cite{Sabin-PhD}). In our case of $\gamma_f\neq0$, then $\gamma(t)-\gamma_f$ is not unitarily equivalent to $\gamma(0)-\gamma_f$ and there is no obvious conserved norm for $\gamma(t)-\gamma_f$. Controlling this operator can only be made through the relative energy, which is an additional difficulty as compared to finite systems.

\medskip

The article is organized as follows. In the next section we quickly state our main results for the above physical cases~\eqref{eq:Fermi-gas-zero-temp}--\eqref{eq:Boltzmann-gas-positive-temp}, even if our results cover more general situations. In Section~\ref{sec:local}, we construct local solutions in Schatten spaces and which have a sufficiently high regularity. Namely, we require that $(\gamma(t)-\gamma_f)(1-\Delta)^{s/2}\in \gS^p$ where $\gS^p$ is the $p$th Schatten space. By using the conservation of the relative Hartree energy and Lieb-Thirring inequalities from~\cite{FraLewLieSei-11,FraLewLieSei-12}, we then completely solve the zero-temperature case~\eqref{eq:Fermi-gas-zero-temp} in Section~\ref{sec:global_zero_temp}, in dimensions $d\geq2$. 

The rest of the paper is devoted to the positive temperature case. First, in Section~\ref{sec:Strichartz}, we give another local well-posedness result, useful for gases at positive temperature, and which is based on a recent Strichartz inequality in Schatten spaces from~\cite{FraLewLieSei-13}. Then, in Section~\ref{sec:entropy} we define the relative entropy of two density matrices following~\cite{LewSab-13} and prove Lieb-Thirring inequalities in the spirit of~\cite{FraLewLieSei-11,FraLewLieSei-12}. We are able to deal with much more than the three examples~\eqref{eq:Fermi-gas-positive-temp}--\eqref{eq:Boltzmann-gas-positive-temp}. Finally, in Section~\ref{sec:global_positive_temp} we construct solutions around many stationary states $\gamma_f=f(-\Delta)$ in dimensions $d=1,2,3$, which include~\eqref{eq:Fermi-gas-positive-temp}--\eqref{eq:Boltzmann-gas-positive-temp} among others.

\section{Main results for Bose and Fermi gases}

In this section, we quickly state our main results for \eqref{eq:Fermi-gas-zero-temp}--\eqref{eq:Boltzmann-gas-positive-temp}, which are the main cases of physical interest. We refer the reader to the appropriate sections for other results.

\subsection{Fermi gas at zero temperature}\label{sec:enonce_zero_temp}

We start by discussing the case of fermions at zero temperature~\eqref{eq:Fermi-gas-zero-temp}, that is, we take
\begin{equation}
\gamma_f=\1(-\Delta \leq\mu):=\Pi_\mu^-.
\label{eq:def_Pi_mu}
\end{equation}
This state is the formal minimizer of the energy 
$$\gamma\mapsto \tr(-\Delta-\mu)\gamma,$$
among all $0\leq\gamma\leq1$, but the energy is however infinite. The corresponding Hartree energy
$$\gamma\mapsto \tr(-\Delta-\mu)\gamma+\frac12 \int_{\R^d}\int_{\R^d}w(x-y)\rho_\gamma(x)\rho_\gamma(y)\,dx\,dy$$
is also infinite. The correct way to investigate the stability of $\Pi_\mu^-$ is to introduce the \emph{relative Hartree energy}, which is the formal difference between the energy of $\gamma$  and that of $\Pi_\mu^-$ (which are both infinite). In the stationary case, this technique goes back to~\cite{HaiLewSer-05a,HaiLewSol-07}. The relative Hartree energy is properly defined in Section~\ref{sec:global_zero_temp} and it is formally given by 
\begin{multline}
\cE(\gamma,\Pi_\mu^-):=\tr(-\Delta-\mu)(\gamma-\Pi_\mu^-)\\+\frac12 \int_{\R^d}\int_{\R^d}w(x-y)\rho_{\gamma-\Pi_\mu^-}(x)\rho_{\gamma-\Pi_\mu^-}(y)\,dx\,dy.
\label{eq:formal_relative_energy} 
\end{multline}
The \emph{relative kinetic energy} is the first term and it is non-negative, since $\Pi_\mu^-$ formally minimizes $\gamma\mapsto \tr(-\Delta-\mu)\gamma$. It is therefore natural to impose the condition that
$$0\leq \tr(-\Delta-\mu)(\gamma-\Pi_\mu^-)<\ii.$$
The trace has to be interpreted in an appropriate way, discussed at length in~\cite{FraLewLieSei-12} and in Section~\ref{sec:global_zero_temp}.

We emphasize that the function $\rho_{\gamma-\Pi_\mu^-}$ has no sign \emph{a priori}. Therefore the sign of the nonlinear term appearing in the relative energy~\eqref{eq:formal_relative_energy}
\begin{equation}
\int_{\R^d}\int_{\R^d}w(x-y)\rho_{\gamma-\Pi_\mu^-}(x)\rho_{\gamma-\Pi_\mu^-}(y)\,dx\,dy=(2\pi)^{d/2}\int_{\R^d}\hat{w}(k)|\hat{\rho_{\gamma-\Pi_\mu^-}}(k)|^2\,dk
\label{eq:focusing_vs_defocusing} 
\end{equation}
is determined by the sign of $\widehat{w}$. In the whole paper we therefore call \emph{defocusing} the case $\widehat{w}\geq0$ and \emph{focusing} the case $\widehat{w}\leq0$.

Our main result, proved in Section~\ref{sec:global_zero_temp} below, is the following.

\begin{theorem}[Global solutions at zero temperature]
Let $d\geq2$, $\mu>0$ and $\Pi_\mu^-=\1(-\Delta\leq\mu)$. Assume that $w\in L^1(\R^d)\cap L^\ii(\R^d)$ with $w(x)=w(-x)$ for a.e. $x\in\R^d$ is such that $\widehat{w}\geq0$ or, in dimension $d=2$, that $\widehat{w}\geq -\epsilon$ for $\epsilon>0$ small enough.

Then, for any $0\leq\gamma_0\leq1$ with finite kinetic energy relative to $\Pi_\mu^-$, 
$$0\leq \tr(-\Delta-\mu)(\gamma_0-\Pi_\mu^-)<\ii,$$
the Hartree equation~\eqref{eq:rHF_DM} admits a unique global solution $\gamma(t)$ such that $0\leq \tr(-\Delta-\mu)(\gamma(t)-\Pi_\mu^-)<\ii$ for all times $t\in\R$. 
Furthermore, the relative Hartree energy is finite, non-negative and conserved:
$$\cE(\gamma(t),\Pi_\mu^-)=\cE(\gamma(0),\Pi_\mu^-),\qquad\forall t\in\R.$$
\end{theorem}

Our statement is rather vague with regards to the function space in which the solution $\gamma(t)$ lives. We refer the reader to Theorem~\ref{thm:global_zero_temp} for a more precise statement.

The fact that we are unable to deal with dimension $d=1$ may sound surprising, but this is due to the lack of a Lieb-Thirring inequality in this case~\cite{FraLewLieSei-12}, which is itself related to the Peierls instability of one-dimensional quantum systems~\cite[Sec. 4.3]{Peierls}.

Our assumption that $w\in L^1(\R^d)\cap L^\ii(\R^d)$ is probably far from optimal, we have not tried to optimize it. Similarly, we have only stated the result for the ``\emph{defocusing case}'' $\widehat{w}\geq0$ corresponding to having a non-negative potential energy in the relative Hartree energy~\eqref{eq:formal_relative_energy}. We are indeed able to prove the local existence of solutions without this additional assumption, but then global well-posedness is not known.

\subsection{Bose and Fermi gases at positive temperature}
We have a similar result at positive temperature, which however requires different assumptions on the interaction potential $w$, and covers the one-dimensional case.

This time, we take
\begin{equation}
\gamma^{\rm fer}_{T,\mu}=\frac{1}{e^{(-\Delta-\mu)/T}+1}
\label{eq:Fermi-gas-positive-temp_gamma_f} 
\end{equation}
with $\mu\in\R$ and $T>0$ for fermions, 
\begin{equation}
\gamma^{\rm bos}_{T,\mu}=\frac{1}{e^{(-\Delta-\mu)/T}-1}
\label{eq:Bose-gas-positive-temp_gamma_f} 
\end{equation}
with $\mu<0$ and $T>0$ for bosons, or
\begin{equation}
\gamma^{\rm bol}_{T,\mu}=e^{(\Delta+\mu)/T}
\label{eq:Boltzmann-gas-positive-temp_gamma_f} 
\end{equation}
with $\mu\in\R$ and $T>0$ for ``boltzons''. The relative Hartree energy which we have introduced in~\eqref{eq:formal_relative_energy} is again formally conserved, but it is not the appropriate functional at positive temperature because it does not control the growth of any norm of $\gamma-\gamma_{T,\mu}$: the term $\tr(-\Delta-\mu)(\gamma-\gamma_{T,\mu})$ has no sign \emph{a priori}. It is more convenient to use the \emph{relative free energy} which contains in addition the difference of the entropies, and is defined as
\begin{multline}
\cF(\gamma,\gamma_{T,\mu}):=\cH(\gamma,\gamma_{T,\mu})\\+\frac12 \int_{\R^d}\int_{\R^d}w(x-y)\rho_{\gamma-\gamma_{T,\mu}}(x)\rho_{\gamma-\gamma_{T,\mu}}(y)\,dx\,dy
\label{eq:formal_relative_free_energy} 
\end{multline}
where 
\begin{equation}
\cH(A,B)=-\tr\Big(S(A)-S(B)-S'(B)(A-B)\Big)\geq0
\label{eq:def_relative_entropy_1}
\end{equation}
is the \emph{relative entropy}, with
\begin{equation*}
S(x)=T\times 
\begin{cases}
-x\log(x)-(1-x)\log(1-x) &\text{for fermions,}\\
-x\log(x)+(1+x)\log(1+x)&\text{for bosons,}\\
-x\log(x)+x&\text{for ``boltzons''.}
\end{cases}
\end{equation*}
The relative entropy $\cH$ is properly defined and studied at length in~\cite{LewSab-13} where we even considered a general concave function $S$, leading to many other stationary states $\gamma_f$ for the Hartree equation. We do not discuss this here for shortness. For the previous gases, our main result is Theorem \ref{thm:global_positive_temp} in Section~\ref{sec:global_positive_temp}, for which we state a simplified version:

\begin{theorem}[Global solutions at positive temperature]\label{thm:global_positive_temp_intro}
Let $d\in\{1,2,3\}$, $T>0$ and $\gamma_{T,\mu}$ be given by~\eqref{eq:Fermi-gas-positive-temp_gamma_f} or~\eqref{eq:Boltzmann-gas-positive-temp_gamma_f} with $\mu\in\R$, or by~\eqref{eq:Bose-gas-positive-temp_gamma_f} with $\mu<0$. Fix $M=1$ for fermions, and $M\geq\max(1,\|\gamma_{T,\mu}\|)$ in the two other cases. Assume that $w\in L^1(\R^d)\cap L^\ii(\R^d)$ and that
\begin{itemize}
\item $\widehat{w}\geq0$ and $\nabla w\in L^1(\R^3)\cap L^\ii (\R^3)$ in dimension $d=3$;
\item $\widehat{w}\geq -\kappa_2$ in dimension $d=2$;
\item $\widehat{w}\geq -\sqrt{T}\,\kappa_1$ in dimension $d=1$,
\end{itemize}
where $\kappa_1$ and $\kappa_2$ are constants which only depend on $M$ and $\mu/T$.

Then, for any initial datum $0\leq\gamma_0\leq M$ with finite entropy relative to $\gamma_{T,\mu}$,
$$0\leq \cH(\gamma_0,\gamma_{T,\mu})<\ii,$$
the Hartree equation~\eqref{eq:rHF_DM} admits a unique global solution $\gamma(t)$ such that 
$0\leq \cH(\gamma(t),\gamma_{T,\mu})<\ii$
for all times $t\in\R$. Furthermore, the relative free energy $\cF(\gamma(t),\gamma_{T,\mu})$ is finite, non-negative and conserved.
\end{theorem}

As compared to the zero temperature case, we are able to deal with the one-dimensional case. The constants $\kappa_1$ and $\kappa_2$ appearing in the statement are related to the best constants in the Lieb-Thirring inequality, and they behave differently depending on the statistics. They stay uniformly bounded when $\mu/T$ varies in a compact set for fermions and boltzons, but they tend to 0 for bosons when $\mu/T\to 0^-$. In dimension $d=1$ we gain an additional factor $\sqrt{T}$ and we can therefore replace the smallness assumption on the negative part of $\widehat{w}$, by the requirement that the temperature is large enough and that $\mu/T$ stays bounded (away from $0$ in the bosonic case).

The dimensional restriction $d\le3$ comes from the difficulty to construct local solutions in high dimension. Technically, this is related to a lack of information about the high-momentum decay in Schatten spaces of an operator with a finite relative entropy. Klein's inequality $\cH(\gamma,\gamma_{T,\mu})\gtrsim\tr(1-\Delta)(\gamma-\gamma_{T,\mu})^2$ is the best information that we have (see~\eqref{eq:Klein2}). By using only this information we need $d\le3$ in our local well-posedness result in Section \ref{sec:Strichartz}. In the zero-temperature case we have the additional information that the operator $(1-\Pi_{2\mu}^-)(\gamma-\Pi_\mu^-)(1-\Pi_{2\mu}^-)$ is \emph{trace-class}, which is used to deal with $d\geq4$. 

That the (free) energy is conserved and positive can be used to prove that the solution does not escape far from the stationary states $\Pi_\mu^-$ and $\gamma^{\rm fer/bos/bol}_{T,\mu}$, which is usually called \emph{orbital stability}. It is an interesting question to investigate the \emph{asymptotic stability} of these states, that is, the weak limit of $\gamma(t)$ when $t\to\ii$. This is studied in~\cite{LewSab-13b}.

\subsection*{Notation}
In the whole paper, we denote by $\cB(\gH)$ the space of bounded operators on the Hilbert space $\gH$, with corresponding operator norm $\|A\|$. We use the notation $\gS^p(\gH)$ for the Schatten space of all the compact operators $A$ on $\gH$ such that $\tr|A|^p<\ii$, with $|A|=\sqrt{A^*A}$, and use the norm 
\begin{equation}
\norm{A}_{\gS^p(\gH)}:=(\tr|A|^p)^{1/p}.
\label{eq:Schatten}
\end{equation}
We refer to~\cite{Simon-79} for the properties of Schatten spaces. The spaces $\gS^2(\gH)$ and $\gS^1(\gH)$ correspond to Hilbert-Schmidt and trace-class operators. We often use the shorthand notation $\cB$ and $\gS^p$ when the Hilbert space $\gH$ is clear from the context. We also denote by
$$\Pi_\mu^-:=\1(-\Delta\leq\mu)\quad\text{and}\quad \Pi_\mu^+:=\1(-\Delta\geq\mu)$$
the spectral projectors of the Laplacian, which are multiplication operators in Fourier space by the functions $k\mapsto \1(k^2\leq \mu)$ and $k\mapsto \1(k^2\geq\mu)$.
Some other spaces and functionals used in the text are summarized in Table~\ref{tab:notation} in the end of the article.

\section{Local well-posedness in Schatten spaces with high regularity}\label{sec:local}

In this section we state and prove the local well-posedness of the Hartree equation \eqref{eq:rHF_DM}, assuming that the initial datum $\gamma(0)$ is a very smooth perturbation of the reference state $\gamma_f=f(-\Delta)$, with $f$ a general function. Our main theorem uses a fixed point method on the Duhamel formulation of~\eqref{eq:rHF_DM} in Schatten spaces and it is based on estimates on the density $\rho_{\gamma(t)}$ which are \emph{pointwise} in time. A more involved local well-posedness result based on Strichartz inequalities is later in Section~\ref{sec:Strichartz}. 

Let us introduce the new variable 
$$\boxed{Q(t):=\gamma(t)-\gamma_f}$$ which represents the variation with respect to the stationary state $\gamma_f$. The equation \eqref{eq:rHF_DM} can be equivalently rewritten as 
\begin{equation}\label{eq:rHF_Q}
 i\partial_t Q=[-\Delta+w*\rho_Q,\gamma_f+Q].
\end{equation}
Here we have used the fact that $w*\rho_{\gamma_f}$ is a constant function which, as an operator, commutes with $\gamma=\gamma_f+Q$. The formulation \eqref{eq:rHF_Q} is understood in a weak sense: if $Q\in C^0(I,\cB(L^2(\R^d)))$ is such that $w*\rho_Q\in L^1_{\rm loc}(I,L^\ii(\R^d)))$ for some time interval $I\ni0$ (here we assume that $w*\rho_Q$ is well-defined), then $Q$ is a solution to \eqref{eq:rHF_Q} if and only if for all $f,g\in H^2(\R^d)$ we have the identity
$$
  i\partial_t\langle f,Q(t)g\rangle=\langle(-\Delta)f,Q(t)g\rangle-\langle Q(t)f,(-\Delta)g\rangle+\langle f,[w*\rho_Q(t),\gamma_f+Q(t)]g\rangle
$$
in the sense of distributions on $I$. As usual for nonlinear evolution equations, this is equivalent to the Duhamel formulation
\begin{equation}\label{eq:rHF_Q_Duhamel}
    Q(t)=e^{it\Delta}Q_0e^{-it\Delta}-i\int_0^te^{i(t-t')\Delta}[w*\rho_Q(t'),\gamma_f+Q(t')]e^{i(t'-t)\Delta}\,dt',
\end{equation}
for all $t\in I$. In this section we use a fixed point procedure to solve~\eqref{eq:rHF_Q_Duhamel}. The main difficulty is that we have to deal with operators. Even if the initial datum $Q_0$ is a finite rank operator, then the rank of $Q(t)$ is not preserved along the flow, due to the \emph{linear response} $[w*\rho_Q(t'),\gamma_f]$ which is never finite-rank. So, we have to work in a bigger space. For the usual Hartree equation~\eqref{eq:Hartree} one can play with regularity and construct solutions in arbitrary Sobolev spaces (under suitable assumptions on $w$). Here we have two natural parameters at our disposal: $s$ which determines the Sobolev regularity on the one hand, and the exponent $p$ of the Schatten class in which $Q(t)$ lives on the other hand. This leads to the following Sobolev-like Schatten space
\begin{multline}
\gS^{p,s}=\bigg\{Q=Q^*\in \cB(L^2(\R^d))\ :\\ \norm{Q(1-\Delta)^{s/2}}_{\gS^p}^p=\tr\big((1-\Delta)^{s/2}Q^2(1-\Delta)^{s/2}\big)^{p/2}<\ii\bigg\}
\label{eq:def_Schatten-Sobolev}
\end{multline}
which, for $s\in\R$ and $p\geq1$, is a Banach space when endowed with the corresponding norm $\norm{Q}_{\gS^{p,s}}:=\norm{Q(1-\Delta)^{s/2}}_{\gS^p}$. The spaces $\gS^{p,s}$ satisfy the inclusion relation
\begin{equation}
\gS^{p,s}\subset \gS^{q,r}\qquad\text{for $s\geq r$ and $p\leq q$.}
\end{equation}
Finite rank operators with smooth eigenfunctions are dense in $\gS^{p,s}$ for all $1\leq p<\ii$ and all $s\geq0$.

Another difficulty of our equation is the proper definition of the density $\rho_Q$ of $Q$. Formally, $\rho_Q(x)=Q(x,x)$ but this only makes sense when the kernel $Q(x,y)$ of $Q$ exists and is continuous. Another, better, definition is based on the spectral decomposition of $Q=\sum_{j}n_j|u_j\rangle\langle u_j|$, leading to $\rho_Q=\sum_j n_j|u_j|^2$. This is particularly appropriate for trace-class operators, that is, when $\sum_j |n_j|=\norm{Q}_{\gS^1}<\ii$, in which case $\rho_Q\in L^1(\R^d)$. 
For higher Schatten spaces, we have two strategies at our disposal to deal with $\rho_Q$. We can define the density $\rho_{Q(t)}$ in a weak sense in time, using the time integral in the Duhamel formula~\eqref{eq:rHF_Q_Duhamel} and Strichartz inequalities. This is discussed later in Section~\ref{sec:Strichartz}. The other strategy which is developed in this section is much simpler. It consists in assuming that $s$ is large enough, which yields a well-defined density for all times, by a variant of the Sobolev inequality for operators.

\begin{lemma}[Density $\rho_Q$ for $Q$ in $\gS^{p,s}$]\label{lem:density_Schatten-Sobolev}
Let $d\geq1$, $p\geq1$ and $s>d(p-1)/p$  or $s\geq0$ if $p=1$. Then any $Q\in \gS^{p,s}$ is locally trace-class. Its density $\rho_Q$ belongs to $L^q(\R^d)$ and satisfies
\begin{equation}
\norm{\rho_Q}_{L^q(\R^d)}\leq C\norm{Q}_{\gS^{p,s}}, \qquad \forall Q\in\gS^{p,s},
\label{eq:estim_density_Schatten-Sobolev} 
\end{equation}
for all $p\leq q<d/(d-s)$ if $s\leq d$ and for $p=q=1$ if $s=0$. 

If $s>d$, then $\rho_Q\in L^p(\R^d)\cap L^\ii(\R^d)$ and
\begin{equation}
\norm{\rho_Q}_{L^p(\R^d)}+\norm{\rho_Q}_{L^\ii(\R^d)}\leq C\norm{Q}_{\gS^{p,s}}, \qquad \forall Q\in\gS^{p,s}.
\label{eq:estim_density_Schatten-Sobolev_Linfty} 
\end{equation}
The constants depend only on $d$, $p$, $q$ and $s$.
\end{lemma}

\begin{proof}
The result is well-known for $p=1$ and $s=0$ (see, e.g., the appendix of \cite{GriHan-12}). Furthermore, since $\gS^{p,s}\subset \gS^{q,s}$ for all $q\geq p$, we only have to prove the inequality~\eqref{eq:estim_density_Schatten-Sobolev} for $q=p$ and for $q=\infty$ if $s>d$.

We compute, for $Q\in\gS^{p,s}$ and a localization function $\chi\in L^\ii_c(\R^d)$
\begin{equation*}
\chi Q\chi =\chi (1-\Delta)^{-s/4}(1-\Delta)^{s/4}Q(1-\Delta)^{s/4}(1-\Delta)^{-s/4}\chi.
\end{equation*}
The Kato-Seiler-Simon (KSS) inequality (see~\cite{SeiSim-75} and~\cite[Thm 4.1]{Simon-79})
\begin{equation}
\norm{f(x)\, g(-i\nabla)}_{\gS^r}\leq (2\pi)^{-d/r} \norm{f}_{L^r(\R^d)}\norm{g}_{L^r(\R^d)},\qquad\forall r\geq2,
\label{eq:KSS}
\end{equation}
tells us that 
$$\norm{\chi (1-\Delta)^{-s/4}}_{\gS^{2p'}}\leq (2\pi)^{-d/p'}\norm{\chi}_{L^{2p'}(\R^d)}\left(\int_{\R^d}\frac{dk}{(1+|k|^2)^{sp'/2}}\right)^{\frac{1}{2p'}}$$
where the integral on the right side is finite for $s>d/p'=d(p-1)/p$. By H\"older's inequality in Schatten spaces, this proves that 
$$\norm{\chi Q\chi}_{\gS^1}\leq C\norm{\chi}_{L^{2p'}(\R^d)}^2\norm{(1-\Delta)^{s/4}|Q|^{1/2}}_{\gS^{2p}}^2.$$
On the other hand, by the Araki-Lieb-Thirring inequality for operators (see~\cite[Thm 9]{LieThi-76} and~\cite[Thm 1]{Araki-90}), we have 
$$\norm{(1-\Delta)^{s/4}|Q|^{1/2}}_{\gS^{2p}}^2\leq \norm{(1-\Delta)^{s/2}Q}_{\gS^{p}}=\norm{Q}_{\gS^{p,s}}.$$
We therefore obtain that $\chi Q\chi$ is trace-class, for any $\chi\in L^\ii_c(\R^d)$, hence that $Q$ is locally trace-class. In particular, $\rho_Q(x)$ is well defined in $L^1_{\rm loc}(\R^d)$.

In order to prove that $\rho_Q\in L^p(\R^d)$, we argue by duality, starting with a finite rank operator $Q$ and a potential $V\in L^\ii_c(\R^d)$. We find
\begin{multline*}
\int_{\R^d}\rho_QV=\tr(QV)=\tr\left((1-\Delta)^{\frac s4}Q(1-\Delta)^{\frac s4}(1-\Delta)^{-\frac s4}V(1-\Delta)^{-\frac s4}\right)\\
\leq \norm{Q}_{\gS^{p,s}}\norm{(1-\Delta)^{-\frac s4}|V(x)|^{\frac12}}_{\gS^{2p'}}^2\leq C\norm{Q}_{\gS^{p,s}}\norm{V}_{L^p{'}(\R^d)}.
\end{multline*}
By density of finite rank operators, the final estimate stays true for all $Q\in\gS^{p,s}$ and all $V\in L^\ii_c(\R^d)$. By duality, this finally proves that $\rho_Q\in L^p(\R^d)$ and that $\norm{\rho_Q}_{L^p(\R^d)}\leq C\norm{Q}_{\gS^{p,s}}$.

If $s>d$, the exact same proof shows that
\begin{equation*}
\int_{\R^d}\rho_QV \leq \norm{(1-\Delta)^{\frac{s}4}Q(1-\Delta)^{\frac{s}4}}\norm{(1-\Delta)^{-\frac{s}4}|V|^{\frac12}}_{\gS^2}^2\leq C\norm{Q}_{\gS^{p,s}}\norm{V}_{L^1}
\end{equation*}
which gives the estimate on $\norm{\rho_Q}_{L^\ii}$.
\end{proof}

With Lemma~\ref{lem:density_Schatten-Sobolev} at hand, it makes sense to look for solutions of~\eqref{eq:rHF_Q_Duhamel} satisfying $Q(t)\in \gS^{p,s}$ for all times, since $\rho_Q$ is well defined in this case. 

\begin{theorem}[Well-posedness in $\gS^{p,s}$]\label{thm:local-wp-Sps}
Let $d\geq1$, $1\leq p<\ii$ and $s>d(p-1)/p$  or $s\geq0$ if $p=1$. Assume that 
\begin{equation}
w\in W^{s,p'}(\R^d)
\label{eq:cond_w} 
\end{equation}
and that 
\begin{equation}
\dps\int_{\R^d}(1+k^2)^{\tfrac{n p}2}|f(k^2)|^p\,dk<\ii\quad\text{and}\quad w\in W^{n,1}(\R^d) \quad\text{if $p\geq2$,}
\label{eq:cond_f_w_1} 
\end{equation}
or that 
\begin{multline}
\sum_{z\in\Z^d}\left(\int_{C_z}(1+k^2)^{n}|f(k^2)|^2\,dk\right)^{p/2}<\ii\\ \quad\text{and}\quad \sum_{z\in\Z^d}\norm{w}_{W^{n,\frac{2p}{3p-2}}(C_z)}<\ii\qquad \text{if $1\leq p<2$,}
\label{eq:cond_f_w_2}
\end{multline}
where $n:=\lceil s\rceil$ is the smallest integer $n\geq s$ and $C_z=z+[-1/2,1/2)^d$.

Then, for any initial datum $Q_0\in\gS^{p,s}$ there exists a \emph{unique maximal solution} $Q(t)\in C^0_t((-T^-,T^+);\gS^{p,s})$ of~\eqref{eq:rHF_Q_Duhamel} with $T^\pm>0$. Furthermore, we have the blow-up criterion
\begin{equation}
T^\pm<\ii \Longrightarrow \lim_{t\to \pm T^\pm}\norm{Q(t)}_{\gS^{p,s}}=+\ii. 
\label{eq:blowup-criterion}
\end{equation}

Finally, the unique solution $Q(t)$ depends continuously on $Q_0\in\gS^{p,s}$, and on $w$ and $f$: if $Q_0^{(n)}\to Q_0$ strongly in $\gS^{p,s}$, $w_n\to w$ and $f_n\to f$ for the norms corresponding to \eqref{eq:cond_w}, \eqref{eq:cond_f_w_1} and \eqref{eq:cond_f_w_2}, then the unique solution $Q^{(n)}(t)$ has maximal times of existence $T_n^\pm$ such that $\liminf_{n\to\ii} T_n^\pm\geq T^\pm$ and $Q^{(n)}(t)\to Q(t)$ in $C_t^0([-T^-+\epsilon,T^+-\epsilon],\gS^{p,s})$ for every $\epsilon>0$.
\end{theorem}

Our result requires some stringent conditions~\eqref{eq:cond_w}, \eqref{eq:cond_f_w_1} and~\eqref{eq:cond_f_w_2} on $f$ and $w$, which are not optimal and which we have not tried to optimize. They are satisfied if for instance $f$ is bounded and decays fast enough, and if $w$ is in the Schwartz class. They are also valid in the four physical situations~\eqref{eq:Fermi-gas-zero-temp}--\eqref{eq:Boltzmann-gas-positive-temp} which we have in mind, since $f$ has a compact support or is exponentially decreasing. We note that when $1\leq p<2$, the conditions in~\eqref{eq:cond_f_w_1} are stronger than the ones similar to~\eqref{eq:cond_f_w_2}. Also, we remark that \eqref{eq:cond_f_w_1} and~\eqref{eq:cond_f_w_2} together with $s>d(p-1)/p$ imply that $\int_{\R^d}|f(k^2)|\,dk<\ii$ which is needed to see that $\rho_{\gamma_f}$ is a well-defined constant.

A natural situation, motivated by the positive temperature case (see Klein's inequality~\eqref{eq:Klein2} in Section~\ref{sec:entropy}), is the case of Hilbert-Schmidt perturbations with a Sobolev regularity $s=1$, that is of the space $\gS^{2,1}$. This situation is covered by the previous result only in dimension $d=1$. By using Strichartz estimates, we will construct solutions in $\gS^{2,1}$ in dimension $d=1,2,3$ later in Section~\ref{sec:Strichartz}.

Before turning to the proof of Theorem~\ref{thm:local-wp-Sps}, we mention an additional property of the solutions in $\gS^{p,s}$. 
Let $q\geq p$, $s\geq r>d(q-1)/q$, and $Q_0\in\gS^{p,s}\subset\gS^{q,r}$. We assume that $w,f$ satisfy~\eqref{eq:cond_w}, \eqref{eq:cond_f_w_1} and \eqref{eq:cond_f_w_2}. Then, the unique solution $Q$ to \eqref{eq:rHF_Q_Duhamel} has \emph{a priori} distinct maximal times of existence $T^\pm_{p,s}$ and $T^\pm_{q,r}$, corresponding to the two spaces $\gS^{p,s}$ and $\gS^{q,r}$. Since $\gS^{p,s}\subset \gS^{q,r}$, they obviously satisfy $T^\pm_{p,s}\le T^\pm_{q,r}$. The next result, which is sometimes called \emph{persistence of regularity} in the literature, shows that we have $T^\pm_{p,s}= T^\pm_{q,r}$. It is a simple corollary of the proof of Theorem~\ref{thm:local-wp-Sps}.

\begin{corollary}[Persistence of regularity]\label{cor:persistance-regularity}
Let $d\geq1$, $1\leq p\leq q<\ii$ and $s\geq r>d(q-1)/q$  or $s\geq r\geq0$ if $p=q=1$. Assume that $w$ and $f$ satisfy the same assumptions~\eqref{eq:cond_w}, \eqref{eq:cond_f_w_1} and \eqref{eq:cond_f_w_2} as in Theorem~\ref{thm:local-wp-Sps}. Then, for any $T_1,T_2>0$, if $Q\in C^0_t([-T_1,T_2],\gS^{p,s})$ is a solution to \eqref{eq:rHF_Q_Duhamel}, we have the following estimate
  \begin{equation}\label{eq:est-persistance-regularity}
    \|Q(t)\|_{\gS^{p,s}}\le \|Q_0\|_{\gS^{p,s}}e^{Ct\left(1+\sup_{t'\in[-T_1,T_2]}\|Q(t')\|_{\gS^{q,r}}\right)},\quad\forall t\in[-T_1,T_2].
  \end{equation}
Hence, $T^\pm_{p,s} = T^\pm_{q,r}$ for all $1\leq p\leq q<\ii$ and all $s\geq r>d(q-1)/q$. 
\end{corollary}

In the rest of the section we prove Theorem~\ref{thm:local-wp-Sps} and Corollary~\ref{cor:persistance-regularity}.

\begin{proof}[Proof of Theorem~\ref{thm:local-wp-Sps}]
We use a simple Banach-Picard fixed point theorem in a ball centered at $0$ in the space $\gS^{p,s}$. To this end, we need two estimates which are stated in the next lemmas.

 \begin{lemma}\label{lemma:commutator-Vgamma-Hs}
Let $d\geq1$, $1\leq p<\ii$ and $s>d(p-1)/p$  or $s\geq0$ if $p=1$. Assume that $f$ and $w$ satisfy~\eqref{eq:cond_f_w_1} and~\eqref{eq:cond_f_w_2}. Then we have $[\rho_Q\ast w,\gamma_f]\in\gS^{p,s}$ for every $Q\in\gS^{p,s}$ and
  $$\norm{[\rho_Q\ast w,\gamma_f]}_{\gS^{p,s}}\le C\|Q\|_{\gS^{p,s}},$$
with a constant $C$ depending only on $d$, $p$, $s$, $w$ and $f$. 
\end{lemma}
 
\begin{proof}
Denote $V=\rho_Q\ast w$. We first assume $p\geq2$. By the KSS inequality~\eqref{eq:KSS}, we have 
$$\|V\gamma_f(1-\Delta)^{s/2}\|_{\gS_p}^2\leq (2\pi)^{-d/p}\|V\|_{L^p(\R^d)}\left(\int_{\R^d}(1+k^2)^{sp/2}|f(k^2)|^p\,dk\right)^{1/p}.$$
For the other term, let $n=\lceil s\rceil$ be the smallest integer $n\geq s$. Then 
$$\gamma_f V(1-\Delta)^{s/2}=\gamma_f V(1-\Delta)^{n/2}(1-\Delta)^{(s-n)/2}$$
and, therefore, $\|V\gamma_f(1-\Delta)^{s/2}\|_{\gS^p}\leq \|V\gamma_f(1-\Delta)^{n/2}\|_{\gS^p}$.
Since as before $\norm{\gamma_f V}_{\gS^p}\leq C\norm{V}_{L^p(\R^d)}$, we only have to estimate the $\gS^p$ norm of
$$\gamma_f V \partial_j^n=\sum_{k=0}^n\binom{n}{k}(-1)^{n-k}\left(\gamma_f\partial_j^k\right)\left(\partial_j^{n-k}V\right),$$
for all $j=1,...,d$. This is done by following again the previous argument, with $\norm{V}_{W^{n,p}(\R^d)}$ appearing on the right side. It remains to use that 
$$\norm{V}_{W^{n,p}(\R^d)}\leq \norm{w}_{W^{n,1}(\R^d)}\norm{\rho_Q}_{L^p(\R^d)}\leq C\norm{w}_{W^{n,1}(\R^d)}\|Q\|_{\gS^{p,s}}$$
by Lemma~\ref{lem:density_Schatten-Sobolev}.
When $1\leq p<2$, the argument is the same, except that this time we use the Birman--Solomjak inequality (see \cite[Th. 4.5]{Simon-79}) 
  \begin{equation}\label{eq:est-birman-solo}
   \|f(x)g(-i\nabla)\|_{\gS^p}\le C\|f\|_{\ell^pL^2}\|g\|_{\ell^p L^2}
  \end{equation}
for all $f,g\in\ell^p(L^2)$ where, following~\cite{Simon-79}, we have introduced the norm
\begin{equation}
\|f\|_{\ell^pL^2}^p:=\sum_{z\in\Z^d}\|f\|_{L^2(C_z)}^p.
\label{eq:def_ell_p_L_q}
\end{equation}
We recall that $C_z=z+[-1/2,1/2)^d$ is the unit cube centered at $z$.
Applying~\eqref{eq:est-birman-solo} we get the $\ell^p(L^2)$ norm of the derivatives of $V$ on the right. They can be estimated by $\norm{\rho_{Q}}_{L^p(\R^d)}$ using our assumption~\eqref{eq:cond_f_w_2} on $w$ and Young's inequality $\|f*g\|_{\ell^pL^2}\le C\|f\|_{\ell^1L^{\frac{2p}{3p-2}}}\|g\|_{L^p}$.
\end{proof}

  \begin{lemma}\label{lemma:commutator-VQ-Hs}
Let $d\geq1$, $1\leq p<\ii$ and $s>d(p-1)/p$  or $s\geq0$ if $p=1$. Assume that $w\in W^{s,p'}(\R^d)$. Then, for any $Q,Q'\in\gS^{p,s}$, we have $[\rho_Q\ast w,Q']\in\gS^{p,s}$ and 
  $$\norm{[\rho_Q\ast w,Q']}_{\gS^{p,s}}\le C\|Q\|_{\gS^{p,s}}\|Q'\|_{\gS^{p,s}}$$
with a constant $C$ depending only on $d$, $p$, $s$ and $w$. 
 \end{lemma}

 \begin{proof} Let $V=\rho_Q\ast w$.
  We have $\|VQ'(1-\Delta)^{s/2}\|_{\gS^p}\le \|V\|_{L^\ii}\|Q'\|_{\gS^{p,s}}$ and
$
\|Q'V(1-\Delta)^{\frac s2}\|_{\gS^p}\le \|Q'\|_{\gS^{p,s}}\|(1-\Delta)^{-\frac s2}V(1-\Delta)^{\frac s2}\|=\|Q'\|_{\gS^{p,s}}\|V\|_{H^s\to H^s}. 
$
By \cite[Thm  1.4]{GulKon-96}, one has
  \begin{equation}\label{eq:est-V-Hs-Hs}
    \|V\|_{H^s\to H^s}\le C\|V\|_{W^{s,\ii}}
  \end{equation}
and the result follows from Young's inequality and Lemma~\ref{lem:density_Schatten-Sobolev}.
 \end{proof}

Using Lemmas~\ref{lemma:commutator-Vgamma-Hs} and~\ref{lemma:commutator-VQ-Hs}, we can now finish the proof of Theorem~\ref{thm:local-wp-Sps}. Let $R>0$ and $Q_0\in\gS^{p,s}$ such that $\|Q_0\|_{\gS^{p,s}}\le R$. Note that the first term in Duhamel's formula has a norm independent of $t$: $\norm{e^{it\Delta}Q_0 e^{-it\Delta}}_{\gS^{p,s}}=\norm{Q_0}_{\gS^{p,s}}$. Let now $T=T(R)>0$ to be chosen later. Define the space $X_T:=C^0_t([-T,T],\gS^{p,s})$. It is a Banach space endowed with the norm $\|Q\|_{X_T}:=\sup_{t\in[-T,T]}\|Q(t)\|_{\gS^{p,s}}.$ For any $Q\in X_T$, we denote by $\Phi(Q)$ the function given by:
  \begin{equation}
   \Phi(Q)(t)=e^{it\Delta}Q_0e^{-it\Delta}-i\int_0^te^{i(t-t')\Delta}[w*\rho_Q(t'),\gamma_f+Q(t')]e^{i(t'-t)\Delta}\,dt'.
  \end{equation}
  From Lemmas \ref{lemma:commutator-Vgamma-Hs} and \ref{lemma:commutator-VQ-Hs}, we deduce that for all $Q\in X_T$, we have $\Phi(Q)\in X_T$ and the estimate 
  $\|\Phi(Q)\|_{X_T}\le R+CT\|Q\|_{X_T}(1+\|Q\|_{X_T})$.
  Furthermore, writing for $Q,Q'\in X_T$ and $t\in[-T,T]$ that
  \begin{multline}\label{eq:est-PhiQ-PhiQ'}
    \Phi(Q)(t)-\Phi(Q')(t)=-i\int_0^te^{i(t-t')\Delta}\left([w*(\rho_Q(t')-\rho_{Q'}(t')),\gamma_f+Q(t')]\right.\\
    \left.+[w*\rho_{Q'}(t'),Q(t')-Q'(t')]\right)e^{i(t'-t)\Delta}\,dt',
  \end{multline}
  we infer, using the same estimates, that
    $$\|\Phi(Q)-\Phi(Q')\|_{X_T}\le CT(1+\|Q\|_{X_T}+\|Q'\|_{X_T})\|Q-Q'\|_{X_T}.$$
  As a consequence, if $T$ is small enough such that $2CT(1+2R)\le1$ and $CT(1+4R)\le1/2$, we conclude that $\Phi$ is a contraction on the closed ball $B_{2R}:=\{Q\in X_T,\,\|Q\|_{X_T}\le 2R\}$, and hence admits a unique fixed point on $B_{2R}$ by the Banach--Picard theorem. 

Let us turn to the proof of uniqueness, independently of the size of the solutions (that is, whether they belong to $B_{2R}$ or not). Let $Q,Q'$ be two solutions to \eqref{eq:rHF_Q_Duhamel} in $C^0([-T_1,T_2],\gS^{p,s})$, for some $T_1,T_2>0$. Then, using the equation \eqref{eq:est-PhiQ-PhiQ'}, we have the refined estimate for all $t\in[-T_1,T_2]$,
  $$\|Q(t)-Q'(t)\|_{\gS^{p,s}}\le C(1+\|Q\|_{X_T}+\|Q'\|_{X_T})\int_0^t\|Q(t')-Q'(t')\|_{\gS^{p,s}}\,dt'.$$
This proves that $Q\equiv Q'$ on $[-T_1,T_2]$ by Gr\"onwall's lemma.  

By a standard argument, the local-in-time solution $Q(t)$ can be extended uniquely up to maximal times of existence $T^\pm$ on the left and on the right of $0$, as soon as $\norm{Q(t)}_{\gS^{p,s}}$ stays bounded. The continuity with respect to $f$, $w$ and $Q_0$ is also standard and the details are left to the reader.
\end{proof}

We end this section with the 

 \begin{proof}[Proof of Corollary \ref{cor:persistance-regularity}]
For simplicity, we assume $p\geq2$, the argument being similar for $p<2$.
From Lemma~\ref{lemma:commutator-Vgamma-Hs} and the proof of Lemma~\ref{lemma:commutator-VQ-Hs}, we have for all $t\in[-T_1,T_2]$
  \begin{equation*}
    \|Q(t)\|_{\gS^{p,s}}
\le\|Q_0\|_{\gS^{p,s}}\\+C\int_0^t\left(1+\|\rho_{Q(t')}\ast w\|_{W^{s,\ii}}\right)\|Q(t')\|_{\gS^{p,s}}\,dt'.
  \end{equation*}
Since $q\geq p$, we have $w\in W^{s,p'}\cap W^{n,1}\subset W^{s,q'}$, and therefore, by Young's inequality and Lemma~\ref{lem:density_Schatten-Sobolev},
$ \|\rho_{Q(t')}\ast w\|_{W^{s,\ii}}\leq C\|\rho_{Q(t')}\|_{L^q}\leq C\|Q(t')\|_{\gS^{q,r}}$.
We obtain 
$$\|Q(t)\|_{\gS^{p,s}}\leq \|Q_0\|_{\gS^{p,s}}+C\left(1+\sup_{t''\in[-T_1,T_2]}\|Q(t'')\|_{\gS^{q,r}}\right)\int_0^t\|Q(t')\|_{\gS^{p,s}}\,dt',$$
and the estimate \eqref{eq:est-persistance-regularity} follows from Gr\"onwall's lemma. The equality $T^\pm_{p,s}=T^\pm_{q,r}$ is a consequence of the blow-up criterion \eqref{eq:blowup-criterion} and of \eqref{eq:est-persistance-regularity}.
 \end{proof}

\section{Energy \& global well-posedness for the Fermi sea}\label{sec:global_zero_temp}

The purpose of this section is to prove the existence and uniqueness of \emph{global-in-time} solutions in the particular case of the Fermi gas at zero temperature and chemical potential $\mu>0$,
$$\boxed{\gamma_f=\1(-\Delta\leq\mu):=\Pi_\mu^-.}$$
Our approach is based on the results of the previous section, as well as on a notion of relative energy with respect to $\Pi_\mu^-$. Using a recent Lieb-Thirring inequality from~\cite{FraLewLieSei-11,FraLewLieSei-12}, we prove that the energy controls $Q(t)$ uniformly, leading to global solutions.

\medskip

Let $\mu>0$ and $\gamma_f=\Pi_\mu^-$. The \emph{relative Hartree energy} of any density matrix $0\leq \gamma\leq 1$ and the reference state $\Pi_\mu^-$, is the formal difference between the Hartree energy of $\gamma$ and that of $\Pi_\mu^-$, which are both infinite, leading to the formal expression
\begin{equation}\label{eq:def-relative-energy}
\cE(\gamma,\Pi_\mu^-):=\tr(-\Delta-\mu)(\gamma-\Pi_\mu^-)+\frac12\int_{\R^d\times\R^d}w(x-y)\rho_{\gamma-\Pi_\mu^-}(x)\rho_{\gamma-\Pi_\mu^-}(y)\,dxdy.
\end{equation}
A density matrix $0\leq\gamma\leq1$ is said to live in the energy space when the first (linear) term $\tr(-\Delta-\mu)(\gamma-\gamma_f)$ is finite, but it is necessary to interpret this term appropriately, following~\cite{HaiLewSer-05a} and~\cite{FraLewLieSei-12}. When $\gamma-\Pi_\mu^-$ is finite rank and smooth, then 
\begin{align}
0\leq& \tr(-\Delta-\mu)(\gamma-\Pi_\mu^-)\nonumber\\
&\qquad=\tr|-\Delta-\mu|^{\frac12}\Big(\Pi_\mu^+(\gamma-\Pi_\mu^-)\Pi_\mu^+-\Pi_\mu^-(\gamma-\Pi_\mu^-)\Pi_\mu^-\Big)|-\Delta-\mu|^{\frac12}\nonumber\\
&\qquad=\tr|-\Delta-\mu|^{\frac12}\Big(\Pi_\mu^+\gamma\Pi_\mu^++\Pi_\mu^-(1-\gamma)\Pi_\mu^-\Big)|-\Delta-\mu|^{\frac12}.\label{eq:def_trace}
\end{align}
The two formulas on the right side make sense for any $\gamma$ and they are taken as a definition of the trace $\tr(-\Delta-\mu)(\gamma-\Pi_\mu^-)$ in all cases. The set of all the density matrices which have a finite kinetic energy relative to $\Pi_\mu^-$ is therefore the convex set
\begin{multline}
 \cK_\mu:=\left\{0\le\gamma=\gamma^*\le1\ :\ |\Delta+\mu|^{\frac12}(\gamma-\Pi_\mu^-)^{\pm\pm}|\Delta+\mu|^{\frac12}\in\gS^1\right\}.
 \label{eq:def_K_mu}
\end{multline}
We have used the convenient notation
\begin{equation}\label{eq:def-++}
Q^{\sigma_1\sigma_2}=\Pi_\mu^{\sigma_1}Q\Pi_\mu^{\sigma_2}\quad\text{with}\quad \sigma_1,\sigma_2\in\{+,-\}.
\end{equation}
Note that the constraint $0\le\gamma\le1$ implies that 
$$(\gamma-\Pi_\mu^-)^{++}-(\gamma-\Pi_\mu^-)^{--}\ge(\gamma-\Pi_\mu^-)^2$$
and, hence,
\begin{multline*}
\norm{(\gamma-\Pi_\mu^-)|-\Delta-\mu|^{\frac12}}_{\gS^2}^2\\ \leq \tr|-\Delta-\mu|^{\frac12}\Big((\gamma-\Pi_\mu^-)^{++}-(\gamma-\Pi_\mu^-)^{--}\Big)|-\Delta-\mu|^{\frac12}<\ii.
\end{multline*}
It is therefore useful to introduce the following Banach space
\begin{multline}
 \cX_\mu:=\left\{Q=Q^*\in\cB(L^2(\R^d))\ :\ Q|\Delta+\mu|^{\frac12}\in\gS^2,\right.\\
  \left.|\Delta+\mu|^{\frac12}Q^{\pm\pm}|\Delta+\mu|^{\frac12}\in\gS^1\right\},
  \label{eq:def_X_mu}
\end{multline}
endowed with the norm 
\begin{multline*}
  \|Q\|_{\cX_\mu}:=\|Q\|+\left\|Q|\Delta+\mu|^{\frac12}\right\|_{\gS^2}+\left\||\Delta+\mu|^{\frac12}Q^{++}|\Delta+\mu|^{\frac12}\right\|_{\gS^1}\\
  +\left\| |\Delta+\mu|^{\frac12}Q^{--}|\Delta+\mu|^{\frac12}\right\|_{\gS^1},
\end{multline*}
where $Q^{\pm\pm}$ is defined in \eqref{eq:def-++}. We then have $\cK_\mu=\{0\le\gamma=\gamma^*\le1\ :\ \gamma-\gamma_f\in\cX_\mu\}$,
which realizes $\cK_\mu$ as a convex subset of the affine space $\gamma_f+\cX_\mu$. 

We emphasize that $|k^2-\mu|$ vanishes on the Fermi sphere $|k|=\sqrt{\mu}$. The fact that operators $Q$ in $\cX_\mu$ satisfy $Q|\Delta+\mu|^{1/2}\in\gS^2$ does \emph{not} imply that $Q$ is Hilbert-Schmidt. Indeed, $Q$ does not even have to be compact (but it is always bounded by assumption). In particular, $\cX_\mu$ is not included in any of the sets $\gS^{p,s}$ which we have introduced in the previous section. However, we have $\gS^{1,s}\subset \cX_\mu$ for all $s\geq2$, which will be used later.

A very important property of the relative kinetic energy is the following \emph{Lieb--Thirring inequality}, which was proved in~\cite[Thm  2.1]{FraLewLieSei-12}:
\begin{multline}\label{eq:Lieb-Thirring-zero-temp}
\tr(-\Delta-\mu)(\gamma-\Pi_\mu^-)
\ge K_{\rm LT}\!\int_{\R^d}\bigg(\big(\rho_{\Pi_\mu^-}+\rho_Q(x)\big)^{1+\tfrac{2}d}-(\rho_{\Pi_\mu^-})^{
1+\tfrac{2}d}\\-\frac{2+d}d(\rho_{\Pi_\mu^-})^{\tfrac{2}d}\,\rho_Q(x)\bigg)dx,
\end{multline}
for all $d\ge2$, $\mu>0$ and all $\gamma\in\cK_\mu$ with $Q:=\gamma-\Pi_\mu^-$. The constant $K_{\rm LT}$ only depends on the dimension $d\geq2$.
The inequality~\eqref{eq:Lieb-Thirring-zero-temp} is \emph{wrong} in dimension $d=1$, see~\cite{FraLewLieSei-12}. The inequality~\eqref{eq:Lieb-Thirring-zero-temp} implies that for any $\gamma\in\cK_\mu$, the density $\rho_{\gamma-\Pi_\mu^-}$ belongs to $L^2(\R^d)+L^{1+2/d}(\R^d)$. In particular, this shows that the nonlinear term in the relative Hartree energy \eqref{eq:def-relative-energy} is well-defined for any $\gamma\in\cK_\mu$ if $d\ge2$ and 
$w\in L^1(\R^d)\cap L^{(d+2)/4}(\R^d)$,
in which case we have $w*\rho_{\gamma-\Pi_\mu^-}\in L^2(\R^d)\cap L^{1+d/2}(\R^d)$. We will soon make the stronger assumption that $w\in L^1(\R^d)\cap L^\ii(\R^d)$ for simplicity.

We recall that the relative density $\rho_{\gamma-\Pi_\mu^-}$ has no sign in general. However, by~\eqref{eq:focusing_vs_defocusing} the second term in \eqref{eq:def-relative-energy} is non-negative when $\hat{w}\ge0$. Furthermore, when $d=2$ and $\|(\hat{w})_-\|_{L^\ii}<K_{\rm LT}/(2\pi)$ with $(\hat{w})_-:=\max(-\hat{w},0)$, we have 
$$\cE(\gamma,\Pi_\mu^-)\ge\left(1-\frac{2\pi\|\hat{w}_-\|_{L^\ii}}{K_{\rm LT}}\right)\tr(-\Delta-\mu)(\gamma-\Pi_\mu^-)\ge0.$$
We cannot apply the same argument when $d\ge3$ because the right side of \eqref{eq:Lieb-Thirring-zero-temp} behaves as $(\rho_{\gamma-\Pi_\mu^-})^{1+2/d}$ when $\rho_{\gamma-\Pi_\mu^-}$ is large, while the second term of \eqref{eq:def-relative-energy} behaves as $(\rho_{\gamma-\Pi_\mu^-})^2$. Summarizing our discussion, we have the 

\begin{proposition}[Coercivity of the relative energy in the defocusing case]\label{prop:relative-energy-positive}
 Let $\gamma\in\cK_\mu$. Then, when $d\ge3$ and $\hat{w}\ge0$ we have 
\begin{equation}
  \cE(\gamma,\Pi_\mu^-)\ge\tr(-\Delta-\mu)(\gamma-\Pi_\mu^-),
\label{eq:coercive_zero_temp_3D}
 \end{equation}
while when $d=2$ and $\|(\hat{w})_-\|_{L^\ii}<K_{\rm LT}/(2\pi)$ we have 
\begin{equation}
\cE(\gamma,\Pi_\mu^-)\ge\left(1-\frac{2\pi\|\hat{w}_-\|_{L^\ii}}{K_{\rm LT}}\right)\tr(-\Delta-\mu)(\gamma-\Pi_\mu^-).
\label{eq:coercive_zero_temp_2D}
\end{equation}
\end{proposition}

The interest of the relative free energy $\cE(\gamma,\Pi_\mu^-)$ is that it is a conserved quantity. In the defocusing case it controls the norm of $\gamma-\Pi_\mu^-$ in $\cX_\mu$ and this can be used to prove global well-posedness.

\begin{theorem}[Global well-posedness in energy space at zero temperature]\label{thm:global_zero_temp}
Assume that $d\geq2$ and that $w\in L^1(\R^d)\cap L^\ii(\R^d)$ is such that $w(x)=w(-x)$ for a.e. $x\in\R^d$. Let $Q_0=\gamma_0-\Pi_\mu^-$ with $\gamma_0\in\cK_\mu$. Then there exists a unique maximal solution $Q(t)=\gamma(t)-\Pi_\mu^-\in C^0_t((-T^-,T^+),\cB)\cap L^\ii_{t,\rm loc}((-T^-,T^+),\cX_\mu)$ of~\eqref{eq:rHF_Q_Duhamel} with $\gamma(t)\in\cK_\mu$ for all $t$ and $T^\pm>0$. We have the blow-up criterion
\begin{equation}
T^\pm<\ii \Longrightarrow \lim_{t\to \pm T^\pm}\norm{Q(t)}_{\cX_\mu}=+\ii. 
\label{eq:blowup_energy_space}
\end{equation}
Its relative Hartree energy is conserved:
$$\cE(\gamma(t),\Pi_\mu^-)=\cE(\gamma_0,\Pi_\mu^-),\quad \forall t\in(-T^-,T^+).$$

Furthermore, when $d\ge3$ and $\hat{w}\ge0$ or when $d=2$ and $\|(\hat{w})_-\|_{L^\ii}<K_{\rm LT}/(2\pi)$, then the solution is global: $T^-=T^+=+\ii$.
\end{theorem}

Our proof is based on the Lieb-Thirring inequality~\eqref{eq:Lieb-Thirring-zero-temp} and for this reason we do not have a result in dimension $d=1$. The proof  is \emph{not} done directly in the energy space $\cX_\mu$. In order to avoid the difficulty associated with the possible divergence on the Fermi sphere $k^2\simeq \mu$, we rather work in a \emph{larger} space $\cY_\mu$, and then prove that solutions which are in the energy space at time 0 remain in it for all times. This strategy is simpler, but it requires the non-optimal condition that $w\in L^1(\R^d)\cap L^\ii(\R^d)$. 

An immediate consequence of Theorem~\ref{thm:global_zero_temp} is the orbital stability of the stationary solution $\Pi_\mu^-$ in the defocusing case.

\begin{corollary}[Orbital stability of $\Pi_\mu^-$ in the defocusing case]\label{cor:orbital_stability_zero_temp}
Assume that $d\geq2$, that $w\in L^1(\R^d)\cap L^\ii(\R^d)$, and that $\hat{w}\ge0$ when $d\ge3$ or $\|(\hat{w})_-\|_{L^\ii}<K_{\rm LT}/(2\pi)$ when $d=2$. We have 
\begin{multline}
\sup_{t\in\R} \tr(-\Delta-\mu)(\gamma(t)-\Pi_\mu^-)\\ 
\leq C\left( \tr(-\Delta-\mu)(\gamma_0-\Pi_\mu^-)+\Big(\tr(-\Delta-\mu)(\gamma_0-\Pi_\mu^-)\Big)^{\frac{2d}{d+2}}\right)
\label{eq:orbital_stability}
\end{multline}
for all $\gamma_0\in\cK_\mu$ and where $C$ only depends on $d$, $\mu$ and $w$.
\end{corollary}

We recall that the relative kinetic energy $\tr(-\Delta-\mu)(\gamma-\Pi_\mu^-)$ controls the Hilbert-Schmidt norm of $(\gamma-\Pi_\mu^-)|\Delta+\mu|^{1/2}$ and the trace norm of 
$|\Delta+\mu|^{1/2}(\gamma-\Pi_\mu^-)^{\pm\pm}|\Delta+\mu|^{1/2}.$ Hence the result means that if $\gamma_0$ is close to $\Pi_\mu^-$ in these norms, then $\gamma(t)$ stays close to $\Pi_\mu^-$ for all times. The energy does not provide a control on the operator norm $\|\gamma(t)-\Pi_\mu^-\|$, however.

\begin{proof}[Proof of Corollary~\ref{cor:orbital_stability_zero_temp}]
The conservation of the relative energy means that
\begin{align*}
\tr(-\Delta-\mu)(\gamma(t)-\Pi_\mu^-)&\leq c\cE(\gamma_0,\Pi_\mu^-)\\
&\leq c\tr(-\Delta-\mu)(\gamma_0-\Pi_\mu^-)+C\norm{\rho_{Q_0}}_{L^2+L^{1+2/d}}^2. 
\end{align*}
The constant on the first line is $c=1$ for $d\geq3$ and $c=K_{\rm LT}/(K_{\rm LT}-2\pi\|\hat{w}_-\|_{L^\ii})$ for $d=2$. 
The rest follows from~\eqref{eq:Lieb-Thirring-zero-temp}.
\end{proof}

The rest of the section is devoted to the proof of Theorem~\ref{thm:global_zero_temp}. 

\begin{proof}[Proof of Theorem~\ref{thm:global_zero_temp}]
Our method is the following: first we show the local well-posedness in a suitable space, which is larger than the energy space $\cX_\mu$. Then we prove the conservation of the relative Hartree energy for the very smooth solutions which we have constructed in Section~\ref{sec:local}. By approximating any initial datum with finite energy by a sequence of smooth finite rank operators and passing to the limit, we obtain solutions in the energy space, with a conserved energy. Uniqueness is guaranteed by the first step.

\subsubsection*{Step 1. Local well-posedness in a space larger than the energy space}
We start by considering the problem in a space $\cY_\mu$ which contains $\cK_\mu$ (for $d\geq2$). As we want to avoid the Fermi sphere $|k|=\sqrt{\mu}$ at which the operator $Q$ can fail to be compact, we use the projector $\Pi_{2\mu}^+$, and only impose a condition on the density of the operator in the ball of radius $\sqrt{2\mu}$. We define
\begin{multline}
 \cY_\mu:=\Big\{Q=Q^*\in\cB\ :\ \Pi_{2\mu}^+Q\Pi_{2\mu}^-\in\gS_2,\,\Pi_{2\mu}^+Q\Pi_{2\mu}^+\in\gS_1,\\ \rho_{\Pi_{2\mu}^-Q\Pi_{2\mu}^-}\in L^2(\R^d)\Big\}.
\label{eq:def_Y_mu} 
\end{multline}
We remark that $-\norm{Q}\Pi_{2\mu}^-\leq \Pi_{2\mu}^-Q\Pi_{2\mu}^-\leq \norm{Q}\Pi_{2\mu}^-$ which shows that the density of 
$\Pi_{2\mu}^-Q\Pi_{2\mu}^-$ is well defined in $L^\ii(\R^d)$ and it therefore makes sense to assume that $\rho_{\Pi_{2\mu}^-Q\Pi_{2\mu}^-}\in L^2(\R^d)$. The space $\cY_\mu$ is endowed with the norm
$$
  \|Q\|_{\cY_\mu}:=\|Q\|+\|\Pi_{2\mu}^+Q\Pi_{2\mu}^-\|_{\gS_2}+\|\Pi_{2\mu}^+Q\Pi_{2\mu}^+\|_{\gS_1}+\|\rho_{\Pi_{2\mu}^-Q\Pi_{2\mu}^-}\|_{L^2}.
$$

The following says that $\cK_\mu-\Pi_\mu^-\subset \cY_\mu$, that is, the set of density matrices of finite energy is included in $\cY_\mu$.

\begin{lemma}\label{lem:energy_space_in_Ymu}
Let $d\geq2$ and $\mu>0$. We have 
\begin{equation}
\norm{Q}_{\cY_\mu}\leq C\left(\norm{Q}_{\cX_\mu}+\norm{Q}_{\cX_\mu}^{1/2}\right)
\label{eq:control_Ymu_Xmu} 
\end{equation}
for all $Q\in\cX_\mu$ satisfying $-\Pi_\mu^-\leq Q\leq \Pi_\mu^+$ and with $C=C(d,\mu)$.
\end{lemma}
\begin{proof}
We have 
$$\norm{\Pi^+_{2\mu }Q \Pi^-_{2\mu }}_{\gS^2}\leq \norm{|\Delta+\mu|^{1/2}Q }_{\gS^2}\norm{\Pi^+_{2\mu }|\Delta+\mu|^{-1/2}}=\frac{\norm{|\Delta+\mu|^{1/2}Q }_{\gS^2}}{\sqrt{\mu}}$$
and a similar inequality for $\|\Pi^+_{2\mu }Q \Pi^+_{2\mu }\|_{\gS^1}$. On the other hand, since $-\Pi_\mu^-\leq \Pi^-_{2\mu }Q \Pi^-_{2\mu }\leq \Pi_{2\mu}^-$, we have $|\rho_{\Pi^-_{2\mu }Q \Pi^-_{2\mu }}(x)|\leq C\mu^{d/2}$ almost everywhere. The Lieb-Thirring inequality~\eqref{eq:Lieb-Thirring} yields 
\begin{align*}
\int_{\R^d}  |\rho_{\Pi^-_{2\mu }Q \Pi^-_{2\mu }}|^2&\leq C\int_{\R^d} \min\left( |\rho_{\Pi^-_{2\mu }Q \Pi^-_{2\mu }}|^2,|\rho_{\Pi^-_{2\mu }Q \Pi^-_{2\mu }}|^{1+\frac{2}{d}}\right)\\
&\leq C\tr (-\Delta-\mu) Q
\end{align*}
with constants depending on $d$ and $\mu$, which proves the result.
\end{proof}

Adapting the proof of Lemma~\ref{lem:density_Schatten-Sobolev}, it is easy to see that $\cY_\mu\subset\gS_{1,\rm{loc}}$ and thus $\rho_Q$ is a well-defined function for $Q\in\cY_\mu$. We actually have 
$\rho_Q\in (L^1+L^2)(\R^d)$ with the precise estimate
\begin{equation}
\norm{\rho_{\Pi_{2\mu}^-Q\Pi_{2\mu}^-}}_{L^2}+\norm{\rho_{\Pi_{2\mu}^+Q\Pi_{2\mu}^-}}_{L^2}+\norm{\rho_{\Pi_{2\mu}^-Q\Pi_{2\mu}^+}}_{L^2}+\norm{\rho_{\Pi_{2\mu}^+Q\Pi_{2\mu}^+}}_{L^1}\leq C\norm{Q}_{\cY_\mu}.
\label{eq:estim_rho_Y_mu}
\end{equation}
It is very important for us that the map 
$$Q\in\cY_\mu \mapsto \frac12\int_{\R^d}\int_{\R^d}w(x-y)\rho_{Q}(x)\rho_{Q}(y)\,dx\,dy$$
is continuous when $w\in L^1(\R^d)\cap L^{\ii}(\R^d)$, which follows from~\eqref{eq:estim_rho_Y_mu}.

\begin{proposition}[Local well-posedness in $\cY_\mu$]\label{prop:local-wp-cY}
Let $d\ge1$ and $w\in L^1(\R^d)\cap L^{\ii}(\R^d)$. 
Let $Q_0\in\cY_{\mu}$ such that 
\begin{equation}
\rho_{e^{it\Delta}\Pi_{2\mu}^-Q_0\Pi_{2\mu}^-e^{-it\Delta}}\in C^0_t(\R,L^2_x(\R^d)).
\label{eq:condition_initial_Y_mu} 
\end{equation}
Then there exists a unique maximal solution $Q\in\cC^0((-T^-,T^+),\cY_\mu)$ to the Duhamel principle \eqref{eq:rHF_Q_Duhamel}, with $T^\pm>0$. We have the blow-up criterion
\begin{equation}
T^\pm<\ii \Longrightarrow \lim_{t\to \pm T^\pm}\norm{Q(t)}_{\cY_\mu}=+\ii. 
\label{eq:blowup_Ymu}
\end{equation}
The solution depends continuously on $Q_0$ in $\cY_\mu$, on $\rho_{e^{it\Delta}\Pi_{2\mu}^-Q_0\Pi_{2\mu}^-e^{-it\Delta}}$ in $C^0_t((-T^-,T^+),L^2_x(\R^d))$, and on $w$ in $L^1(\R^d)\cap L^{\ii}(\R^d)$. 
\end{proposition}

The assumption~\eqref{eq:condition_initial_Y_mu} looks complicated but it is automatically satisfied if $Q_0$ is a smooth finite rank operator, or if $Q_0=\gamma_0-\Pi_\mu^-$  with $\gamma_0\in \cK_\mu$, since in this case $\norm{e^{it\Delta}Q_0e^{-it\Delta}}_{\cX_\mu}=\norm{Q_0}_{\cX_\mu}$ for all $t$. This is the situation which we will need later on.

\begin{proof}[Proof of Proposition~\ref{prop:local-wp-cY}]
The proof of the proposition relies on a fixed point technique, exactly like for the proof of Theorem~\ref{thm:local-wp-Sps} in  Section~\ref{sec:local}. It follows from the estimates:
\begin{equation}
\sup_{t\in\R}\norm{e^{it\Delta}[\Pi_\mu^-,\rho_Q\ast w]e^{-it\Delta}}_{\cY_\mu}\leq C\norm{Q}_{\cY_\mu},
\label{eq:estim_Y_mu_1}
\end{equation}
\begin{equation}
\sup_{t\in\R}\norm{e^{it\Delta}[Q',\rho_Q\ast w]e^{-it\Delta}}_{\cY_\mu}\leq C\norm{Q'}_{\cY_\mu}\norm{Q}_{\cY_\mu}.
\label{eq:estim_Y_mu_2}
\end{equation}
All the terms in the definition of $\|\cdot\|_{\cY_\mu}$ are invariant under the free Schr\"odinger evolution, except for the one involving the low-momentum density.

We start by proving~\eqref{eq:estim_Y_mu_1} and note as a start that
$$\norm{[\Pi_\mu^-,\rho_Q\ast w]}\leq \norm{\rho_Q\ast w}_{L^\ii}\leq C(\norm{w}_{L^2}+\norm{w}_{L^\ii})\norm{Q}_{\cY_\mu}$$
by~\eqref{eq:estim_rho_Y_mu}. Then we remark that $\Pi_{2\mu}^+[\rho_Q\ast w,\Pi_\mu^-]\Pi_{2\mu}^+=0$ since $\Pi_\mu^-\Pi_{2\mu}^+=0$. Writing for short $V:=\rho_Q\ast w$, we turn to 
\begin{equation*}
\norm{\Pi_{2\mu}^+[V,\Pi_\mu^-]\Pi_{2\mu}^-}_{\gS^2}=\norm{\Pi_{2\mu}^+V\Pi_\mu^-}_{\gS^2}\leq  \norm{V\Pi_\mu^-}_{\gS^2}=C\mu^{d/4}\norm{V}_{L^2},
\end{equation*}
which can be estimated by $C(\norm{w}_{L^1}+\norm{w}_{L^2})\norm{Q}_{\cY_\mu}$ by~\eqref{eq:estim_rho_Y_mu}. Finally, we estimate the density of $e^{it\Delta}\Pi_{2\mu}^-[V,\Pi_\mu^-]\Pi_{2\mu}^-e^{-it\Delta}$ by duality, similarly as in Lemma~\ref{lem:density_Schatten-Sobolev}:
\begin{align*}
\int_{\R^d}\rho_{e^{it\Delta}\Pi_{2\mu}^-V\Pi_\mu^-e^{-it\Delta}}W&=\tr(e^{it\Delta}\Pi_{2\mu}^-[V,\Pi_\mu^-]\Pi_{2\mu}^-e^{-it\Delta}W)\\
&\leq 2\norm{\Pi_{2\mu}^-V}_{\gS^2}\norm{\Pi_{2\mu}^-W}_{\gS^2}\leq C\norm{V}_{L^2}\norm{W}_{L^2} 
\end{align*}
where $C$ is independent of $t$.

Next we turn to~\eqref{eq:estim_Y_mu_2} and note first that $\norm{[V,Q']}\leq 2\norm{V}_{L^\ii}\norm{Q'}\leq C\norm{Q}_{\cY_\mu}\norm{Q'}_{\cY_\mu}$. For the Hilbert-Schmidt term, we write 
$$\Pi_{2\mu}^+VQ'\Pi_{2\mu}^- = \Pi_{2\mu}^+V\Pi_{2\mu}^-Q'\Pi_{2\mu}^- + \Pi_{2\mu}^+V\Pi_{2\mu}^+Q'\Pi_{2\mu}^-,$$
 and estimate the two terms on the right side. We have $\|V\Pi_{2\mu}^-\|_{\gS^2}= C\mu^{d/4}\|V\|_{L^2}$,
 leading to 
 $$\|\Pi_{2\mu}^+V\Pi_{2\mu}^-Q'\Pi_{2\mu}^-\|_{\gS^2}\le C\|V\|_{L^2}\|Q'\| \le C\|Q\|_{\cY_\mu}\|Q'\|_{\cY_\mu}.$$
 We also have 
 $$\|\Pi_{2\mu}^+V\Pi_{2\mu}^+Q'\Pi_{2\mu}^-\|_{\gS^2}\le \|V\|_{L^\ii}\|\Pi_{2\mu}^+Q'\Pi_{2\mu}^-\|_{\gS^2}\le C\|Q\|_{\cY_\mu}\|Q'\|_{\cY_\mu},$$
 and hence  $\|\Pi_{2\mu}^+VQ'\Pi_{2\mu}^-\|_{\gS^2}\le C\|Q\|_{\cY_\mu}\|Q'\|_{\cY_\mu}.$
The proof for $\Pi_{2\mu}^+Q'V\Pi_{2\mu}^-$ and for $\Pi_{2\mu}^+[V,Q']\Pi_{2\mu}^+$ is similar.
Finally, $\rho[e^{it\Delta}\Pi_{2\mu}^-[V,Q']\Pi_{2\mu}^-e^{-it\Delta}]$ is estimated by duality as before.

The rest of the proof of Proposition~\ref{prop:local-wp-cY} goes along the same lines as that of Theorem~\ref{thm:local-wp-Sps}. The assumption~\eqref{eq:condition_initial_Y_mu} on $Q_0$ is needed to ensure that the term $e^{it\Delta}Q_0e^{-it\Delta}$ belongs to $C^0_t(\R,\cY_\mu)$. In order to iterate the fixed point theorem and construct a maximal solution, we need to verify that the condition~\eqref{eq:condition_initial_Y_mu} is propagated along the flow. So we have to prove that
$\rho_{e^{is\Delta}\Pi_{2\mu}^-Q(t)\Pi_{2\mu}^-e^{-is\Delta}}\in C^0_s(\R,L^2_x(\R^d))$.
From Duhamel's formula we have 
\begin{multline*}
\rho\left[e^{is\Delta}\Pi_{2\mu}^-Q(t)\Pi_{2\mu}^-e^{-is\Delta}
\right]=\rho\left[e^{i(t+s)\Delta}Q_0e^{-i(t+s)\Delta}\right]\\
-i\int_0^t\rho\left[e^{i(s+t-t')\Delta}[w*\rho_Q(t'),\gamma_f+Q(t')]e^{i(t'-t-s)\Delta}\right]\,dt'.
\end{multline*}
The first term is continuous and locally bounded under our assumption~\eqref{eq:condition_initial_Y_mu} on $Q_0$, and the second 
can be uniformly estimated by $C\int_0^t(\norm{Q(t')}_{\cY_\mu}+\norm{Q(t')}_{\cY_\mu}^2)\,dt'$. The continuity follows from a density argument.
\end{proof}

\subsubsection*{Step 2. Conservation of relative energy for smooth operators}

In Theorem~\ref{thm:local-wp-Sps} we have constructed smooth solutions to the Hartree equation, in the Sobolev-Schatten spaces $\gS^{p,s}$. We remark that $\gS^{1,s}\subset \cY_\mu$ for $s>\max(2,d/2)$, since then $\rho_Q\in L^2$ by Lemma~\ref{lem:density_Schatten-Sobolev}. By an argument similar to Corollary~\ref{cor:persistance-regularity}, one can prove that the maximal time of existence of solutions in $\gS^{1,s}$ is the same as in $\cY_\mu$.
The purpose of this step is to prove the conservation of the relative Hartree energy $\cE(\gamma(t),\Pi_\mu^-)$ when $\gamma_0-\Pi_\mu^-\in\gS^{1,s}$ with $s$ large enough.

\begin{proposition}[Conservation of relative energy in $\gS^{1,s}$]\label{prop:conservation-relative-energy-cS1s}
 Let $d\ge1$, $s>2+\max(2,d/2)$, and $w\in L^1\cap L^\ii$ satisfying all the assumptions of Theorem~\ref{thm:local-wp-Sps}, with furthermore $w(x)=w(-x)$ for a.e. $x\in\R^d$. Let $Q_0\in\gS^{1,s}$ and consider $Q\in C^0_t((-T^-,T^+),\gS^{1,s})$ the unique maximal solution to \eqref{eq:rHF_Q_Duhamel} built in Theorem~\ref{thm:local-wp-Sps} with $\gamma_f=\Pi_\mu^-$. Then, the relative particle number and the relative energy are conserved along the flow: we have 
$$\cE(\Pi_\mu^-+Q(t),\Pi_\mu^-)=\cE(\Pi_\mu^-+Q_0,\Pi_\mu^-)\quad\text{and}\quad \tr\,Q(t)=\tr\,Q_0$$ 
for all $t\in(-T^-,T^+)$. 
\end{proposition}

\begin{proof}
  We abbreviate $\cE(Q):=\cE(\Pi_\mu^-+Q,\Pi_\mu^-)$ for shortness:
  \begin{multline*}
   \cE(Q(t)) = \tr|-\Delta-\mu|^{1/2}(Q(t)^{++}-Q(t)^{--})|-\Delta-\mu|^{1/2} \\
    + \frac12 \iint\rho_{Q(t)}(x)\rho_{Q(t)}(y)w(x-y)\,dx\,dy.
  \end{multline*}
  When $Q\in\gS^{1,s}$ and $s>\max(2,d/2)$, then $V_Q=w*\rho_Q\in (L^2\cap L^\ii)(\R^d)$ and we infer that
  $$\iint\rho_{Q(t)}(x)\rho_{Q(t)}(y)w(x-y)\,dxdy = \tr(V_Q Q).$$

From Duhamel's formula we see that $\partial_t Q\in \gS^{1,s-2}$ and we can differentiate the function $t\mapsto\cE(Q(t))$. The first term in $\cE(Q(t))$ is differentiable when $\partial_t Q\in\gS^{1,2}$ which requires $s\ge4$. Similarly, the second term is differentiable when $s>2+d/2$. 
  Using that $w(x)=w(-x)$, we have for all $t\in(-T^-,T^+)$
  $$\partial_t{\cE}(Q(t))=\tr|-\Delta-\mu|^{1/2}(\dot{Q}^{++}-\dot{Q}^{--})|-\Delta-\mu|^{1/2}+\tr(V_Q\dot{Q}),$$
  with 
  $$\dot{Q}^{++}-\dot{Q}^{--}=(-i)\left([-\Delta,Q^{++}-Q^{--}]+[V_Q,Q]^{++}-[V_Q,Q]^{--}\right).$$
  Since $Q^{++},Q^{--}\in\gS^{1,4}$, we have
  $$\tr|-\Delta-\mu|^{\frac 12}[-\Delta,Q^{++}]|-\Delta-\mu|^{\frac 12}=0=\tr|-\Delta-\mu|^{\frac 12}[-\Delta,Q^{--}]|-\Delta-\mu|^{\frac 12}.$$
  We now remark that $(-\Delta-\mu)[V_Q,Q]\in\gS^1$. Indeed, $V_Q\in L^\ii$ and $Q\in \gS^{1,2}$, hence $(-\Delta-\mu)QV_Q\in\gS^1$. For the other term, we commute the Laplacian as follows
$$(-\Delta)V_Q Q=\big((-\Delta V_Q)+2\nabla V_Q\cdot\nabla\big)Q+V_Q(-\Delta)Q$$
and we use that $\Delta V_Q,\nabla V_Q\in L^\ii$, since $w\in W^{2,1}\cap W^{2,\ii}$ by assumption.
As a consequence, $(-\Delta-\mu)[V_Q,Q]\in\gS^1$ and we have 
  $$\tr|-\Delta-\mu|^{1/2}([V_Q,Q]^{++}-[V_Q,Q]^{--})|-\Delta-\mu|^{1/2}=\tr(-\Delta-\mu)[V_Q,Q].$$
  Since we also have $V_Q Q\in\gS^1$, then $\tr[V_Q,Q]=0$ and 
  $$\tr(-\Delta-\mu)[V_Q,Q]=\tr(-\Delta)[V_Q,Q]=\tr V_Q[Q,-\Delta].$$
  Summarizing, we thus have proved that 
  $$\partial_t\tr|-\Delta-\mu|^{1/2}(Q^{++}-Q^{--})|-\Delta-\mu|^{1/2}=\tr[-\Delta,V_Q]Q.$$
  Now let us compute the time derivative of the second term in the energy:
  $$\partial_t\frac12 \tr(V_Q Q)=\tr(V_Q\dot{Q})=(-i)\tr V_Q\left([-\Delta+V_Q,\gamma_f+Q]\right).$$
  By the proof of Lemma \ref{lemma:commutator-Vgamma-Hs}, we know that $V_Q\gamma_f\in\gS^{1,s}$, and therefore the operators $V_Q^2\gamma_f$ and $V_Q\gamma_f V_Q$ are both in $\gS^1$ since $V_Q\in L^\ii$. We deduce that $\tr V_Q[V_Q,\gamma_f]=0$. We have also seen that $V_Q Q\in\gS^1$, hence $V_Q^2 Q, V_Q Q V_Q\in\gS^1$ and hence $\tr V_Q[V_Q,Q]=0$. We finally have
 $\partial_t\tr(V_Q Q)=2\tr V_Q[-\Delta,Q],$
  implying that
  $$\partial_t{\cE}(Q(t))=\tr V_Q[Q,-\Delta]+\tr V_Q[-\Delta,Q]=0,$$
  and the energy is conserved. The proof for $\tr\,Q(t)$ is similar.
\end{proof}

\subsubsection*{Step 3. A solution starting in $\cK_\mu$ stays in $\cK_\mu$}

Here we prove that any solution in $\cY_\mu$ such that $Q_0=\gamma_0-\Pi_\mu^-$ with $\gamma_0\in\cK_\mu$ stays in the energy space for all times, and that the relative Hartree energy is conserved. Since states in $\cK_\mu$ do not necessarily have a finite particle number relative to $\Pi_\mu^-$, we \emph{a priori} lose this conservation law.

\begin{proposition}[Conservation of the relative energy]\label{prop:conservation_energy_Kmu}
Assume that $d\geq2$ and that $w\in L^1\cap L^\ii$ with $w(x)=w(-x)$ for a.e. $x\in\R^d$. Let $Q_0=\gamma_0-\Pi_\mu^-$ with $\gamma_0\in\cK_\mu$. Then the unique maximal solution $Q(t)=\gamma(t)-\Pi_\mu^-\in C^0_t((-T^-,T^+),\cY_\mu)$ satisfies that $\gamma(t)\in\cK_\mu$ for all $t\in\R$ and the relative energy is conserved:
$$\cE(\gamma(t),\Pi_\mu^-)=\cE(\gamma_0,\Pi_\mu^-),\quad \forall t\in(-T^-,T^+).$$
\end{proposition}

\begin{proof}
We start by assuming that $w$ satisfies all the assumptions of Theorem~\ref{thm:local-wp-Sps}, with $p=1$ and some $s>2+\max(2,d/2)$. We will remove this condition at the end of the proof.
Let $\gamma_0=Q_0+\Pi_\mu^-$ be any density matrix in $\cK_\mu$, and denote by $Q(t)$ the associated maximal solution in $C^0((-T^-,T^+),\cY_\mu)$. The following is an improvement of~\cite[Lemma 3.2]{FraLewLieSei-12} and it says that one can approximate $Q_0$ by a sequence of smooth finite-rank operators, with the low-momentum density converging strongly in $L^2$.

\begin{lemma}\label{lem:density_finite_rank}
Let $\gamma=Q+\Pi_\mu^-\in\cK_\mu$ with $d\geq2$. Then there exists a sequence $-\Pi_\mu^-\leq Q^{(n)}\leq\Pi_\mu^+$
 of finite rank operators, compactly supported in the Fourier domain, such that 
 \begin{itemize}
  \item For all $\phi\in L^2(\R^d)$, $Q^{(n)}\phi\to Q\phi$ in $L^2(\R^d))$;
  \item $\lim_{n\to\ii}\|(Q^{(n)}-Q)|-\Delta-\mu|^{1/2}\|_{\gS^2}=0$;
  \item $\lim_{n\to\ii}\||-\Delta-\mu|^{1/2}(Q^{(n)}-Q)^{\pm\pm}|-\Delta-\mu|^{1/2}\|_{\gS^1}=0$;
  \item $\rho_{Q^{(n)}}\to\rho_{Q}$ in $L^1_{\rm loc}(\R^d)$;
  \item $\rho_{\Pi_{2\mu}^-Q^{(n)}\Pi_{2\mu}^-}\to\rho_{\Pi_{2\mu}^-Q\Pi_{2\mu}^-}$ strongly in $L^{2}(\R^d)$.
 \end{itemize}
\end{lemma}

We finish the proof of Proposition~\ref{prop:conservation-relative-energy-cS1s} before turning to that of Lemma~\ref{lem:density_finite_rank}. We take a sequence $Q_0^{(n)}$ as in Lemma~\ref{lem:density_finite_rank} and remark that $Q_0^{(n)} \to Q_0$ in $\cY_\mu$. Furthermore, by the proof of Lemma~\ref{lem:density_finite_rank}, we also have 
$\rho[e^{it\Delta}\Pi_{2\mu}^-Q_0^{(n)}\Pi_{2\mu}^-e^{-it\Delta}] \to\rho[e^{it\Delta}\Pi_{2\mu}^-Q_0\Pi_{2\mu}^-e^{-it\Delta}]$ 
strongly in $L^{2}(\R^d)$ and uniformly in $t$ on any compact set of $\R$. By Proposition~\ref{prop:local-wp-cY}, this proves that the maximal times of existence of the associated solution $Q^{(n)}(t)$ satisfy $\liminf_{n\to\ii}T_n^\pm\geq T^\pm$ and that $Q^{(n)}(t)\to Q(t)$ in $C^0_t([-T^-+\epsilon,T^+-\epsilon],\cY_\mu)$ for any $\epsilon>0$. By~\eqref{eq:estim_rho_Y_mu}, this implies in particular that $\rho_{Q^{(n)}(t)}\to \rho_{Q(t)}$ strongly in $C^0_t([-T^-+\epsilon,T^+-\epsilon],L^1+L^2)$ and, since $w\in L^1\cap L^\ii$, that
\begin{multline*}
\lim_{n\to\ii} \int_{\R^d}\int_{\R^d}w(x-y)\rho_{Q^{(n)}(t)}(x)\rho_{Q^{(n)}(t)}(y)\,dx\,dy\\
=\int_{\R^d}\int_{\R^d}w(x-y)\rho_{Q(t)}(x)\rho_{Q(t)}(y)\,dx\,dy 
\end{multline*}
uniformly in $t$ on $[-T^-+\epsilon,T^+-\epsilon]$. By Proposition~\ref{prop:conservation-relative-energy-cS1s}, the relative Hartree energy of $Q^{(n)}(t)$ is conserved and, therefore,
\begin{multline}
\lim_{n\to\ii}\tr(-\Delta-\mu)Q^{(n)}(t)\\=\tr(-\Delta-\mu)Q_0+\frac12\int_{\R^d}\int_{\R^d}w(x-y)\rho_{Q_0}(x)\rho_{Q_0}(y)\,dx\,dy \\
-\frac12\int_{\R^d}\int_{\R^d}w(x-y)\rho_{Q(t)}(x)\rho_{Q(t)}(y)\,dx\,dy
\label{eq:limit_energy_Qt}
\end{multline}
where the right side is uniformly bounded for $t\in[-T^-+\epsilon,T^+-\epsilon]$. Here we have used that
$\lim_{n\to\ii}\tr(-\Delta-\mu)Q_0^{(n)}=\tr(-\Delta-\mu)Q_0$, by Lemma~\ref{lem:density_finite_rank}.

From~\eqref{eq:estim_rho_Y_mu} and the assumptions on $w$, for $n$ large enough the potential $w\ast \rho_{Q^{(n)}(t)}$ is in $L^1([-T^-+\epsilon,T^+-\epsilon], L_x^\ii)$. By \cite{Yajima-87}, this implies that there exists a unitary operator $U^{(n)}(t,t')$ on $L^2(\R^d)$ associated with the time-dependent Hamiltonian $-\Delta+w\ast \rho_{Q^{(n)}(t)}$. Then $Q^{(n)}(t)=U^{(n)}(t,0)(\Pi_\mu^-+Q_0^{(n)})U^{(n)}(0,t)-\Pi_\mu^-$. In particular, since $0\le\Pi_\mu^-+Q_0^{(n)}\le1$, then $0\le\Pi_\mu^-+Q^{(n)}(t)\le1$ for all $t\in(-T^-+\epsilon,T^+-\epsilon)$. By the same argument, we have $0\leq \gamma(t)=\Pi_\mu^-+Q(t)\leq 1$ as well, for all $t\in(-T^-,T^+)$. In order to show that $\gamma(t)\in \cK_\mu$, it suffices to prove that $\tr(-\Delta-\mu)Q(t)$ is finite for all $t$.

We recall that the relative kinetic energy is defined by
$\tr(-\Delta-\mu)Q:=\tr|\Delta+\mu|^{1/2}\big(Q^{++}-Q^{--}\big)|\Delta+\mu|^{1/2}$
and it is weakly lower semi-continuous since $Q^{++},-Q^{--}\geq0$. From~\eqref{eq:limit_energy_Qt} we deduce that
\begin{multline*}
\tr(-\Delta-\mu)Q(t) \leq \tr(-\Delta-\mu)Q_0+\frac12\iint w(x-y)\rho_{Q_0}(x)\rho_{Q_0}(y)\,dx\,dy \\
-\frac12\iint w(x-y)\rho_{Q(t)}(x)\rho_{Q(t)}(y)\,dx\,dy
\end{multline*}
which proves that $\gamma(t)\in\cK_\mu$ for all $t$, and that $\cE(\gamma(t),\Pi_\mu^-)\leq \cE(\gamma_0,\Pi_\mu^-)$, for all $t\in(-T^-,T^+)$. In order to derive the other inequality, it suffices to exchange the role of $t$ and $0$. We start the dynamics at $t$ with initial condition $\gamma(t)\in\cK_\mu$ and, by uniqueness, the so-obtained solution coincides with the other dynamics started at $0$. The previous argument provides the reverse inequality $\cE(\gamma_0,\Pi_\mu^-)\leq \cE(\gamma(t),\Pi_\mu^-)$.

Our proof was based on the fact that $w$ is very smooth, but now that we have shown the conservation of the energy for all solutions in $\cK_\mu$, we can remove this condition by using the continuity of the solutions in $\cY_\mu$ with respect to $w$. On the other hand, the blow up criterion~\eqref{eq:blowup_Ymu} follows from the one in $\cY_\mu$ by~\eqref{eq:control_Ymu_Xmu}. We skip the details. 
\end{proof}

It remains to provide the
\begin{proof}[Proof of Lemma~\ref{lem:density_finite_rank}]
Like in~\cite[Lemma 3.2]{FraLewLieSei-12}, we let $Q^{(n)}:=P_nQP_n$ with $P_n:=\1(1/n\leq |-\Delta-\mu|\leq n)$, which localizes away from the Fermi surface and from infinity. The operator $Q^{(n)}$ is not yet of finite rank, but it is Hilbert-Schmidt and very smooth, and it can itself be approximated by a finite rank operator by~\cite[Thm  6]{HaiLewSer-09}. The convergence properties of $Q^{(n)}$ are all proved in~\cite[Lemma 3.2]{FraLewLieSei-12}, except for the last one, for which we compute  
\begin{align}
\Pi_{2\mu}^-Q^{(n)}\Pi_{2\mu}^--\Pi_{2\mu}^-Q\Pi_{2\mu}^-=& \Pi_{2\mu}^-P_nQP_n\Pi_{2\mu}^--\Pi_{2\mu}^-Q\Pi_{2\mu}^-\nonumber\\
=&-\Pi_{2\mu}^-P_n^\perp QP_n^\perp\Pi_{2\mu}^- - \Pi_{2\mu}^-P_n^\perp QP_n\Pi_{2\mu}^-\nonumber\\
& -\Pi_{2\mu}^- P_nQP_n^\perp\Pi_{2\mu}^-\label{eq:decomp_density_rho}
\end{align}
with $n\geq 2\mu$. 
The density of $\Pi_{2\mu}^-P_n^\perp QP_n^\perp\Pi_{2\mu}^-$ is in $L^\ii$ due to the projection $\Pi_{2\mu}^-$ and, by the Lieb-Thirring inequality~\eqref{eq:Lieb-Thirring-zero-temp}, we have
\begin{multline*}
\int_{\R^d}\rho^2_{\Pi_{2\mu}^-P_n^\perp QP_n^\perp\Pi_{2\mu}^-}\leq 
C\int_{\R^d}\min(\rho^2_{\Pi_{2\mu}^-P_n^\perp QP_n^\perp\Pi_{2\mu}^-},\rho^{1+2/d}_{\Pi_{2\mu}^-P_n^\perp QP_n^\perp\Pi_{2\mu}^-})\\ \leq C\tr|-\Delta-\mu|\Pi_{2\mu}^-P_n^\perp QP_n^\perp\Pi_{2\mu}^-\to0. 
\end{multline*}
We still need to control the density of the two other terms in~\eqref{eq:decomp_density_rho} and we only explain how to deal with the first one. We write 
\begin{align*}
\Pi_{2\mu}^-P_n^\perp QP_n\Pi_{2\mu}^-=&P_n^\perp Q^{--}P_n +\Pi_{2\mu}^-P_n^\perp Q^{++}P_n\Pi_{2\mu}^-\\
&+\Pi_{\mu}^-P_n^\perp QP_n\Pi_{2\mu}^-\Pi_{\mu}^+ +\Pi_{\mu}^+\Pi_{2\mu}^-P_n^\perp QP_n\Pi_{\mu}^-
\end{align*}
Following the proof in~\cite[p 456--457]{FraLewLieSei-12}, we find that the density of the last two terms can be estimated in $L^2$ by a constant times
$$\norm{QP_n^\perp|\Delta+\mu|^{1/2}}_{\gS^2}^{1/2}\norm{QP_n|\Delta+\mu|^{1/2}}_{\gS^2}^{1/2}$$
which tends to 0 when $n\to\ii$. We therefore only have to deal with the first two terms. Using that $Q^{--}\leq0$, we write
$$\frac{1}{4\epsilon} P_n^\perp Q^{--}P_n^\perp + \epsilon P_n Q^{--}P_n \leq P_n^\perp Q^{--}P_n\leq -\frac{1}{4\epsilon} P_n^\perp Q^{--}P_n^\perp - \epsilon P_n Q^{--}P_n$$
and deduce
$|\rho_{P_n^\perp Q^{--}P_n}|\leq (4\epsilon)^{-1} |\rho_{P_n^\perp Q^{--}P_n^\perp}| + \epsilon|\rho_{P_n Q^{--}P_n}|.$
From~\eqref{eq:Lieb-Thirring-zero-temp} and the fact that the density of all the operators appearing in the previous estimates are uniformly bounded, we conclude that
$$\int_{\R^d}|\rho_{P_n^\perp Q^{--}P_n}|^2\leq -\frac{1}{4\epsilon}\tr|\Delta+\mu|P_n^\perp Q^{--}P_n^\perp-\epsilon\tr|\Delta+\mu|Q^{--}.$$
The argument is the same for the term involving $Q^{++}$ and we get the result by letting first $n\to\ii$ and then $\epsilon\to0$.
\end{proof}

\subsubsection*{Step 4. Global well-posedness in the defocusing case}
Now we assume that $d\ge3$ and $\hat{w}\ge0$ or that $d=2$ and $\|(\hat{w})_-\|_{L^\ii}<K_{\rm LT}/(2\pi)$. Then we deduce from~\eqref{eq:coercive_zero_temp_3D} and~\eqref{eq:coercive_zero_temp_2D} that
$$\tr(-\Delta-\mu)Q(t)\leq\begin{cases}
\cE(\gamma_0,\Pi_\mu^-)&\text{for $d\geq3$,}\\[0.2cm]
\dps \frac{\cE(\gamma_0,\Pi_\mu^-)}{1-\frac{2\pi\|\hat{w}_-\|_{L^\ii}}{K_{\rm LT}}}&\text{for $d=2$,}
\end{cases}$$
and all $t\in(-T^-,T^+)$. This shows that $Q(t)$ is uniformly bounded in $\cX_\mu$ and therefore in $\cY_\mu$, by~\eqref{eq:control_Ymu_Xmu}. This implies global well-posedness, by the blow-up criterion~\eqref{eq:blowup_Ymu} in $\cY_\mu$.
\end{proof}

\section{Local well-posedness with Strichartz estimates}\label{sec:Strichartz}
  
In Section~\ref{sec:local}, we have constructed solutions to the Hartree equation~\eqref{eq:rHF_DM} in $\gS^{p,s}$ with $s>d(p-1)/p$. We then used similar techniques in Section~\ref{sec:global_zero_temp} to deal with the special case of a Fermi gas at zero temperature, $\gamma_f=\Pi_\mu^-$. Our goal is to solve the positive temperature cases~\eqref{eq:Fermi-gas-positive-temp}--\eqref{eq:Boltzmann-gas-positive-temp} as well. Because of Klein's inequality (see~\eqref{eq:Klein2} in Section~\ref{sec:entropy}), it is natural to look for solutions in the space $\gS^{2,1}$ of operators $Q$ such that $Q(1-\Delta)^{1/2}\in\gS^2$. This space is covered by Theorem~\ref{thm:local-wp-Sps} only in dimension $d=1$. We treat here the cases $d=1,2,3$ by a different method.

So far we have only employed estimates relating $Q(t)$ and $\rho_{Q(t)}$ which are uniform with respect to the time variable. That is, we have not used the regularization properties of the time-dependent Schr\"odinger propagator, which are usually expressed by means of \emph{Strichartz inequalities}. The usual inequality of this kind~\cite{Strichartz-77,KeeTao-98,Cazenave-03,Tao-06} is of little use in our context and, instead, our results are based on a more precise Strichartz inequality in Schatten spaces, which has recently been proved in~\cite{FraLewLieSei-13}.

In this section we take again a general reference state $\gamma_f=f(-\Delta)$. We do not assume anymore that $\gamma_f=\Pi_\mu^-$. We want to lower the exponent $s$ measuring the regularity of $Q$. For shortness, we only consider two cases in this section: $\gS^2$ in dimension $d=1,2$ and $\gS^{2,1}$ in dimension $d=1,2,3$. Our method allows to decrease the regularity exponent $s$ in higher dimensions, but without reaching $s=1$. Also, we can deal with other Schatten exponents $p$, but we do not mention this here.

 \begin{theorem}[Well-posedness in $\gS^2$ in $d=1,2$]\label{thm:local-wp-S2}
  Let $d\in\{1,2\}$, $w\in L^1(\R^d)\cap L^\ii (\R^d)$ and $f\in L^\ii(\R^+,\R)$ such that $\int_{\R^d}|f(k^2)|\,dk<\ii$.
Assume that 
\begin{equation}
Q_0\in\gS^2\quad\text{and}\quad w*\rho_{e^{it\Delta}Q_0e^{-it\Delta}}\in L^{8/d}_{t,{\rm loc}}(\R,L^2_x\cap L^\ii_x).
\label{eq:condition_initial_gS2}
\end{equation}
Then there exists a unique maximal solution $Q(t)\in C^0((-T^-,T^+),\gS^2)$ to~\eqref{eq:rHF_Q_Duhamel} with $T^\pm>0$, such that 
$w*\rho_{Q(t)}\in L^{8/d}_{t,{\rm loc}}((-T^-,T^+),L^2_x\cap L^\ii_x)$.
We have the blow up criterion
\begin{equation}\label{eq:blowup-S2}
T^\pm<\ii \Longrightarrow\lim_{t\to \pm T^\pm}\int_0^{\pm t}\norm{w\ast\rho_{Q(t)}}_{L^2_x\cap L^\ii_x}^{8/d}\,dt=+\ii. 
\end{equation}
Furthermore, this solution is continuous with respect to $Q_0$, $w$ and $f$: if $Q_0^{(n)}\to Q_0$, $w_n\to w$ and $f_n\to f$ for the norms appearing above, then the corresponding  unique solution $Q_n(t)$ has a time of existence such that $\liminf T_n^\pm\geq T^\pm$, and we have for every $\epsilon>0$
\begin{multline*}
\lim_{n\to\ii}\norm{w\ast \rho_{Q(t)}-w_n\ast \rho_{Q_n(t)}}_{L^{8/d}((-T^-+\epsilon,T^+-\epsilon),L^2_x\cap L^\ii_x)}\\=\lim_{n\to\ii}\norm{Q_n(t)-Q(t)}_{L^\ii((-T^-+\epsilon,T^+-\epsilon),\gS^2)}=0.
\end{multline*}
  \end{theorem}

 \begin{theorem}[Well-posedness in $\gS^{2,1}$ in $d=1,2,3$]\label{thm:local-wp-S21}
  Let $d\in\{1,2,3\}$, $w,\nabla w\in L^1(\R^d)\cap L^\ii (\R^d)$ and $f\in L^\ii(\R^+,\R)$ such that $\int_{\R^d}(|f(k^2)|+k^2|f(k^2)|^2)\,dk<\ii$.
Assume that 
\begin{equation}
Q_0\in\gS^{2,1}\quad\text{and}\quad w*\rho_{e^{it\Delta}Q_0e^{-it\Delta}}\in L^{8/d}_{t,{\rm loc}}(\R,H^1_x\cap W^{1,\ii}_x).
\label{eq:condition_initial_gS21}
\end{equation}
Then there exists a unique maximal solution $Q(t)\in C^0((-T^-,T^+),\gS^{2,1})$ to~\eqref{eq:rHF_Q_Duhamel} with $T^\pm>0$, such that 
$w*\rho_{Q(t)}\in L^{8/d}_{t,{\rm loc}}((-T^-,T^+),H^1_x\cap W^{1,\ii}_x)$.
We have the blow up criterion
\begin{equation}\label{eq:blowup-S21}
T^\pm<\ii \Longrightarrow\lim_{t\to \pm T^\pm}\int_0^{\pm t}\norm{w\ast\rho_{Q(t)}}_{H^1_x\cap W^{1,\ii}_x}^{8/d}\,dt=+\ii. 
\end{equation}

Furthermore, this solution is continuous with respect to $Q_0$, $w$ and $f$, in the same fashion as in the previous theorem.
  \end{theorem}

As in Proposition~\ref{prop:local-wp-cY}, we need the complicated assumptions~\eqref{eq:condition_initial_gS2} and~\eqref{eq:condition_initial_gS21} on the density of the initial datum $Q_0$. This is necessary since an operator $Q_0\in\gS^2$ (or $\in\gS^{2,1}$ for $d=3$) has no well-defined density in general. The assumption on the density is satisfied if for instance $Q_0\in \gS^{8/7}$, by the Strichartz inequality of~\cite{FraLewLieSei-13} which precisely states that
\begin{equation}
\norm{\rho_{e^{it\Delta}Qe^{-it\Delta}}}_{L^p_tL^q_x}\leq C\norm{Q}_{\gS^{\frac{2q}{q+1}}},\quad\text{for}\ 1\leq q\leq 1+\frac{2}{d}\ \text{and}\ \frac{2}{p}+\frac{d}{q}=d.
\label{eq:Strichartz}
\end{equation}
Later,~\eqref{eq:condition_initial_gS21} will be satisfied for another reason, namely because $\gamma_0=\gamma_f+Q_0$ will have a finite entropy relative to $\gamma_f$.

The strategy to prove Theorems~\ref{thm:local-wp-S2} and~\ref{thm:local-wp-S21} is different from our previous proofs of local well-posedness. We use a fixed-point argument on the \emph{potential} $V_Q=w\ast\rho_Q$. To do so, we notice that, if $Q\in C^0_t([-T,T],\gS^1_{{\rm loc}})$ is such that $V_Q\in L^1_t((-T,T),L^\ii_x)$, it is equivalent that $Q$ solves \eqref{eq:rHF_Q_Duhamel} and that
  \begin{equation}\label{eq:duhamel-W}
    Q(t)=e^{it\Delta}\cW_{V_Q}(t,0)(\gamma_f+Q_0)\cW_{V_Q}(t,0)^*e^{-it\Delta}-\gamma_f,\quad\forall t\in[-T,T],
 \end{equation}
  where the wave operator $\cW_V(t,0)$ is defined for any $V\in L^1_t((-T,T),L^\ii_x)$ via the Dyson series
 \begin{equation}\label{eq:Dyson-series}
    \cW_V(t,0)=1+\sum_{n\ge1}\cW^{(n)}_V(t,0),
 \end{equation}
 with, for all $t,t_0\in[-T,T]$, 
 \begin{multline}
    \cW_V^{(n)}(t,t_0):=(-i)^n\int_{t_0}^t\,dt_n\int_{t_0}^{t_n}\,dt_{n-1}\cdots\int_{t_0}^{t_2}\,dt_1\times\\
    \times e^{i(t_0-t_n)\Delta}V(t_n)e^{i(t_n-t_{n-1})\Delta}\cdots e^{i(t_2-t_1)\Delta}V(t_1)e^{i(t_1-t_0)\Delta}.
    \label{eq:nth_wave_op}
 \end{multline}
 We recall that this series is absolutely convergent in $\cB(L^2(\R^d))$ for any fixed $t,t_0$, since 
 $$\|\cW^{(n)}_V(t,t_0)\|\le\frac{1}{n!}\|V\|_{L^1_t([-T,T],L^\ii_x(\R^d))}^n$$
for any $n\geq1$. As a consequence, the potential $V_Q$ solves the equation
 \begin{equation}\label{eq:duhamel-potential}
    V_Q(t)=w*\left(\rho\left[e^{it\Delta}\cW_{V_Q}(t,0)(\gamma_f+Q_0)\cW_{V_Q}(t,0)^*e^{-it\Delta}\right]-\rho_{\gamma_f}\right),
 \end{equation}
in the space $L^1_t((-T,T),L^\ii_x)$. We find a local solution to \eqref{eq:duhamel-potential} in the smaller space $L^{8/d}_t((-T,T),L^2_x\cap L^\ii_x)$, for $T>0$ small enough, by a fixed point argument. The solution will then be extended to a maximal time of existence by verifying that the conditions~\eqref{eq:condition_initial_gS2} and~\eqref{eq:condition_initial_gS21} are propagated along the flow. 

The rest of this section is devoted to the proof of the two theorems. As usual, we first detail useful estimates before turning to the actual proof. The main ingredient is an estimate from \cite{FraLewLieSei-13} on the wave operators $\cW_V^{(n)}(t,t_0)$. For $n=1$, the estimate dual to~\eqref{eq:Strichartz} is 
  \begin{equation}\label{eq:est-W1}
\left\|\cW^{(1)}_V(t,t_0)\right\|_{\gS^{2q}}\le C\|V\|_{L^p_t L^q_x},
  \end{equation}
and, for $n\geq2$, it reads
  \begin{equation}\label{eq:est-Wn}
\left\|\cW^{(n)}_V(t,t_0)\right\|_{\gS^{2\left\lceil\frac{q}{n}\right\rceil}}\le\frac{C^n}{(n!)^{\frac{1}{p}-\epsilon}}\|V\|_{L^p_t L^q_x}^n.
  \end{equation}
Here $t,t_0\in\R$, $d\ge1$, $1+d/2\le q<\ii$, $2/p+{d}/q=2$ and $0<\epsilon<1/p$. The constant $C$ depends on $d,p,q,\epsilon$ but not on $t,t_0\in\R$.
 
In addition to this result, we need some estimates adapted to the space $\gS^{2,1}$ (that is, involving derivatives) for the second order term $\cW^{(2)}_V(t,t_0)$. These estimates are only valid for $d=1,2,3$, and this is where the dimensional restriction in Theorem \ref{thm:local-wp-S21} comes from.

 \begin{lemma}[Second order term $\cW^{(2)}$]\label{lemma:est-W2-H1}
  Let $d\in\{1,2,3\}$. For any $U,V\in L^{8/(8-d)}_t(\R,L^2_x\cap L^4_x)$, and for any $t,t_0\in\R$, we define the operator 
  $$\cW^{(2)}_{(U,V)}(t,t_0):=-\int_{t_0}^t\,dt_2\int_{t_0}^{t_2}\,dt_1e^{-it_2\Delta}U(t_2)e^{i(t_2-t_1)\Delta}V(t_1)e^{it_1\Delta}.$$
  Then, we have the estimate
 \begin{equation}\label{eq:est-W2-14-14}
    \left\|(1-\Delta)^{-\frac14}\cW^{(2)}_{(U,V)}(t,t_0)(1-\Delta)^{-\frac14}\right\|_{\gS^2}\le C\|U\|_{L^{\frac{8}{8-d}}_t(L^2_x\cap L^4_x)}\|V\|_{L^{\frac{8}{8-d}}_t(L^2_x\cap L^4_x)}.
 \end{equation}
Furthermore, for any $U,V\in L^{\frac{5}{5-d}}_t(\R,L^{5/2}_x)\cap L^{\frac{10}{10-d}}_t(\R,L^{5/2}_x)$, we have
  \begin{multline}\label{eq:est-W2-12}
    \left\|\cW^{(2)}_{(U,V)}(t,t_0)(1-\Delta)^{-\frac12}\right\|_{\gS^2}\\
    \le C\|U\|_{L^{\frac{5}{5-d}}_t(L^{5/2}_x)}^{1/2}\|V\|_{L^{\frac{5}{5-d}}_t(L^{5/2}_x)}^{1/2}\|U\|_{L^{\frac{10}{10-d}}_t(L^{5/2}_x)}^{1/2}\|V\|_{L^{\frac{10}{10-d}}_t(L^{5/2}_x)}^{1/2}.
 \end{multline}
The  constant $C>0$ only depends on the dimension $d$.
 \end{lemma}

 \begin{proof}[Proof of Lemma \ref{lemma:est-W2-H1}]
  We start by proving \eqref{eq:est-W2-14-14}. Let $U,V\in L^{\frac{8}{8-d}}_t(\R,(L^2_x\cap L^4_x)(\R^d))$. The method is the same as the proof of \cite[Thm  2]{FraLewLieSei-13}: noticing that $e^{-it\Delta}V(x)e^{it\Delta}=V(x+2tp)$ with $p:=-i\nabla$, we deduce that
  \begin{multline*}
   \left\|(1-\Delta)^{-\frac14}\cW^{(2)}_{(U,V)}(t,t_0)(1-\Delta)^{-\frac14}\right\|_{\gS^2}^2 = \tr\int_{t_0}^tdt_2\int_{t_0}^{t_2}dt_1\int_{t_0}^tdt_4\int_{t_0}^{t_4}dt_3\\
   (1-\Delta)^{-\frac14}U(t_2,x+2t_2p)V(t_1,x+2t_1p)(1-\Delta)^{-\frac12}\times\\
   \times V(t_3,x+2t_3p)U(t_4,x+2t_4p)(1-\Delta)^{-\frac14},
  \end{multline*}
  which can be estimated by
  \begin{multline*}
    \int_{\R^4}dt_1dt_2dt_3dt_4\|(1-\Delta)^{-\frac14}|U|^{\frac12}(t_2,x+2t_2p)\|_{\gS^8}\times\\
     \times\| |U|^{\frac12}(t_2,x+2t_2p) |V|^{\frac12}(t_1,x+2t_1p)\|_{\gS^4} \| |V|^{\frac12}(t_1,x+2t_1p)(1-\Delta)^{-\frac14}\|_{\gS^8}\times\\
     \times\|(1-\Delta)^{-\frac14}|V|^{\frac12}(t_3,x+2t_3p)\|_{\gS^8}\| |V|^{\frac12}(t_3,x+2t_3p) |U|^{\frac12}(t_4,x+2t_4p)\|_{\gS^4}\times\\
      \times\| |U|^{\frac12}(t_4,x+2t_4p) (1-\Delta)^{-\frac14}\|_{\gS^8}.
    \end{multline*}
  Using that $k\mapsto(1+k^2)^{-1/4}\in L^8(\R^d)$ for $d=1,2,3$ and the estimate 
  \begin{equation}\label{eq:est-generalized-KSS}
   \|f(\alpha x+\beta p)g(\gamma x+\delta p)\|_{\gS^r}\le (2\pi)^{-d/r}\frac{\|f\|_{L^r(\R^d)}\|g\|_{L^r(\R^d)}}{|\alpha\delta-\beta\gamma|^{\frac{d}{r}}}
  \end{equation}
  valid for all $r\ge2$ by \cite[Lemma 1]{FraLewLieSei-13}, we deduce that
  \begin{multline*}
   \left\|(1-\Delta)^{-\frac14}\cW^{(2)}_{(U,V)}(t,t_0)(1-\Delta)^{-\frac14}\right\|_{\gS^2}^2 \le C\int_{\R^4}dt_1dt_2dt_3dt_4\times\\
      \times\frac{\|U(t_1)\|_{L^2_x\cap L^4_x}\|U(t_2)\|_{L^2_x\cap L^4_x}\|V(t_3)\|_{L^2_x\cap L^4_x}\|V(t_4)\|_{L^2_x\cap L^4_x}}{|t_1-t_2|^{\frac d4}|t_3-t_4|^{\frac d4}}.
  \end{multline*}
  The inequality \eqref{eq:est-W2-14-14} follows from the multilinear Hardy-Littlewood-Sobolev (HLS) inequality \cite[Thm  4]{FraLewLieSei-13}. To prove \eqref{eq:est-W2-12}, we use the same method:
  \begin{multline*}
   \left\|\cW^{(2)}_{(U,V)}(t,t_0)(1-\Delta)^{-\frac12}\right\|_{\gS^2}^2\\
    \le
    \int_{\R^4}dt_1dt_2dt_3dt_4\|(1-\Delta)^{-\frac12} |V|^{\frac12}(t_3,x+2t_3p)\|_{\gS^5}\times\\
     \times\| |V|^{\frac12}(t_3,x+2t_3p) |U|^{\frac12}(t_4,x+2t_4p)\|_{\gS^5}\times\\
     \times\| |U|^{\frac12}(t_4,x+2t_4p) |U|^{\frac12}(t_2,x+2t_2p)\|_{\gS^5}\times\\
     \times\| |U|^{\frac12}(t_2,x+2t_2p) |V|^{\frac12}(t_1,x+2t_1p)\|_{\gS^5}\| |V|^{\frac12}(t_1,x+2t_1p)(1-\Delta)^{-1/2}\|_{\gS^5}.
  \end{multline*}
  Then, again by \eqref{eq:est-generalized-KSS}, we have 
  \begin{multline*}
   \left\|\cW^{(2)}_{(U,V)}(t,t_0)(1-\Delta)^{-\frac12}\right\|_{\gS^2}^2\\
    \le
    \int_{\R^4}dt_1dt_2dt_3dt_4\frac{\|U(t_1)\|_{L^{5/2}_x}\|U(t_2)\|_{L^{5/2}_x}\|V(t_3)\|_{L^{5/2}_x}\|V(t_4)\|_{L^{5/2}_x}}{|t_1-t_2|^{\frac d5}|t_2-t_4|^{\frac d5}|t_3-t_4|^{\frac d5}},
  \end{multline*}
  which leads to \eqref{eq:est-W2-12} by the multilinear HLS inequality.
 \end{proof}

 \begin{lemma}\label{lemma:Wn-gammaf-S2}
  Let $d\ge1$. We have 
  $$\|\cW^{(n)}_V(t,t_0)g(-i\nabla)\|_{\gS^2}\le (2\pi)^{d/2}\frac{\|V\|_{L^1_tL^\ii_x}^{n-1}}{(n-1)!}\|V\|_{L^1_tL^2_x}\|g\|_{L^2},\,\forall t,t_0\in\R$$
for any $V\in L^1_t(\R,L^2_x\cap L^\ii_x)$, $g\in L^2(\R^d)$, and all $n\ge1$.
 \end{lemma}

 \begin{proof}
  We write
  \begin{align*}
&\|\cW^{(n)}_V(t,t_0)g(-i\nabla)\|_{\gS^2} \\
    &\le \int_{t_0}^t\!dt_n\int_{t_0}^{t_n}\!dt_{n-1}\cdots\int_{t_0}^{t_2}\!dt_1\|V(t_n)\|_{L^\ii_x}\cdots\|V(t_2)\|_{L^\ii_x}\|V(t_1,x)g(-i\nabla)\|_{\gS^2}\\
    &\le (2\pi)^{\frac{d}2}\|V\|_{L^1_tL^2_x}\|g\|_{L^2}\!\!\int_{t_0}^t\!dt_n\int_{t_0}^{t_n}\!dt_{n-1}\cdots\!\int_{t_0}^{t_3}\,dt_2\|V(t_n)\|_{L^\ii_x}\cdots\|V(t_2)\|_{L^\ii_x}\\
    &\le (2\pi)^{\frac{d}2}\frac{\|V\|_{L^1_tL^\ii_x}^{n-1}}{(n-1)!}\|V\|_{L^1_tL^2_x}\|g\|_{L^2}.
  \end{align*}

\vspace{-0,5cm}
 \end{proof}

By using the previous estimates, we are now able to get local solutions to the potential equation~\eqref{eq:duhamel-potential}.

\begin{proposition}[Solving the potential equation~\eqref{eq:duhamel-potential}]\label{prop:local-wp-potential}
Let $d\in\{1,2,3\}$, $w\in L^1(\R^d)\cap L^\ii (\R^d)$ and $f\in L^\ii(\R^+,\R)$ such that $\int_{\R^d}|f(k^2)|\,dk<\ii$. Assume that 
\begin{equation}
Q_0\in\gS^2\quad\text{and}\quad w*\rho_{e^{it\Delta}Q_0e^{-it\Delta}}\in L^{8/d}_{t,{\rm loc}}(\R,L^2_x\cap L^\ii_x)
\label{eq:condition_initial_gS2_bis}
\end{equation}
with the additional condition that $Q_0\in\gS^{2,1}$ if $d=3$. Then there exists a time $T>0$ and a solution 
$V\in L^{8/d}_t((-T,T),L^2_x\cap L^\ii_x)$
of~\eqref{eq:duhamel-potential}. The corresponding density satisfies
\begin{multline}
\rho\left[e^{it\Delta}\cW_{V}(t,0)(\gamma_f+Q_0)\cW_{V}(t,0)^*e^{-it\Delta}\right]-\rho_{\gamma_f}-\rho\left[e^{it\Delta}Q_0e^{-it\Delta}\right]\\
\in L^{8/d}_t((-T,T),L^1_x+ L^2_x). 
\label{eq:density_in_L1L2}
\end{multline}
Furthermore, if 
\begin{multline*}
  \norm{Q_0}_{\gS^2\,/\,\gS^{2,1}}+\norm{w\ast \rho_{e^{it\Delta}Q_0e^{-it\Delta}}}_{L^{8/d}_t((-1,1),L^2_x\cap L^\ii_x)}+\norm{w}_{L^1\cap L^\ii}\\
  +\norm{f}_{L^\ii}+\int_{\R^d}|f(k^2)|\,dk\le R,
\end{multline*}
then the time $T>0$ can be chosen to depend only on $R$. 
\end{proposition}

 \begin{proof}[Proof of Proposition \ref{prop:local-wp-potential}]
Let $Q_0$ be as in the statement, $0<T\leq1$ to be chosen later and 
$R\geq \norm{Q_0}_{\gS^2\,/\,\gS^{2,1}}+\norm{w\ast \rho_{e^{it\Delta}Q_0e^{-it\Delta}}}_{L^{8/d}_t((-1,1),L^2_x\cap L^\ii_x)}$
where the norm for $Q_0$ is that of $\gS^2$ when $d=1,2$ and $\gS^{2,1}$ when $d=3$. For any $V\in L^{8/d}_t((-T,T),L^2_x\cap L^\ii_x)$, define
  $$\Phi(V)(t)=w*\left(\rho\left[e^{it\Delta}\cW_V(t,0)(\gamma_f+Q_0)\cW_V(t,0)^*e^{-it\Delta}\right]-\rho_{\gamma_f}\right),$$
  for all $t\in(-T,T)$. We now provide the necessary estimates to apply a fixed-point argument to $\Phi$. First, we perform the decomposition
  \begin{multline}\label{eq:decomposition-Q}
   e^{it\Delta}\cW_V(t,0)(\gamma_f+Q_0)\cW_V(t,0)^*e^{-it\Delta}=\gamma_f+e^{it\Delta}Q_0e^{-it\Delta}\\
    +e^{it\Delta}(\cW^{(1)}_V(t,0)Q_0+Q_0\cW^{(1)}_V(t,0)^*)e^{-it\Delta}\\
    +\sum_{n\ge1}e^{it\Delta}(\cW^{(n)}_V(t,0)\gamma_f+\gamma_f\cW^{(n)}_V(t,0)^*)e^{-it\Delta}+A_V(t).
  \end{multline} 
  Let us now show that $A_V\in C^0_t([-T,T],\gS^1)$, which will imply that $\rho_{A_V}\in C^0_t([-T,T],L^1)$. From the definition of $A_V$, we have
  \begin{align*}
   A_V(t)=&\sum_{n\ge2}e^{it\Delta}(\cW^{(n)}_V(t,0)Q_0+Q_0\cW^{(n)}_V(t,0)^*)e^{-it\Delta}\\
   &+e^{it\Delta}\cW^{(1)}_V(t,0)Q_0\cW^{(1)}_V(t,0)^*e^{-it\Delta}\\
   &+\sum_{\max(n,m)\ge2}e^{it\Delta}\cW^{(n)}_V(t,0)Q_0\cW^{(m)}_V(t,0)^*e^{-it\Delta}\\
   &+\sum_{n,m\ge1}e^{it\Delta}\cW^{(n)}_V(t,0)\gamma_f\cW^{(m)}_V(t,0)^*e^{-it\Delta}.
  \end{align*}
   We estimate the last term using Lemma \ref{lemma:Wn-gammaf-S2}:
   \begin{align*}
    \sum_{n,m\ge1}\!\!\|e^{it\Delta}\cW^{(n)}_V(t,0)\gamma_f\cW^{(m)}_V(t,0)^*e^{-it\Delta}\|_{\gS^1} &\le \left(\sum_{n\ge1}\|\cW^{(n)}_V(t,0)|\gamma_f|^{1/2}\|_{\gS^2}\right)^2\\
    &\le C\left(\sum_{n\ge1}\frac{\|V\|_{L^1_t(L^2_x\cap L^\ii_x)}^n}{(n-1)!}\right)^2.
   \end{align*}
   Notice that we used here that $T\le1$ to have $\|V\|_{L^1_t}\le\|V\|_{L^{8/d}_t}$. The terms in $A_V(t)$ involving $Q_0$ are treated differently according to the dimension and we first deal with $d=1,2$. Let $0<\epsilon=\epsilon(d)<1-d/4$. Then, by~\eqref{eq:est-Wn}, we have for all $n\ge1$ and for $d=1$,
    $$\|\cW^{(n)}_V(t)\|_{\gS^{2\lceil2/n\rceil}}\le\frac{C^n}{(n!)^{\frac{3}{4}-\epsilon}}\|V\|_{L^{\frac43}_tL^2_x}^n\le\frac{C^n}{(n!)^{1-\frac d4-\epsilon}}\|V\|_{L^{8/d}_t(L^2_x\cap L^\ii_x)}^n,$$
   while for $d=2$ we have 
    $$\|\cW^{(n)}_V(t)\|_{\gS^{2\lceil2/n\rceil}}\le\frac{C^n}{(n!)^{\frac{1}{2}-\epsilon}}\|V\|_{L^2_{t,x}}^n\le\frac{C^n}{(n!)^{1-\frac d4-\epsilon}}\|V\|_{L^{8/d}_t(L^2_x\cap L^\ii_x)}^n.$$
   Therefore, using that
   $$\|e^{it\Delta}\cW^{(1)}_V(t,0)Q_0\cW^{(1)}_V(t,0)^*e^{-it\Delta}\|_{\gS^1}\le C\|\cW^{(1)}_V(t,0)\|_{\gS^4}^2\|Q_0\|_{\gS^2},$$
   \begin{multline*}
      \sum_{n\ge2}\|e^{it\Delta}(\cW^{(n)}_V(t,0)Q_0+Q_0\cW^{(n)}_V(t,0)^*)e^{-it\Delta}\|_{\gS^1}\\
      \le 2\|Q_0\|_{\gS^2}\sum_{n\ge2}\|\cW^{(n)}_V(t,0)\|_{\gS^2},
   \end{multline*}
and
   \begin{multline*}
      \sum_{\max(n,m)\ge2}\|e^{it\Delta}\cW^{(n)}_V(t,0)Q_0\cW^{(m)}_V(t,0)^*e^{-it\Delta}\|_{\gS^1}\\
      \le C\|Q_0\|_{\gS^2}\sum_{n\ge2}\|\cW^{(n)}_V(t,0)\|_{\gS^2}\sum_{m\ge1}\|\cW^{(m)}_V(t,0)\|,
   \end{multline*}
   we deduce the estimate uniform in $t\in[-T,T]$:
   \begin{multline}\label{eq:est-AV-S1}
   \|A_V(t)\|_{\gS^1}\le2\|Q_0\|_{\gS^2}\left(\sum_{n\ge0}\frac{C^n}{(n!)^{1-\frac d4-\epsilon}}\|V\|_{L^{8/d}_t(L^2_x\cap L^\ii_x)}^n\right)^2\\
   +C\left(\sum_{n\ge1}\frac{1}{(n-1)!}\|V\|_{L^{8/d}_t(L^2_x\cap L^\ii_x)}^n\right)^2.
  \end{multline}
  If $d=3$, we have $\|\cW^{(1)}_V(t)\|_{\gS^5}\leq C\|V\|_{L^{5/2}_{t,x}}$ by~\eqref{eq:est-W1} and 
  $$\forall n\geq2,\ \|\cW^{(n)}_V(t)\|_{\gS^{2\lceil 5/(2n)\rceil}}\le\frac{C^n}{(n!)^{\frac25-\epsilon}}\|V\|_{L^{5/2}_{t,x}}^n\le\frac{C^n}{(n!)^{\frac25-\epsilon}}\|V\|_{L^{8/3}_t(L^2_x\cap L^\ii_x)}^n,$$
by \eqref{eq:est-Wn}.
  In particular, $\cW^{(n)}_V(t,0)\in\gS^2$ only for $n\ge3$, while $\cW^{(2)}_V(t,0)\in\gS^4$. Hence, the operators  $\cW^{(n)}_V(t,0)Q_0$ is trace-class when $n\ge3$ and $\cW^{(n)}_V(t,0)Q_0\cW^{(m)}(t,0)^*$ is trace-class when $\max(n,m)\ge3$, by H\"older's inequality. Furthermore, using the KSS inequality \eqref{eq:KSS}, we infer that  
  $$\|\cW^{(1)}_V(t)(1-\Delta)^{-1/2}\|_{\gS^{10/3}}\le C\int_0^t\|V(s)\|_{L^{10/3}_x}\,ds\le C\|V\|_{L^{\frac83}_t(L^2_x\cap L^\ii_x)}$$
for all $t\in[-T,T]$. 
  Hence, the operator
  $$\cW^{(1)}_V(t)Q_0\cW^{(1)}_V(t)^*=\cW^{(1)}_V(t)(1-\Delta)^{-1/2}(1-\Delta)^{1/2}Q_0\cW^{(1)}_V(t)^*$$
  belongs to $\gS^1$ since $1=3/10+1/2+1/5$.
  Finally, by \eqref{eq:est-W2-12}, we have $\cW^{(2)}_V(t)Q_0\in\gS^1$ for all $t\in[-T,T]$ with
  \begin{align*}
    \|\cW^{(2)}_V(t)Q_0\|_{\gS^1} &\le C\|V\|_{L^{5/2}_{t,x}}\|V\|_{L^{10/7}_tL^{5/2}_x}\|Q_0\|_{\gS^{2,1}}\\
    &\le C\|V\|_{L^{8/3}_t(H^1_x\cap W^{1,\ii}_x)}^2\|Q_0\|_{\gS^{2,1}}.
  \end{align*}
  As a consequence, we find that $A_V(t)\in\gS^1$ also for $d=3$, with the same estimate \eqref{eq:est-AV-S1}, because $(n!)^{-2/5}\le(n!)^{-1/4}=(n!)^{-1+d/4}$, except that $\norm{Q_0}_{\gS^2}$ is replaced by $\norm{Q_0}_{\gS^{2,1}}$. From the fact that $A_V(t)\in\gS^1$ with an estimate uniform in $t\in[-T,T]$, we deduce that  
  \begin{align*}
   \|w*\rho_{A_V(t)}\|_{L^{8/d}_t(L^2_x\cap L^\ii_x)} &\le C\|w\|_{L^1\cap L^2}\|\rho_{A_V(t)}\|_{L^{8/d}_t L^1_x}\\
   &\le CT^{d/8}\|\rho_{A_V(t)}\|_{L^\ii_t L^1_x}\le CT^{d/8}\|A_V(t)\|_{L^\ii_t\gS^1}.
  \end{align*}

  In order to estimate $\|\Phi(V)\|_{L^{8/d}_t(L^2_x\cap L^\ii_x)}$, it only remains to bound
  $$\left\|w*\rho\left[e^{it\Delta}(\cW^{(1)}_V(t,0)Q_0+Q_0\cW^{(1)}_V(t,0)^*)e^{-it\Delta}\right]\right\|_{L^{8/d}_t(L^2_x\cap L^\ii_x)}$$
  and
  $$\sum_{n\ge1}\left\|w*\rho\left[e^{it\Delta}(\cW^{(n)}_V(t,0)\gamma_f+\gamma_f\cW^{(n)}_V(t,0)^*)e^{-it\Delta}\right]\right\|_{L^{8/d}_t(L^2_x\cap L^\ii_x)}.$$
  We treat these terms by duality. Let $U\in L^2_x$. Then by Lemma \ref{lemma:Wn-gammaf-S2} and our assumption on $f$, the operator $\cW^{(n)}(t,0)\gamma_f U(x)$ is trace-class with 
  \begin{align*}
    \|\cW^{(n)}(t,0)\gamma_f U(x)\|_{\gS^1} &\le \|\cW^{(n)}(t,0)|\gamma_f|^{1/2}\|_{\gS^2}\||\gamma_f|^{1/2} U(x)\|_{\gS^2} \\
    &\le \frac{C}{(n-1)!}\|V\|_{L^{8/d}_t(L^2_x\cap L^\ii_x)}^n\|U\|_{L^2}.
  \end{align*}
  As a consequence, we have
  \begin{multline*}
   \sum_{n\ge1}\left\|\rho\left[e^{it\Delta}(\cW^{(n)}_V(t,0)\gamma_f+\gamma_f\cW^{(n)}_V(t,0)^*)e^{-it\Delta}\right]\right\|_{L^\ii_t L^2_x}\\
   \le C\sum_{n\ge1}\frac{1}{(n-1)!}\|V\|_{L^{8/d}_t(L^2_x\cap L^\ii_x)}^n,
  \end{multline*}
  and hence by Young's inequality
  \begin{multline*}
   \sum_{n\ge1}\left\|w*\rho\left[e^{it\Delta}(\cW^{(n)}_V(t,0)\gamma_f+\gamma_f\cW^{(n)}_V(t,0)^*)e^{-it\Delta}\right]\right\|_{L^{8/d}_t(L^2_x\cap L^\ii_x)}\\
   \le CT^{d/8}\|w\|_{L^1\cap L^2}\sum_{n\ge1}\frac{1}{(n-1)!}\|V\|_{L^{8/d}_t(L^2_x\cap L^\ii_x)}^n.
  \end{multline*}
 
Note that 
\begin{multline*}
    \int_{-T}^T\,dt\int_{\R^d}dx\, \rho\left[e^{it\Delta}\cW^{(1)}_V(t)Q_0e^{-it\Delta}\right](x)U(t,x)\\
    =i\tr\left[\left(\cW^{(2)}_{(U,V)}(T,0)+\cW^{(2)}_{(U,V)}(-T,0)\right)Q_0\right].
 \end{multline*}
In dimensions $d=1,2$, we obtain 
\begin{multline*}
    \left|\int_{-T}^T\,dt\int_{\R^d}dx\, \rho\left[e^{it\Delta}\cW^{(1)}_V(t)Q_0e^{-it\Delta}\right](x)U(t,x)\right|\\
    \leq C\norm{Q_0}_{\gS^2}\norm{\cW^{(2)}_{(U,V)}(T,0)}_{\gS^2}\leq C\norm{Q_0}_{\gS^2}\norm{U}_{L^{\frac{8}{8-d}}_t(L^4_x)}\norm{V}_{L^{\frac{8}{8-d}}_t(L^4_x)}
 \end{multline*}
by~\eqref{eq:est-Wn}, and this proves that
\begin{align*}
\norm{\rho\left[e^{it\Delta}\cW^{(1)}_V(t)Q_0e^{-it\Delta}\right]}_{L^{8/d}_t(L^{4/3}_x)}&\leq C\norm{Q_0}_{\gS^2}\norm{V}_{L^{\frac{8}{8-d}}_t(L^4_x)}\\
&\leq CT^{1-\frac d4}\norm{Q_0}_{\gS^2}\|V\|_{L^{\frac8d}_t(L^4_x)}.
\end{align*}
In dimension $d=3$, we have to use that $Q_0\in \gS^{2,1}$ through Lemma~\ref{lemma:est-W2-H1} and we obtain by the exact same method
\begin{align*}
\norm{\rho\left[e^{it\Delta}\cW^{(1)}_V(t)Q_0e^{-it\Delta}\right]}_{L^{8/d}_t(L^2_x+ L^{4/3}_x)}&\leq C\norm{Q_0}_{\gS^{2,1}}\norm{V}_{L^{\frac{8}{8-d}}_t(L^2_x\cap L^4_x)}\\
&\leq CT^{1-\frac d4}\norm{Q_0}_{\gS^{2,1}}\|V\|_{L^{\frac8d}_t(L^2_x\cap L^4_x)}.
\end{align*}
In all cases, we get
 \begin{multline*}
    \left\|w*\rho\left[e^{it\Delta}(\cW^{(1)}_V(t)Q_0+Q_0\cW^{(1)}_V(t)^*)e^{-it\Delta}\right]\right\|_{L^{\frac8d}_t(L^2_x\cap L^\ii_x)}\\
    \le CT^{1-\frac d4}\norm{w}_{L^1\cap L^\ii }\|V\|_{L^{\frac8d}_t(L^2_x\cap L^\ii_x)}.
 \end{multline*}

Gathering all our bounds, we obtain the final estimate
 \begin{align*}
  \|\Phi(V)\|_{L^{\frac8d}_t(L^2_x\cap L^\ii_x)}&\le R + CT^{1-\frac d4}\|V\|_{L^{\frac8d}_t(L^2_x\cap L^\ii_x)}\\
  &\ + CT^{d/8}\sum_{n\ge1}\frac{1}{(n-1)!}\|V\|_{L^{8/d}_t(L^2_x\cap L^\ii_x)}^n\\
  &\ +CT^{d/8}\|Q_0\|_{\gS^2\,/\,\gS^{2,1}}\left(\sum_{n\ge0}\frac{C^n}{(n!)^{1-\frac d4-\epsilon}}\|V\|_{L^{8/d}_t(L^2_x\cap L^\ii_x)}^n\right)^2\\
   &\ +CT^{d/8}\left(\sum_{n\ge1}\frac{1}{(n-1)!}\|V\|_{L^{8/d}_t(L^2_x\cap L^\ii_x)}^n\right)^2.
 \end{align*}
  This implies that for $T=T(R)>0$ small enough, $\Phi$ stabilizes the ball
 $\{\|V\|_{L^{8/d}_t(L^2_x\cap L^\ii_x)}\le 2R\}.$
  Using very similar estimates, we can also show that $\Phi$ is a contraction on this ball, decreasing $T$ if necessary. We do not detail these estimates here for brevity. This ends the proof of Proposition~\ref{prop:local-wp-potential}.
 \end{proof}

By using the local existence result of Proposition~\ref{prop:local-wp-potential}, we are now able to prove the main theorems of this section. We start with the

\begin{proof}[Proof of Theorem~\ref{thm:local-wp-S2}]
From Proposition~\ref{prop:local-wp-potential} we can construct a potential $V\in L^{8/d}_t(L^2_x\cap L^\ii_x)$ solving~\eqref{eq:duhamel-potential} for some small enough time $T>0$. We have to show that the condition~\eqref{eq:condition_initial_gS2} is propagated and the result will then follow from an iteration argument. So we define the operator
$$Q(t):=e^{it\Delta}\cW_V(t,0)(\gamma_f+Q_0)\cW_V(t,0)^*e^{-it\Delta}-\gamma_f,$$
and prove that $Q\in L^\ii_t((-T,T),\gS^2)$. The continuity then follows from the fact that $Q$ satisfies \eqref{eq:rHF_Q_Duhamel}.
Indeed, since $\cW_V(t,0)$ is unitary, we have
$$\norm{e^{it\Delta}\cW_V(t,0)Q_0\cW_V(t,0)^*e^{-it\Delta}}_{\gS^2}\leq \norm{Q_0}_{\gS^2}.$$
On the other hand,
\begin{multline*}
e^{it\Delta}\cW_V(t,0)\gamma_f\cW_V(t,0)^*e^{-it\Delta}-\gamma_f\\
=e^{it\Delta}\Big((\cW_V(t,0)-1)\gamma_f\cW_V(t,0)^*+\gamma_f(\cW_V(t,0)^*-1)\Big)e^{-it\Delta}
\end{multline*}
which is Hilbert-Schmidt (uniformly for $t$ in compact sets), by Lemma~\ref{lemma:Wn-gammaf-S2} and the assumptions on $f$.
We then have to show that $w\ast \rho[e^{is\Delta} Q(t^*)e^{-is\Delta}]\in L^{8/d}_{s,\rm loc}(\R,L^2_x\cap L^\ii_x)$ for any $t^*\in(-T,T)$. To do so, we recall \eqref{eq:decomposition-Q}: 
  \begin{multline*}
   Q(t^*)=e^{it^*\Delta}Q_0e^{-it^*\Delta}+e^{it^*\Delta}(\cW^{(1)}_V(t^*,0)Q_0+Q_0\cW^{(1)}_V(t^*,0)^*)e^{-it^*\Delta}\\
    +\sum_{n\ge1}e^{it^*\Delta}(\cW^{(n)}_V(t^*,0)\gamma_f+\gamma_f\cW^{(n)}_V(t^*,0)^*)e^{-it^*\Delta}+A_V(t^*).
  \end{multline*} 
  Here $A_V(t^*)\in\gS^1$ and therefore $e^{is\Delta}A_V(t^*)e^{-is\Delta}\in L^\ii_s(\R,\gS^1)$ and 
  $w*\rho[e^{is\Delta}A_V(t^*)e^{-is\Delta}]\in L^\ii_s(\R,L^2_x\cap L^\ii_x)$.
  By our assumption~\eqref{eq:condition_initial_gS2} on $Q_0$, we have 
  $w*\rho[e^{i(s+t^*)\Delta}Q_0e^{-i(s+t^*)\Delta}]\in L^{8/d}_{s,{\rm loc}}(\R,L^2_x\cap L^\ii_x)$. The other terms in the decomposition of $Q(t^*)$ are treated by duality as in the proof of Proposition~\ref{prop:local-wp-potential}. At this point we have constructed a maximal solution.

Let us now prove uniqueness. If $Q,Q'\in C^0_t([-T^-,T^+],\gS^2)$ are two solutions to \eqref{eq:rHF_Q_Duhamel} such that $V_Q$ and $V_{Q'}$ belong to $L^{8/d}_t((-T^-,T^+),L^2_x\cap L^\ii_x)$, then $Q=Q'$ will obviously follow from $V_Q=V_{Q'}$. We cannot use a Gr\"onwall-type argument due to the fact that we lack \emph{pointwise} estimates in $t$ on $\Phi(V_Q)(t)-\Phi(V_{Q'})(t)$. Instead, we use a continuity argument: let $t^*$ be the largest time $\leq T$ for which $V_{Q(t)}=V_{Q'(t)}$ for all $0\leq t\leq t^*$. From the fixed point argument of Proposition~\ref{prop:local-wp-potential}, we know that $t^*>0$. Indeed, let
\begin{multline*}
R\geq \norm{Q_0}_{\gS^2\,/\,\gS^{2,1}}+\norm{w\ast \rho_{e^{it\Delta}Q_0e^{-it\Delta}}}_{L^{8/d}_t((-1,1),L^2_x\cap L^\ii_x)}\\
+\norm{V_{Q}}_{L^{8/d}_t((-1,1),L^2_x\cap L^\ii_x)}+\norm{V_{Q'}}_{L^{8/d}_t((-1,1),L^2_x\cap L^\ii_x)}.
\end{multline*}
We have proved that there exists $t'=t'(R)>0$ small enough such that \eqref{eq:duhamel-potential} has a unique solution on the ball
  $\{\|V\|_{L^{8/d}_t(L^2_x\cap L^\ii_x))}\le 2R\},$
  Since both $V_Q$ and $V_{Q'}$ belong to this ball, we have $V_{Q}\equiv V_{Q'}$ a.e. on $[-t',t']$, and therefore $t^*>0$. We now argue by contradiction: assume that $t^*<T$. Since $V_{Q}\equiv V_{Q'}$ a.e. on $(-t^*,t^*)$, $Q\equiv Q'$ on $[-t^*,t^*]$ and in particular $Q(t^*)=Q'(t^*)$, $Q(-t^*)=Q'(-t^*)$ (we use here the continuity of $t\mapsto Q(t)$ in $\gS^2$). 
  Also, the density $\rho_{e^{is\Delta}Q(t^*)e^{-is\Delta}}$ belongs to $L^{8/d}_{s,\rm loc}(\R,L^2_x\cap L^\ii_x)$, as we have already proved. Hence we can apply Proposition~\ref{prop:local-wp-potential} again at $t^*$, leading to a contradiction. The argument is the same for negative times.

The argument for the continuity with respect to $Q_0$, $f$ and $w$ is similar. Let $(Q_0^{(n)},f_n,w_n)\to(Q_0,f,w)$, and $V_n$ the associated solutions. First assume that $V_n$ and $V$ are defined on $[0,T]$ for all $n$ and some $T>0$, and that $(V_n)$ is uniformly bounded in $L^{8/d}_t([0,T],L^2_x\cap L^\ii_x)$. Let us show that $V_n\to V$. From the fixed point technique in Proposition~\ref{prop:local-wp-potential}, we deduce that $\norm{V_n-V}_{L^{8/d}_t([0,T'],L^2_x\cap L^\ii_x)}\to0$ for $T'>0$ small enough. Note that $t\mapsto\norm{V_n-V}_{L^{8/d}_t([0,t],L^2_x\cap L^\ii_x)}$ is increasing, hence has a subsequence which converges everywhere, to a non-decreasing function. Let now $t^*$ be the first time such that the limit fails to be 0 and assume that $t^*<T$. We choose $\epsilon>0$ small enough such that the time of existence for the problem in the ball of radius $2\norm{(Q(t^*-\epsilon),f,w)}$ is $>2\epsilon$. Since for $n$ large enough, $(Q_n(t^*-\epsilon),f_n,w_n)$ falls into this ball, we get again a contradiction. Hence, $V_n\to V$ in $L^{8/d}_t([0,T],L^2_x\cap L^\ii_x)$, from which we deduce that $Q_n\to Q$ in $C^0_t([0,T],\gS^2)$ and that $w*\rho_{e^{it\Delta}Q_n(s)e^{-it\Delta}}\to w*\rho_{e^{it\Delta}Q(s)e^{-it\Delta}}$ in $L^{8/d}_{t,{\rm loc}}(\R,L^2_x\cap L^\ii_x)$ for all $s\in[0,T]$. Now let $0<T<T^+$, and let 
\begin{multline*}
  \frac{R}2=\sup_{s\in[0,T]}\norm{Q(s)}_{\gS^2}+\norm{V}_{L^{8/d}_t([0,T],L^2_x\cap L^\ii_x)}+\norm{w}_{L^1_x\cap L^\ii_x}+\norm{f}_{L^\ii}\\
+\int_{\R^d}|f(k^2)|\,dk+\sup_{s\in[0,T]}\norm{w*\rho_{e^{it\Delta}Q(s)e^{-it\Delta}}}_{L^{8/d}_t([0,T],L^2_x\cap L^\ii_x)},
\end{multline*}
which is finite by the estimates provided before. By Proposition \ref{prop:local-wp-potential}, for $n$ large enough, $V_n$ and $V$ are defined up to time $0<T(R)\le T$ and uniformly bounded on $[0,T(R)]$. By the previous argument, this implies that 
$V_n\to V$, $Q_n\to Q$, and $w*\rho_{e^{it\Delta}Q_n(s)e^{-it\Delta}}\to w*\rho_{e^{it\Delta}Q(s)e^{-it\Delta}}$ on $[0,T(R)]$. In particular, for $n$ large enough 
$$\|Q_n(T(R))\|_{\gS^2}+\|w*\rho_{e^{it\Delta}Q_n(T(R))e^{-it\Delta}}\|_{L^{8/d}_t([0,T(R)],L^2_x\cap L^\ii_x)}\le R$$
and we can again extend the solutions up to $[0,2T(R)]$. Repeating this procedure a finite number of times, we can cover the whole $[0,T]$. This ends the proof of Theorem~\ref{thm:local-wp-S2}.
\end{proof}

We finally turn to the 

\begin{proof}[Proof of Theorem~\ref{thm:local-wp-S21}]
The proof is the same as for Theorem~\ref{thm:local-wp-S2}, with some minor differences. First we remark that, by~\eqref{eq:density_in_L1L2}, the solution $V(t)$ constructed in Proposition~\ref{prop:local-wp-potential} is automatically in $L^{8/d}((-T,T),H^1_x\cap W^{1,\ii}_x)$ when $\nabla w\in L^1\cap L^\ii $ and when $\rho_{e^{it\Delta}Q_0e^{-it\Delta}}\in L^{8/d}_{t,\rm loc}(H^1_x\cap W^{1,\ii}_x)$. So, in order to mimic the proof of Theorem~\ref{thm:local-wp-S2}, we only have to show that $Q\in C^0_t(\gS^{2,1})$. First, by \eqref{eq:est-V-Hs-Hs}, we have 
  $\|(1-\Delta)^{1/2}V(1-\Delta)^{-1/2}\|\le C\|V\|_{W^{1,\ii}},$
  hence 
  $$\|(1-\Delta)^{1/2}\cW_V(t,0)(1-\Delta)^{-1/2}\|\le C\sum_{n\ge0}\frac{\|V\|_{L^1_tW^{1,\ii}_x}^n}{n!},$$
which is finite when $V\in  L^{8/d}_{t,\rm loc}(W^{1,\ii}_x)$. This implies that
  $$\|e^{it\Delta}\cW_V(t)Q_0\cW_V(t)^*e^{-it\Delta}\|_{\gS^{2,1}}\le C\left(\sum_{n\ge0}\frac{\|V\|_{L^1_tW^{1,\ii}_x}^n}{n!}\right)\|Q_0\|_{\gS^{2,1}}.$$
  Using that $\|\gamma_f V (1-\Delta)^{1/2}\|_{\gS^{2}}\le C\|V\|_{H^1_x}$ as we have proved in Lemma~\ref{lemma:commutator-Vgamma-Hs}, we infer that, for $n\ge1$,
  \begin{align*}
    \|\gamma_f\cW^{(n)}_V(0,t)(1-\Delta)^{1/2}\|_{\gS^2} &\le C\frac{1}{(n-1)!}\|V\|_{L^1_tW^{1,\ii}_x}^{n-1}\|\gamma_f V(1-\Delta)^{1/2}\|_{L^1_t\gS^2}\\
    &\le C\frac{1}{(n-1)!}\|V\|_{L^1_tW^{1,\ii}_x}^{n-1}\|V\|_{L^1_t H^1_x}.
  \end{align*}
  Hence, we have $e^{it\Delta}\cW_V(t)\gamma_f\cW_V(t)^*e^{-it\Delta}-\gamma_f \in L^\ii_{t}((-T,T),\gS^{2,1})$, which finishes the proof of Theorem~\ref{thm:local-wp-S21}. 
\end{proof}

\section{Inequalities on the relative entropy}\label{sec:entropy}

The purpose of this section is to provide the necessary tools to prove global well-posedness for a large class of functions $f$'s, including the gases at positive temperature~\eqref{eq:Fermi-gas-positive-temp}--\eqref{eq:Boltzmann-gas-positive-temp}. We use a general concept of relative entropy $\cH(\gamma,\gamma_f)$ of two one-particle density matrices $\gamma$ and $\gamma_f$ which has been recently introduced in~\cite{LewSab-13}. The purpose of this section is to prove some useful lower bounds on this relative entropy. We particularly discuss \emph{Lieb-Thirring inequalities} in the spirit of~\cite{FraLewLieSei-11,FraLewLieSei-12}, which are important to get bounds on the density $\rho_\gamma$, as we have already seen in the zero temperature case.

\subsection{Generalized relative entropy}

Before turning to the general definition of the relative entropy, let us recall the usual case of the von Neumann entropy. To simplify our exposition, we start by working in a finite-dimensional space $\gK$. The entropy of the system takes the form
$$\boxed{\cS(\gamma)=\tr_\gK \big(S(\gamma)\big)}$$
where
\begin{equation}
S(x)=
\begin{cases}
-x\log(x)-(1-x)\log(1-x),\quad x\in[0,1], &\text{(fermions),}\\
-x\log(x)+(1+x)\log(1+x),\quad x\in[0,M], &\text{(bosons),}\\
-x\log(x)+x,\quad x\in[0,M], &\text{(boltzons).}
\end{cases}
\label{eq:s_for_bosons_fermions}
\end{equation}
We always work with bounded operators and $M>0$ could in principle be any positive number. Later we will restrict ourselves to $M=1$ for simplicity.
The function $S$ is concave in these cases, hence so is $\cS$ (see, e.g.,~\cite[Chap. 3]{OhyPet-93}). It is important to remember that $\gamma$ is the \emph{one-particle density matrix} of the system, and not the state itself. Both for fermions and bosons, there is a unique quasi-free state associated with $\gamma$, which is an operator $D\geq0$ acting on the Fock space ${\rm d}\Gamma(\gK)$, such that $\tr_{{\rm d}\Gamma(\gK)}(D)=1$ and which has $\gamma$ as one-particle density matrix~\cite{BacLieSol-94}. Then, the usual von Neumann entropy of $D$ coincides with the above formulas: $-\tr_{{\rm d}\Gamma(\gK)}\big( D\log(D)\big)=\cS(\gamma)=\tr_\gK \big(S(\gamma)\big)$. The situation is similar for boltzons. 

Consider now any self-adjoint operator $h$ on $\gK$. The corresponding density matrices are obtained by minimizing the free energy, which is by definition the difference between the energy and the entropy:
$$\gamma_f=\underset{0\leq \gamma=\gamma^*\leq M}{\text{argmin}} \Big\{ \tr_\gK \big(h\gamma-S(\gamma)\big)\Big\}.$$
The optimal state is
$$\gamma_f=\begin{cases}
\dps\frac{1}{e^{h}+1}&\text{for fermions,}\\[0.2cm]
\dps\frac{1}{e^{h}-1}&\text{for bosons,}\\[0.2cm]
\dps e^{-h}&\text{for boltzons.}
\end{cases}$$
In the bosonic case, it is necessary to have $h\geq \log(1+1/M)$. Normally we should multiply $S$ by the temperature $T$, in which case we get $h/T$ instead of $h$ in the previous formulas, but we can always think of including $T$ in the definition of $h$.

Now that we have found the minimizer of the free energy, it is useful to look at the difference between the free energy of any state $\gamma$ and that of the minimizer $\gamma_f$. This difference is precisely the \emph{relative entropy}
\begin{equation}
\cH(\gamma,\gamma_f):=\tr_\gK \big(h(\gamma-\gamma_f)-S(\gamma)+S(\gamma_f)\big).
\label{eq:def_relative_entropy_usual}
\end{equation}
That the relative entropy $\cH(\gamma,\gamma_f)$ is non-negative is an obvious consequence of the fact that $\gamma_f$ minimizes the free energy. The relevance of $\cH(\gamma,\gamma_f)$ comes from the fact that it can make sense for infinite systems, even if the three terms in~\eqref{eq:def_relative_entropy_usual} do not make sense separately.

Here we have prescribed the Hamiltonian $h$ and we got the Gibbs state $\gamma_f$. But we can turn the matter around and find $h$ from $\gamma_f$. Indeed, writing that the derivative of the free energy vanishes at the point $\gamma_f$, we find that $h=S'(\gamma_f)$. Therefore, we can re-write the relative entropy as
\begin{equation}
\boxed{\cH(\gamma,\gamma_f):=-\tr_\gK \Big(S(\gamma)-S(\gamma_f)-S'(\gamma_f)(\gamma-\gamma_f)\Big).}
\label{eq:def_relative_entropy}
\end{equation}
In this formula, $S$ is given by~\eqref{eq:s_for_bosons_fermions}. In~\cite{LewSab-13}, the relative entropy~\eqref{eq:def_relative_entropy} was studied for $S$ a general concave function, and several properties were proved. First by Klein's lemma (see~\cite[Prop. 3.16]{OhyPet-93} and~\cite[Lem. 1]{LewSab-13}), it is known that $\cH(\gamma,\gamma_f)\geq0$ for any $0\leq\gamma,\gamma_f\leq1$, on a finite-dimensional space $\gK$. The main contribution of~\cite{LewSab-13} was to prove that $\cH$ is monotone, that is
$$\cH(PAP,PBP)\leq \cH(A,B)$$
for every orthogonal projection, if and only if $-S'$ is operator-monotone. Monotonicity is a natural concept which allows to define the relative entropy in infinite dimension as is always done in statistical mechanics, that is, by a thermodynamic limit. In mathematical terms, the proper definition in infinite dimension is
\begin{equation}
\cH(\gamma,\gamma_f):=\lim_{k\to\ii} \cH(P_k\gamma P_k,P_k\gamma_fP_k)
\label{eq:def_infinite_dimension}
\end{equation}
where $P_k$ is any increasing sequence of finite-rank orthogonal projections such that $P_k\to1$ strongly. The limit always exist in $\R^+\cup\{+\ii\}$ since the right side is monotone, and it can be shown to be independent of the chosen sequence $(P_k)$, see~\cite[Thm. 2]{LewSab-13}.

In this paper we always assume that 
\begin{equation}
S:[0,1]\to\R^+ \ \text{\it is continuous on $[0,1]$}
\label{eq:hyp_s_1}
\end{equation}
and that
\begin{equation}
\text{\it $-S'$ is operator monotone on $(0,1)$ and not constant.} 
\label{eq:hyp_s_2}
\end{equation}
Recall that an operator monotone function on an interval $(a,b)$ is always $C^\ii$ (see~\cite[Chap. V]{Bhatia}). Here we work on the interval $[0,1]$ for simplicity but, by a simple change of variable,  the results are all exactly the same on an interval $[0,M]$ with $M<\ii$. Of course, in the three usual physical cases~\eqref{eq:s_for_bosons_fermions}, the assumptions~\eqref{eq:hyp_s_1} and~\eqref{eq:hyp_s_2} are fulfilled (only for $\mu<-\log(2)$ in the bosonic case). Other examples of functions $S$ satisfying~\eqref{eq:hyp_s_1} and~\eqref{eq:hyp_s_2} are provided in~\cite{AudHiaPet-10} and include
$$S(x)=\begin{cases}
-(t+x)\log(t+x),&t\geq0\\
\log(t+x),&t>0,\\
x^m,&0<m\leq 1,\\
-x^m,&1\leq m\leq 2.\\
\end{cases}$$

\begin{remark}\label{rmk:Lowner}
Any function $S$ satisfying~\eqref{eq:hyp_s_1} and~\eqref{eq:hyp_s_2} can be written as
\begin{multline}
S(x)=a'x+c'+b\int_{0}^\ii\left(\log(t+x)-\log(t+1/2)-\frac{2x-1}{2t+1}\right)\,d\nu_1(t)\\
+b\int_{0}^\ii\left(\log(t+1-x)-\log(t+1/2)+\frac{2x-1}{2t+1}\right)\,d\nu_2(t)
\label{eq:Lowner2}
\end{multline}
with $\nu_1$ and $\nu_2$ two positive Borel measures satisfying $2\int_{0}^\ii (2t+1)^{-2}(d\nu_1+d\nu_2)(t)=1$ (see~\cite{LewSab-13}). The constraint that $S$ admits limits at $0$ and $1$ means that $-\int_0^1\log(t)(d\nu_1+d\nu_2)(t)<\ii$ and it implies $\nu_1(\{0\})=\nu_2(\{1\})=0$, as well as 
\begin{equation}
\lim_{x\to0^+}xS'(x)= \lim_{x\to1^-}(1-x)S'(x)=0.
\label{eq:limit_end_points}
\end{equation}
The usual Fermi-Dirac distribution is obtained by taking $d\nu_1=d\nu_2=dt/2$, $a'=0=c'$, $b=2$.
\end{remark}

\subsection{Klein inequality}\label{sec:Klein}

From now on we work again in the Hilbert space $\gH=L^2(\R^d)$ and we consider a reference state of the form $\gamma_f=f(-\Delta)$ coming from an entropy $S$ as before, which requires $f=(S')^{-1}$ (the inverse function of $x\in(0,1)\mapsto S'(x)$). In order to have $0\leq \gamma_f\leq1$, we need that $S'((0,1))\supset(0,\ii)$. Because we assume that $-S'$ is operator monotone, $S'$ is decreasing and therefore we need to have
\begin{equation}
\lim_{x\to0^+}S'(x)=+\ii\quad\text{and}\quad \lim_{x\to1^-}S'(x)\leq0. 
\label{eq:hyp_s_3}
\end{equation}
We think of including the chemical potential $\mu$ in the function $f$. Of course, replacing $f(-\Delta)$ by $f(-\Delta-\mu)$ amounts to replacing $S(x)$ by $S(x)+\mu x$. We see that, by choosing $\mu$ sufficiently negative, we can always ensure that $S'(1)\leq0$. In the case of the Bose-Einstein entropy, any $\mu\leq -\log(2)$ will do, if we work on $[0,1]$ with a temperature $T=1$. In the Fermi-Dirac case, we have $S'(x)\to-\ii$ when $x\to1$ and one can take any $\mu\in\R$.

The first important information that we have when $\cH(\gamma,\gamma_f)<\ii$, is that  $\gamma-\gamma_f\in\gS^{2,1}$. This is stated in the following result, inspired of~\cite[Thm 1]{HaiLewSei-08} and~\cite[Lemma 1]{FraHaiSeiSol-12}.

\begin{lemma}[Klein inequality]\label{lem:Klein}
We assume that $d\geq1$ and that $S$ satisfies~\eqref{eq:hyp_s_1},~\eqref{eq:hyp_s_2},~\eqref{eq:hyp_s_3}. Let $\gamma_f=f(-\Delta)$ with $f=(S')^{-1}$. Then there exists a constant $C>0$ (depending only on $S$) such that 
\begin{equation}
\cH(\gamma,\gamma_f)\geq C\tr\,(1-\Delta)(\gamma-\gamma_f)^2,
\label{eq:Klein2}
\end{equation}
for all $0\leq\gamma\leq1$.
\end{lemma}

\begin{proof}
It was proved in~\cite[Thm. 3]{LewSab-13} that $\cH(A,B)\geq C\,\tr\,\big(1+|S'(B)|\big)(A-B)^2$ for all $A,B$. Since $-\Delta=S'(\gamma_f)$, the result follows.
\end{proof}

The result can be used to approximate $A$ by a finite-rank perturbation of $B$ in a suitable sense, see~\cite[Cor. 1]{LewSab-13}.

\subsection{Lieb-Thirring inequality}\label{sec:Lieb-Thirring}

We have seen in the previous section that $\cH(\gamma,\gamma_f)$ controls the $\gS^{2,1}$ norm. Except in dimension $d=1$, this does not tell us anything about the density $\rho_{\gamma-\gamma_f}$. Here we discuss Lieb-Thirring inequalities in the spirit of~\cite{FraLewLieSei-11,FraLewLieSei-12}, which give an information of this type.

First, in order that $\gamma_f$ has a finite density $\rho_{\gamma_f}$, we need
\begin{equation*}
\int_{\R^d}f(|p|^2)\,dp =\frac{|S^{d-1}|}{2}\int_0^{f(0)}\lambda |S''(\lambda)|\,S'(\lambda)^{\tfrac{d}2-1}\,d\lambda<\ii.
\end{equation*}
By using L\"owner's formula~\eqref{eq:Lowner2}, this is the same as 
\begin{equation}
\int_0^1 S'(\lambda)_+^{\tfrac{d}2}\,d\lambda<\ii.
\label{eq:integrability_condition} 
\end{equation}
The integrability condition~\eqref{eq:integrability_condition} is satisfied for the Fermi-Dirac and Bose-Einstein distributions, since $S'$ diverges logarithmically at 0 in those cases. It is simple to verify that~\eqref{eq:integrability_condition} is satisfied for $S(x)=x^{m}$ with $0<m<1$ in dimension $d\leq2$ and $1-2/d< m<1$ in dimension $d\geq3$.

By using Klein's inequality~\eqref{eq:Klein2}, it is easy to verify that, in dimension $d=1$, the relative density $\rho_\gamma-\rho_{\gamma_f}$ is in $L^q(\R)$ for all $2\leq q<\ii$. The following is very similar to Lemma~\ref{lem:density_Schatten-Sobolev}. 

\begin{lemma}[Density in 1D]\label{lem:density_1D}
We assume that $d=1$ and that $S$ satisfies~\eqref{eq:hyp_s_1},~\eqref{eq:hyp_s_2},~\eqref{eq:hyp_s_3} and~\eqref{eq:integrability_condition}. Let $\gamma_f=f(-\Delta)$ with $f=(S')^{-1}$. Then we have 
\begin{equation}
\cH(\gamma,\gamma_f)\geq \begin{cases}
\dps K_{\rm LT,S}\int_{\R}|\rho_{\gamma}-\rho_{\gamma_f}|^q&\text{for $2\leq q<3$}\\
\dps K_{\rm LT,S}\left(\int_{\R}|\rho_{\gamma}-\rho_{\gamma_f}|^q\right)^{\frac{2}{q-1}-\epsilon}&\text{for $3\leq q<\ii$ and $\epsilon>0$.}\\
\end{cases}
\label{eq:density_1D} 
\end{equation}
for all $0\leq\gamma\leq1$, with constants $K_{\rm LT,S}$ which only depend on $S$, $q$ and $\epsilon$.
\end{lemma}

\begin{proof}[Proof of Lemma \ref{lem:density_1D}]
By~\cite[Cor. 1]{LewSab-13}, it suffices to prove the result for $Q:=\gamma-\gamma_f$ a smooth finite rank operator. We write 
$$\int_\R \rho_Q V=\tr(QV)=\tr\left((1-\Delta)^{\frac{1}{2q}}Q(1-\Delta)^{\frac{1}{2q}}\frac{1}{(1-\Delta)^{\frac{1}{2q}}}V\frac{1}{(1-\Delta)^{\frac{1}{2q}}}\right)$$
with $2\leq q<3$. Similarly as in Lemma~\ref{lem:density_Schatten-Sobolev}, we have by the Araki-Lieb-Thirring inequality 
$\|(1-\Delta)^{1/(2q)}|Q|^{1/q}\|_{\gS^{2q}}\leq \|(1-\Delta)^{1/2}Q\|_{\gS^2}^{1/q}$.
On the other hand, since $\norm{Q}\leq1$ we have $\norm{|Q|^{1-2/q}}\leq1$ as well, for $q\geq2$. We then obtain, from the KSS  inequality~\eqref{eq:KSS},
$$\int_\R \rho_Q V\leq (2\pi)^{1/q-1}\left(\tr(1-\Delta)Q^2\right)^{1/q}\norm{V}_{L^{q'}(\R)}\left(\int_\R\frac{dp}{(1+p^2)^{\frac{1}{q-1}}}\right)^{1/q'}.$$
The integral on the right is finite for $q<3$. This ends the proof of the first inequality, by duality and by~\eqref{eq:Klein2}.
When $q\geq3$, we write, instead, 
\begin{equation*}
\int_\R \rho_Q V
=\tr\left((1-\Delta)^{\frac{1}{4q'}+\epsilon}Q(1-\Delta)^{\frac{1}{4q'}+\epsilon}\frac{1}{(1-\Delta)^{\frac{1}{4q'}+\epsilon}}V\frac{1}{(1-\Delta)^{\frac{1}{4q'}+\epsilon}}\right)
\end{equation*}
and get, similarly as before,
\begin{align*}
\norm{\rho_Q}_{L^q(\R)}&\leq C\norm{(1-\Delta)^{\frac{1}{4q'}+\epsilon}Q(1-\Delta)^{\frac{1}{4q'}+\epsilon}}_{\gS^q}\\
&\leq C\norm{(1-\Delta)^{\frac{1}{4q'}+\epsilon}Q(1-\Delta)^{\frac{1}{4q'}+\epsilon}}_{\gS^{\frac{2q'}{1+4\epsilon q'}}}\\
&\leq C\left(\tr(1-\Delta)Q^2\right)^{\frac{1+4\epsilon q'}{2q'}}.
\end{align*}
We have used here that $2q'\leq q$ for $q\geq3$.
\end{proof}

In higher dimensions the previous argument does not work any more, since $(1+|p|^2)^{-1}$ is not integrable. However, we can derive a similar bound by using the Lieb-Thirring inequality of~\cite{FraLewLieSei-11,FraLewLieSei-12}.

\begin{theorem}[Lieb-Thirring inequality]\label{thm:Lieb-Thirring}
We assume that $d\geq2$ and that $S$ satisfies~\eqref{eq:hyp_s_1},~\eqref{eq:hyp_s_2},~\eqref{eq:hyp_s_3} and~\eqref{eq:integrability_condition}. Then, with $\gamma_f=f(-\Delta)$ and $f=(S')^{-1}$, we have the Lieb-Thirring inequality
\begin{multline}
\cH(\gamma,\gamma_f)
\;\geq\;
K_{\rm LT,S}\int_{\R^d}\bigg(\big(\rho_{\gamma_f}+\rho_{\gamma-\gamma_f}(x)\big)^{1+\tfrac{2}d}-(\rho_{\gamma_f})^{
1+\tfrac{2}d}\\
-\frac{2+d}d(\rho_{\gamma_f})^{\tfrac{2}d}\,\rho_{\gamma-\gamma_f}(x)\bigg)dx
\label{eq:Lieb-Thirring}
\end{multline}
for all $0\leq\gamma\leq1$ and with a constant $K_{\rm LT,S}$ depending only on $S$ and $d$. In particular, $\rho_{\gamma-\gamma_f}\in L^2(\R^d)+L^{1+2/d}(\R^d)$ when $\cH(\gamma,\gamma_f)<\ii$.
\end{theorem}

The inequality dual to~\eqref{eq:Lieb-Thirring} (that is, expressed in terms of a potential $V(x)$) has been derived in~\cite[Thm. 5.4]{FraLewLieSei-12}. Instead of getting~\eqref{eq:Lieb-Thirring} from this result, we prove~\eqref{eq:Lieb-Thirring} directly using the zero temperature inequality~\eqref{eq:Fermi-gas-zero-temp}. A similar approach is used in~\cite{FraLewLieSei-12}.

Since $\rho_\gamma-\rho_{\gamma_f}$ belongs to $L^2(\R^d)+L^{1+2/d}(\R^d)$ when $\cH(\gamma,\gamma_f)<\ii$ and $\rho_{\gamma_f}$ is a constant by~\eqref{eq:integrability_condition}, we deduce that $\rho_\gamma\in L^1_{\rm loc}(\R^d)$, hence that $\gamma$ is locally trace-class. 

\begin{proof}
By~\cite[Cor. 1]{LewSab-13}, it suffices to prove the theorem for $\gamma$ a smooth finite rank perturbation of $\gamma_f=f(-\Delta)$, and we make this assumption throughout the proof. Let us start by deriving a useful integral representation formula for $\cH(A,B)$, when $A$ and $B$ are finite matrices. Inserting 
$$S(x)=S(0)+\int_{0}^1S'(\lambda) \1(x\geq\lambda)\,d\lambda,$$
in the definition of $\cH(A,B)$, we obtain 
$$\cH(A,B)=\int_0^1 \tr\big(S'(B)-S'(\lambda)\big)\big(\1(A\geq\lambda)-\1(B\geq\lambda)\big)\,d\lambda.$$
We remark that $\1(B\geq\lambda)=\1(S'(B)\leq S'(\lambda))$, since $S'$ is a decreasing function. So we obtain
\begin{align*}
&\tr\big(S'(B)-S'(\lambda)\big)\big(\1(A\geq\lambda)-\1(B\geq\lambda)\big)\\
&\qquad=-\tr\big|S'(B)-S'(\lambda)\big|\1(B\geq\lambda)\big(\1(A\geq\lambda)-\1(B\geq\lambda)\big)\1(B\geq\lambda)\\
&\qquad\qquad +\tr\big|S'(B)-S'(\lambda)\big|\1(B\leq\lambda)\big(\1(A\geq\lambda)-\1(B\geq\lambda)\big)\1(B\leq\lambda).
\end{align*}
A simple calculation shows that
\begin{multline*}
\big(\1(A\geq\lambda)-\1(B\geq\lambda)\big)^2=
-\1(B\geq\lambda)\big(\1(A\geq\lambda)-\1(B\geq\lambda)\big)\1(B\geq\lambda)\\
+\1(B\leq\lambda)\big(\1(A\geq\lambda)-\1(B\geq\lambda)\big)\1(B\leq\lambda)
\end{multline*}
and therefore we arrive at
$$\cH(A,B)=\int_0^1 \tr\big|S'(B)-S'(\lambda)\big|\big(\1(A\geq\lambda)-\1(B\geq\lambda)\big)^2\,d\lambda.$$
The formula is still true in infinite dimension, but proving it requires some work. In our particular case we only need the lower bound
\begin{equation}
\cH(\gamma,\gamma_f)\geq \int_0^1 \tr\big|S'(\gamma_f)-S'(\lambda)\big|\big(\1(\gamma\geq\lambda)-\1(\gamma_f\geq\lambda)\big)^2\,d\lambda.
\label{eq:integral_representation_lower_bound} 
\end{equation}
To prove~\eqref{eq:integral_representation_lower_bound}, we take an increasing sequence of finite-dimensional projections $P_k$, with $P_k\to1$ strongly in $H^1(\R^d)$. We can use the formula in finite dimension and that $\cH(P_k\gamma P_k,P_k\gamma_fP_k)\nearrow \cH(\gamma,\gamma_f)$ when $k\to\ii$, by the definition of $\cH$ in~\eqref{eq:def_infinite_dimension}. The right term is slightly more complicated to deal with. First, $|S'(P_k\gamma_fP_k)-S'(\lambda)|^{1/2}\to |S'(\gamma_f)-S'(\lambda)|^{1/2}$ and $\1(P_k\gamma_fP_k\geq\lambda)\to \1(\gamma_f\geq\lambda)$ strongly in $H^1(\R^d)$. By~\cite[Thm  VIII.24]{ReeSim1}, this follows from the fact that, since $S'$ is strictly decreasing on $(0,1)$, $\lambda$ can never be an eigenvalue of $\gamma_f$. Similarly, we have that $\1(P_k\gamma P_k\geq\lambda)\to \1(\gamma\geq\lambda)$ strongly if $\lambda$ is not an eigenvalue of $\gamma$, and in this case we get 
\begin{multline}
\liminf_{k\to\ii} \tr\big|S'(P_k\gamma_fP_k)-S'(\lambda)\big|\big(\1(P_k\gamma P_k\geq\lambda)-\1(P_k\gamma_fP_k\geq\lambda)\big)^2\\
\geq \tr\big|S'(\gamma_f)-S'(\lambda)\big|\big(\1(\gamma\geq\lambda)-\1(\gamma_f\geq\lambda)\big)^2
\label{eq:Fatou_LT}
\end{multline}
by Fatou's lemma for operators. The strong convergence of $\1(P_k\gamma P_k\geq\lambda)$ can however fail if $\lambda$ happens to be an eigenvalue of $\gamma$. Fortunately, the set of eigenvalues is at most countable and it is therefore of Lebesgue measure zero. Applying Fatou's theorem for the integral over $\lambda$, this ends the proof of the lower bound~\eqref{eq:integral_representation_lower_bound}. 

Let us estimate the terms in the integral for every $\lambda$. In dimension $d\geq2$, we use the Lieb-Thirring inequality~\eqref{eq:Lieb-Thirring-zero-temp}
and obtain
\begin{multline}
\cH(\gamma,\gamma_f)\geq K(d)\int_{\R^d}dx\int_0^{f(0)}S'(\lambda)^{1+\tfrac{d}2}F\left(\frac{\rho_\lambda(x)}{S'(\lambda)^{\tfrac{d}2}}\right)d\lambda\\
+K(d)\int_{\R^d}dx\int_{f(0)}^1 \rho_\lambda(x)^{1+\tfrac{2}d}d\lambda
\label{eq:estim_LT_1}
\end{multline}
with $\rho_\lambda:=\rho_{\1(\gamma\geq\lambda)}-\rho_{\1(-\Delta\leq S'(\lambda))}$ and
$F(t):=(c_d+t)_+^{1+{2}/d}-c_d^{1+{2}/d}-\frac{2+d}dc_d^{{2}/d}\,t$ and where $c_d=\rho_{\1(-\Delta\leq 1)}$.
The function $F(t)$ is controlled from above and below by $G(t):=\min(t^2,|t|^{1+2/d})$, for $t\geq-c_d$. The estimates are a bit simpler with $G$ and we shall use it in the rest of the proof. So we can bound from below the first integral in~\eqref{eq:estim_LT_1} by a constant times
$$\int_0^{f(0)}\!\!\!S'(\lambda)^{1+\tfrac{d}2}\!\left(\frac{\rho_\lambda(x)^2}{S'(\lambda)^d}\1(\rho_\lambda(x)\leq S'(\lambda)^{\frac{d}2})+\frac{\rho_\lambda(x)^{1+\tfrac{2}d}}{S'(\lambda)^{1+\tfrac{d}2}}\1(\rho_\lambda(x)\geq S'(\lambda)^{\frac{d}2})\!\right)\!d\lambda.$$
By the Cauchy-Schwarz inequality, we have
\begin{multline*}
\int_0^{f(0)}\1\big(\rho_\lambda(x)\leq S'(\lambda)^{\frac{d}2}\big)S'(\lambda)^{1-\tfrac{d}2}\rho_\lambda(x)^2\,d\lambda\\
\geq\frac{\dps\left(\int_0^{f(0)}\rho_\lambda(x)\1(\rho_\lambda(x)\leq S'(\lambda)^{\frac{d}2})d\lambda\right)^2}{\dps\int_0^{f(0)}S'(\lambda)^{\tfrac{d}2-1}d\lambda}
\end{multline*}
where the denominator is finite by our assumption~\eqref{eq:integrability_condition} and since $d\geq2$. On the other hand, by Jensen's inequality we have
\begin{multline*}
\int_0^{f(0)}\1\big(\rho_\lambda(x)\geq S'(\lambda)^{\frac{d}2}\big)\rho_\lambda(x)^{1+\tfrac{2}d}\,d\lambda\\
\geq f(0)^{-\tfrac{2}d}\left(\int_0^{f(0)}\rho_\lambda(x)\1(\rho_\lambda(x)\geq S'(\lambda)^{\frac{d}2})\,d\lambda\right)^{1+\tfrac{2}d}. 
\end{multline*}
The integral with $f(0)\leq\lambda\leq1$ is similar. 
Let us now denote 
$$\rho_1(x)=\int_{f(0)}^1\rho_\lambda(x)\,d\lambda,\qquad \rho_2(x):=\int_0^{f(0)}\rho_\lambda(x)\1(\rho_\lambda(x)\leq S'(\lambda)^{\frac{d}2})\,d\lambda,$$
and
$$\rho_3(x):=\int_0^{f(0)}\rho_\lambda(x)\1(\rho_\lambda(x)\geq S'(\lambda)^{\frac{d}2})\,d\lambda,$$
which are such that $\rho_1(x)+\rho_2(x)+\rho_3(x)=\int_0^1\rho_\lambda(x) \,d\lambda=\rho_{\gamma}(x)-\rho_{\gamma_f}(x)$.
Then we have proved that
$$\cH(\gamma,\gamma_f)\geq c\int_{\R^d}\left(\rho_1(x)^{1+\tfrac{2}d}+\rho_2(x)^{2}+\rho_3(x)^{1+\tfrac{2}d}\right)\,dx.$$
Note that, by convexity of $F$,
\begin{align*}
\rho_1(x)^{1+\tfrac{2}d}+\rho_2(x)^{2}+\rho_3(x)^{1+\tfrac{2}d}&\geq G(\rho_1(x))+G(\rho_2(x))+G(\rho_3(x))\\
&\geq c\Big(F(\rho_1(x))+F(\rho_2(x))+F(\rho_3(x))\Big)\\
&\geq 3c\,F\left(\frac{\rho_1(x)+\rho_2(x)+\rho_3(x)}{3}\right)\\
&= 3c\,F\left(\frac{\rho_{\gamma-\gamma_f}(x)}{3}\right).
\end{align*}
The last term can be estimated by a constant times the one appearing on the left of~\eqref{eq:Lieb-Thirring} and this concludes the proof of Theorem~\ref{thm:Lieb-Thirring}.
\end{proof}

\subsection{High momentum estimate and consequences}\label{sec:weak-CV}

In this section we prove a uniform estimate on the high-momentum density of a matrix $\gamma$ such that $\cH(\gamma,\gamma_f)<\ii$. 

\begin{theorem}[High momentum estimate]\label{thm:high_momentum_estimate}
We assume that $d\geq2$ and that $S$ satisfies~\eqref{eq:hyp_s_1},~\eqref{eq:hyp_s_2},~\eqref{eq:hyp_s_3} and~\eqref{eq:integrability_condition}. We denote as before $\gamma_f=f(-\Delta)$ with $f=(S')^{-1}$. Let $1\leq r\leq 1+2/d$ and $0\leq \theta\leq 1$ be such that $r^{-1}=(1-\theta)d/(d+2)+\theta$. Then 
\begin{multline}
\norm{\rho_{\Pi_A^+Q\Pi_A^+}}_{L^r+L^2}\leq C\frac{\cH(\gamma_f+Q,\gamma_f)^{\frac1r}}{A^\theta}\\
+C\cH(\gamma_f+Q,\gamma_f)^{1/2}\left(\int_{0}^{f(A/2)}|S'(\lambda)|^{\frac{d}{2}-1}\,d\lambda\right)^{1/2},
\label{eq:high_momentum_estimate}
\end{multline}
and
\begin{multline}
\norm{\rho_{\Pi_A^+Q\Pi_A^-}}_{L^r+L^2}+\norm{\rho_{\Pi_A^-Q\Pi_A^+}}_{L^r+L^2}\leq C\frac{\cH(\gamma_f+Q,\gamma_f)^{\frac1r}+\cH(\gamma_f+Q,\gamma_f)^{\frac12}}{A^{\frac\theta2}}\\
+C\cH(\gamma_f+Q,\gamma_f)^{1/2}\left(\int_{0}^{f(A/2)}|S'(\lambda)|^{\frac{d}{2}-1}\,d\lambda\right)^{1/2},
\label{eq:high_momentum_estimate_off_diagonal}
\end{multline}
for all $A\geq C$ and a constant $C$ depending only on $S$ and $r$.
\end{theorem}

When $r<1+2/d$, the terms on the right side are small for large $A$, since $S$ satisfies the integrability condition~\eqref{eq:integrability_condition} and since $f(A/2)\to0$ when $A\to\ii$. 
Our bound is not small for $r=1+2/d$, however.

An estimate similar to~\eqref{eq:high_momentum_estimate} and~\eqref{eq:high_momentum_estimate_off_diagonal} can easily be derived in dimension $d=1$ by arguing as in Lemma~\ref{lem:density_1D}:
\begin{equation}
\norm{\Pi_A^+Q\Pi_A^+}_{L^2(\R)}+\norm{\Pi_A^+Q\Pi_A^-}_{L^2(\R)}+\norm{\Pi_A^-Q\Pi_A^+}_{L^2(\R)}\leq \frac{C}{\sqrt{A}}\cH(\gamma_f+Q,\gamma_f)^{1/2}.
\label{eq:high_momentum_estimate_1D}
\end{equation}

Before turning to the proof of Theorem~\ref{thm:high_momentum_estimate}, we mention a corollary. It concerns the approximation of operators with finite relative entropy by means of operators in $\gS^{2,s}$ for a large $s$. It is similar to~\cite[Cor. 1]{LewSab-13} but includes in addition an information about the density $\rho$.

\begin{corollary}[Approximation by operators in $\gS^{2,s}$]\label{cor:density_gS2s_rho}
Let $d\geq1$. We assume that $S$ satisfies~\eqref{eq:hyp_s_1},~\eqref{eq:hyp_s_2},~\eqref{eq:hyp_s_3} and~\eqref{eq:integrability_condition}. We denote as before $\gamma_f=f(-\Delta)$ with $f=(S')^{-1}$. Let $0\leq\gamma\leq 1$ be such that $\cH(\gamma,\gamma_f)<\ii$, and let $Q=\gamma-\gamma_f$ and $\Pi_A^-=\1(-\Delta\leq A)$.
Then we have
\begin{itemize}
\item $\dps\lim_{A\to\ii} \norm{\Pi_A^-Q\Pi_A^--Q}_{\gS^{2,1}}=0$;

\smallskip

\item $\dps\lim_{A\to\ii} \cH(\gamma_f+\Pi_A^-Q\Pi_A^-,\gamma_f)=\cH(\gamma,\gamma_f)$;

\smallskip

\item $\dps\lim_{A\to\ii} \norm{\rho_{\Pi_A^-Q\Pi_A^-}-\rho_Q}_{L^r+L^2}=0$ for all $1\leq r<1+2/d$.
\end{itemize}
The last limit is uniform in $\gamma$, for $\cH(\gamma,\gamma_f)\leq C$ and $r\leq 1+2/d-1/C$.
\end{corollary}

\begin{proof}[Proof of Corollary~\ref{cor:density_gS2s_rho}]
For the first two items see~\cite[Cor. 1]{LewSab-13}, only the limit of the density is new. We write
$Q-\Pi_A^-Q\Pi_A^-=\Pi_A^+Q\Pi_A^++\Pi_A^-Q\Pi_A^++\Pi_A^+Q\Pi_A^-.$
These terms are estimated in $L^r+L^2$ in Theorem~\ref{thm:high_momentum_estimate}, with a bound uniform for 
$\cH(\gamma,\gamma_f)\leq C$ and $r\leq 1+2/d-1/C$.
\end{proof}

It remains to write the

\begin{proof}[Proof of Theorem~\ref{thm:high_momentum_estimate}]
We write the proof for $Q$ a smooth finite rank operator and we use~\eqref{eq:integral_representation_lower_bound} as in the proof of Theorem~\ref{thm:Lieb-Thirring}. We have as before
$$Q=\int_0^1\big(\1(\gamma\geq\lambda)-\1(-\Delta\leq S'(\lambda))\big)\,d\lambda,$$
and introduce for simplicity $\mu:=S'(\lambda)$, $P_\mu:=\1(\gamma\geq\lambda)$ and $Q_\mu:=P_\mu-\Pi_\mu^-$, in such a way that
$\Pi_A^+Q\Pi_A^+=\int_0^1\Pi_A^+ Q_\mu \Pi_A^+\,d\lambda$.
We use the notation $Q_\mu^{++}:=\Pi_\mu^+Q_\mu\Pi_\mu^+$, $Q_\mu^{+-}:=\Pi_\mu^+Q_\mu\Pi_\mu^-$ and so on.
We estimate the density of $\Pi_A^+Q_\mu \Pi_A^+$ for every fixed $\mu$, before integrating with respect to $\lambda=(S')^{-1}(\mu)$.

We first deal with the case $\mu\leq A/2$. Then $\Pi_\mu^-\Pi_A^+=0$ and, using that $(p^2-\mu)\1(p^2\geq A)\geq A/2$ for $\mu\leq A/2$, we get
$$\int_{\R^d}\rho_{\Pi_A^+Q_\mu\Pi_A^+}=\int_{\R^d}\rho_{\Pi_A^+Q_\mu^{++}\Pi_A^+}\leq \frac{2}{A} \tr(-\Delta-\mu)Q_\mu^{++}\leq \frac{2}{A} \tr|\Delta+\mu|Q_\mu^2.$$
On the other hand, by the usual Lieb-Thirring inequality, we have
$$\int_{\R^d}\rho_{\Pi_A^+Q_\mu\Pi_A^+}^{1+2/d}\leq C\tr(-\Delta)\Pi_A^+Q_\mu^{++}\leq 2C\tr|\Delta+\mu|Q_\mu^2.$$
By H\"older's inequality we deduce that
$$\norm{\rho_{\Pi_A^+Q_\mu\Pi_A^+}}_{L^{r}}\leq \frac{C}{A^\theta}\big(\tr|\Delta+\mu|Q_\mu^2\big)^{\frac1r},$$
with $r^{-1}=(1-\theta)d/(d+2)+\theta$.

Assume now that $\mu\geq A/2$. Then we are not going to use $A$, but we perform a similar decomposition using $2\mu$. We write
$$\Pi_A^+Q_\mu\Pi_A^+=\Pi_{2\mu}^+Q_{\mu,A}\Pi_{2\mu}^++\Pi_{2\mu}^+Q_{\mu,A}\Pi_{2\mu}^-+\Pi_{2\mu}^-Q_{\mu,A}\Pi_{2\mu}^++\Pi_{2\mu}^-Q_{\mu,A}\Pi_{2\mu}^-$$
and $Q_{\mu,A}:=\Pi_A^+Q_\mu\Pi_A^+$. The previous argument gives
\begin{equation}
\norm{\rho_{\Pi_{2\mu}^+Q_{\mu,A}\Pi_{2\mu}^+}}_{L^{r}}\leq C\frac{\big(\tr|\Delta+\mu|Q_\mu^2\big)^{\frac1r}}{\mu^\theta}\leq C\frac{\big(\tr|\Delta+\mu|Q_\mu^2\big)^{\frac1r}}{A^\theta}.
\label{eq:estim_high_kinetic_1}
\end{equation}
On the other hand, we have $\|\rho_{\Pi_{2\mu}^-Q_{\mu,A}\Pi_{2\mu}^-}\|_{L^\ii}\leq C\mu^{d/2}$ and, by the Lieb-Thirring inequality~\eqref{eq:Lieb-Thirring-zero-temp},
$$\int_{\R^d}\min\left(|\rho_{\Pi_{2\mu}^-Q_{\mu,A}\Pi_{2\mu}^-}|^{1+2/d}\,,\, \mu^{1-\frac{d}2}|\rho_{\Pi_{2\mu}^-Q_{\mu,A}\Pi_{2\mu}^-}|^2\right)\leq C \tr|\Delta+\mu|Q_\mu^2.$$
From these two bounds we deduce that
\begin{equation}
\norm{\rho_{\Pi_{2\mu}^-Q_{\mu,A}\Pi_{2\mu}^-}}_{L^2}\leq C\mu^{\frac{d}{4}-\frac12}\big(\tr|\Delta+\mu|Q_\mu^2\big)^{1/2}.
\label{eq:estim_Qmu2} 
\end{equation}
Finally, the density of $\Pi_{2\mu}^+Q_{\mu,A}\Pi_{2\mu}^-$ can be estimated in $L^2$ by duality:
\begin{align}
\tr\big(\Pi_{2\mu}^+Q_{\mu,A}\Pi_{2\mu}^-V\big) &\leq \|\Pi_{2\mu}^+Q\|_{\gS^2}\|\Pi_{2\mu}^-V\|_{\gS^2}\nonumber\\
 &\leq C\mu^{\frac{d}4-\frac12}\big(\tr|\Delta+\mu|Q^2_\mu\big)^{\frac12}\norm{V}_{L^2}.
\label{eq:estim_off_diagonal_Qmu}
\end{align}
So for $\mu\geq A/2$ we obtain
\begin{equation*}
\norm{\rho_{\Pi_A^+Q_\mu\Pi_A^+}}_{L^r+L^2}\leq C\mu^{\frac{d}4-\frac12}\big(\tr|\Delta+\mu|Q^2_\mu\big)^{1/2}+C\frac{\big(\tr|\Delta+\mu|Q_\mu^2\big)^{\frac1r}}{A^\theta}.
\end{equation*}

Now we integrate with respect to $\lambda$ and appeal to~\eqref{eq:integral_representation_lower_bound}. We use the concavity of $x\mapsto x^{1/r}$ and the Cauchy-Schwarz inequality for the integral over $\lambda$ involving $\mu^{\frac{d}4-\frac12}$. We obtain
\begin{multline*}
\norm{\rho_{\Pi_A^+Q\Pi_A^+}}_{L^r+L^2}\leq C\frac{\cH(\gamma_f+Q,\gamma_f)^{\frac1r}}{A^\theta}\\
+C\cH(\gamma_f+Q,\gamma_f)^{1/2}\left(\int_{0}^1|S'(\lambda)|^{\frac{d}{2}-1}\1(S'(\lambda)\geq A/2)\,d\lambda\right)^{1/2},
\end{multline*}
which is the desired result.

We estimate the cross terms $\Pi_A^\pm Q_\mu\Pi_A^\mp$ by using similar arguments. We start with $\mu\leq A/2$ and therefore look at $\Pi_A^+ Q_\mu\Pi_A^-=\Pi_A^+ Q_\mu^{++}\Pi_A^-+\Pi_A^+ Q_\mu^{+-}$. The second term is estimated by duality like in~\eqref{eq:estim_off_diagonal_Qmu}, leading to
$$\norm{\rho_{\Pi_A^+ Q_\mu^{+-}}}_{L^2}\leq \frac{C\mu^{\frac{d}{4}}}{\sqrt{A}}\big(\tr|\Delta+\mu|Q^2_\mu\big)^{1/2}.$$
For the first term, we remark that for every $V\geq0$
\begin{align*}
\big|\tr\big( \Pi_A^+ Q_\mu^{++}\Pi_A^- V\big)\big|&\leq \norm{\sqrt{V}\Pi_A^+ \sqrt{Q_\mu^{++}}}_{\gS^2}\norm{\sqrt{Q_\mu^{++}}\Pi_A^- \sqrt{V}}_{\gS^2}\\
&\leq A^{\frac\theta2}\tr\big(\Pi_A^+ Q_\mu^{++}\Pi_A^+V\big)+\frac{1}{4A^{\frac\theta2}} \tr\big(\Pi_A^- Q_\mu^{++}\Pi_A^-V\big) ,
\end{align*}
which implies the pointwise inequality
$$|\rho_{\Pi_A^+ Q_\mu^{++}\Pi_A^-}|\leq A^{\frac\theta2}\,\rho_{\Pi_A^+ Q_\mu^{++}\Pi_A^+}+\frac{1}{4A^{\frac\theta2}}\rho_{\Pi_A^- Q_\mu^{++}\Pi_A^-}.$$
We have shown before that $\|\rho_{\Pi_A^+ Q_\mu^{++}\Pi_A^+}\|_{L^r}\leq CA^{-\theta}(\tr|\Delta+\mu|Q_\mu^2)^{1/r}$ and, by the Lieb-Thirring inequality~\eqref{eq:Lieb-Thirring-zero-temp}, we have
$$\norm{\rho_{\Pi_A^- Q_\mu^{++}\Pi_A^-}}_{L^{1+2/d}+L^2}\leq C\left(\big(\tr|\Delta+\mu|Q^2_\mu\big)^{\frac{d}{d+2}}+\mu^{\frac{d}{4}-\frac12}\big(\tr|\Delta+\mu|Q^2_\mu\big)^{\frac12}\right).$$
So we get, for $\mu\leq A/2$,
$$\norm{\rho_{\Pi_A^+ Q_\mu\Pi_A^-}}_{L^r+L^2}\leq \frac{C}{A^{\frac\theta2}}\left(\big(\tr|\Delta+\mu|Q^2_\mu\big)^{\frac1r}+(1+\mu^{\frac{d}{4}-\frac12})\big(\tr|\Delta+\mu|Q^2_\mu\big)^{\frac12}\right).$$

Finally, we deal with the case $\mu\geq A/2$ in a similar fashion. We write
\begin{multline}
\Pi_A^+Q_\mu\Pi_A^-=\Pi_{2\mu}^+Q_\mu\Pi_A^- + \Pi_A^+\Pi_{2\mu}^-Q_\mu^{++}\Pi_A^-+\Pi_A^+Q_\mu^{--}\Pi_A^-\\+\Pi_A^+Q_\mu^{-+}\Pi_A^-+\Pi_A^+\Pi_{2\mu}^-Q_\mu^{+-}\Pi_A^-.
\label{eq:decomp_Qmu}
\end{multline}
The density of $\Pi_{2\mu}^+Q_\mu\Pi_A^-$ is estimated exactly as in~\eqref{eq:estim_off_diagonal_Qmu}. The density of $\Pi_A^+\Pi_{2\mu}^-Q_\mu^{++}\Pi_A^-$ is bounded as before by using
$$2|\rho_{\Pi_A^+\Pi_{2\mu}^-Q_\mu^{++}\Pi_A^-}|\leq \rho_{\Pi_A^-Q_\mu^{++}\Pi_A^-}+\rho_{\Pi_A^+\Pi_{2\mu}^-Q_\mu^{++}\Pi_A^+\Pi_{2\mu}^-}.$$
The two densities on the right can be estimated in $L^\ii$ by $C\mu^{d/2}$ and in $L^{1+2/d}+L^2$ by the Lieb-Thirring inequality. Like in~\eqref{eq:estim_Qmu2}, this leads to
$$\norm{\rho_{\Pi_A^+\Pi_{2\mu}^-Q_\mu^{++}\Pi_A^-}}_{L^2}\leq C\mu^{\frac{d}{4}-\frac12}\big(\tr|\Delta+\mu|Q_\mu^2\big)^{1/2}.$$
The argument is the same for the term involving $Q^{--}$ in~\eqref{eq:decomp_Qmu}. Finally, the two terms involving $Q^{+-}$ and $Q^{-+}$ in~\eqref{eq:decomp_Qmu} are bounded by using directly~\cite[Eq.~(3.11)--(3.12)]{FraLewLieSei-12}, leading to
$$\norm{\rho_{\Pi_A^+Q_\mu^{-+}\Pi_A^-}}_{L^2}+\norm{\rho_{\Pi_A^+\Pi_{2\mu}^-Q_\mu^{+-}\Pi_A^-}}_{L^2}\leq C\mu^{\frac{d}{4}-\frac12}\big(\tr|\Delta+\mu|Q^2_\mu\big)^{1/2}.$$
The final bound is
\begin{equation}
\norm{\rho_{\Pi_A^+Q_\mu\Pi_A^-}}_{L^2}\leq C\mu^{\frac{d}{4}-\frac12}\big(\tr|\Delta+\mu|Q^2_\mu\big)^{1/2}
\label{eq:final_estim_Qmu_off_diagonal}
\end{equation}
for $\mu\geq A/2$ and~\eqref{eq:high_momentum_estimate_off_diagonal} follows after integrating with respect to $\lambda$.
\end{proof}

\section{Free energy \& global well-posedness for generalized Gibbs states}\label{sec:global_positive_temp}

We have proved the existence of local-in-time solutions to the nonlinear Hartree equation 
$$i\dot\gamma=[-\Delta+w\ast\rho_\gamma,\gamma],$$ 
in various situations in Section~\ref{sec:local}, and of global-in-time solutions when $\gamma$ is a local perturbation of the Fermi sea $\gamma_f=\1(-\Delta\leq \mu)$ in Section~\ref{sec:global_zero_temp}. The purpose of this section is to derive a similar result for more general reference states $\gamma_f$, including~\eqref{eq:Fermi-gas-positive-temp}--\eqref{eq:Boltzmann-gas-positive-temp}.

Consider a reference state $\gamma_f=f(-\Delta)$ with $f=(S')^{-1}$ and $S$ satisfying~\eqref{eq:hyp_s_1},~\eqref{eq:hyp_s_2},~\eqref{eq:hyp_s_3} and~\eqref{eq:integrability_condition}. The following \emph{relative free energy}, analogous to the relative energy of Section \ref{sec:global_zero_temp}, is formally conserved and will be used as a Lyapunov function to get global solutions:
\begin{equation}
\cF_f(\gamma,\gamma_f):=\cH(\gamma,\gamma_f)+\frac12\int_{\R^d}\int_{\R^d}w(x-y)\rho_{\gamma-\gamma_f}(x)\rho_{\gamma-\gamma_f}(y)\,dx\,dy.
\label{eq:def_relative_free_energy} 
\end{equation}
This functional is defined on the \emph{free energy space}
\begin{equation}
\cK_f:=\Big\{0\leq\gamma=\gamma^*\leq1\ :\ \cH(\gamma,\gamma_f)<\ii\Big\}.
\label{eq:def_free_energy_space}
\end{equation}
From Klein's inequality~\eqref{eq:Klein2}, we know that any $\gamma\in\cK_f$ satisfies $(1-\Delta)^{1/2}(\gamma-\gamma_f)\in\gS^2$ and from Lemma~\ref{lem:density_1D} and the Lieb-Thirring inequality in Theorem~\ref{thm:Lieb-Thirring}, we know that $\rho_{\gamma}-\rho_{\gamma_f}\in L^2(\R^d)+L^{1+\min(1,2/d)}(\R^d)$. We infer that the last term is well-defined, provided that $w$ decays fast enough. The condition is $w\in L^1(\R^d)$ for $d=1,2$, and $w\in L^1(\R^d)\cap L^{(d+2)/4}(\R^d)$ for $d\geq3$. We have the equivalent of Proposition \ref{prop:relative-energy-positive} for the relative free energy:

\begin{proposition}[Coercivity of the relative free energy in the defocusing case]\label{prop:relative-free-energy-positive}
 Let $\gamma\in\cK_f$. Then, when $d\ge3$ and $\hat{w}\ge0$ we have 
 \begin{equation}
  \cF_f(\gamma,\gamma_f)\ge\cH(\gamma,\gamma_f),
 \end{equation}
while when $d=1,2$ and $\|\hat{w}_-\|_{L^\ii}<K_{\rm LT,S}/(2\pi)^{d/2}$ we have 
\begin{equation}
  \cF_f(\gamma,\gamma_f)\ge\left(1-\frac{(2\pi)^{d/2}\|\hat{w}_-\|_{L^\ii}}{K_{\rm LT,S}}\right)\cH(\gamma,\gamma_f),
\end{equation}
where $K_{\rm LT,S}$ is the Lieb-Thirring constant of Lemma~\ref{lem:density_1D} and Theorem \ref{thm:Lieb-Thirring}.
\end{proposition}

We use this result to globalize the local-in-time solutions to the Hartree equation with finite relative free energy built in Theorem \ref{thm:local-wp-Sps}. 

\begin{theorem}[Global well-posedness at positive temperature for $d=1,2,3$]\label{thm:global_positive_temp}
Assume that $d=1,2,3$ and that $w\in L^1(\R^d)\cap L^\ii (\R^d)$ if $d=1,2$ with furthermore $|\nabla w|\in L^1(\R^d)\cap L^\ii (\R^d)$ if $d=3$. Assume also that $w(x)=w(-x)$ for a.e. $x\in\R^d$. Let $\gamma_f=f(-\Delta)$ with $f=(S')^{-1}$ and $S$ satisfying \eqref{eq:hyp_s_1}, \eqref{eq:hyp_s_2}, \eqref{eq:hyp_s_3} and \eqref{eq:integrability_condition}. Let $Q_0=\gamma_0-\gamma_f$ with $\gamma_0\in\cK_f$. Then there exists a unique maximal solution $Q(t)=\gamma(t)-\gamma_f\in C^0_t((-T^-,T^+),\gS^2)$ if $d=1,2$, and $Q(t)\in C^0_t((-T^-,T^+),\gS^{2,1})$ if $d=3$, of~\eqref{eq:rHF_Q_Duhamel} with initial datum $Q_0$, such that $\gamma(t)\in\cK_f$ for almost all $t\in(-T^-,T^+)$ and $T^\pm>0$. We have the blow-up criterion
\begin{equation}\label{eq:blowup_energy_space_positive_temp}
T^\pm<\ii \Longrightarrow \lim_{t\to \pm T^\pm}\cH(\gamma(t),\gamma_f)=+\ii. 
\end{equation}
The relative free energy is constant for almost every $t\in(-T^-,T^+)$.
Furthermore, when $d=3$ and $\hat{w}\ge0$ or when $d=1,2$ and $\|(\hat{w})_-\|_{L^\ii}<K_{\rm LT,S}/(2\pi)^{d/2}$, then the solution is global, $T^-=T^+=+\ii$.
\end{theorem}

This result holds for a fixed function $S$ and for $0\leq\gamma_0,\gamma_f\leq1$. By scaling and multiplication by a constant we can deal with the general situation and, in particular, get Theorem~\ref{thm:global_positive_temp_intro} stated in the introduction. For instance, in order to deal with density matrices satisfying $\leq M$, it suffices to multiply the time-dependent equation by $1/M$ and to consider the operators $\gamma_f/M$ and $\gamma_0/M$. This changes $w$ into $Mw$ and the entropy $S$ into $M^{-1}S(Mx)$. In dimensions $d=1,2$ the condition on $\widehat{w}$ then involves $M$ through the constant $K_{\rm LT,S}$. Similarly, in the physical examples~\eqref{eq:Fermi-gas-positive-temp}--\eqref{eq:Boltzmann-gas-positive-temp}, the entropy $S$ includes the temperature $T$ and the chemical potential $\mu$. By a scaling argument, we can always reduce the problem to $T=1$ with chemical potential $\mu/T$ and then $K_{\rm LT,S}$ only depends on $\mu/T$. Indeed, if $\gamma_f=f(-\Delta/T)$, then defining the unitary operator $(U_Tf):=T^{d/4}f(\sqrt{T}x)$, we see that the time-dependent equation satisfied by $\tilde\gamma(t):=U_T^*\gamma U_T(t/T)$ is
\begin{equation}
\begin{cases}
i\, \partial_t \tilde\gamma=\big[-\Delta+w_T\ast\rho_{\tilde\gamma}\,,\, \tilde\gamma\big],\\
\tilde\gamma(0)=f(-\Delta)+U_T^*Q(0)U_T,
 \end{cases}
\label{eq:rescaled}
\end{equation}
with $w_T(x)=T^{-1}w(x/\sqrt{T})$ and, hence, $\widehat{w_T}(k)=T^{d/2-1}\widehat{w}(\sqrt{T}k)$. In the Boltzmann and Fermi-Dirac cases, it is easy to verify that the Lieb-Thirring constant $K_{\rm LT,S}$ stays positive when $\mu$ varies in an interval around $0$. This is not true in the Bose-Einstein case since the operator norm of $\gamma_{1,\mu/T}^{\rm bos}$ diverges when $\mu/T\to0^-$, in which case we have to divide the equation by a large constant, as explained before. Thus the constants $\kappa_1$ and $\kappa_2$ appearing in Theorem~\ref{thm:global_positive_temp_intro} stay bounded away from 0 when $\mu/T$ stays in a compact set for fermions and boltzons, but it converges to $0$ when $\mu/T\to0^-$ for bosons.

\medskip

The rest of the section is devoted to the proof of Theorem~\ref{thm:global_positive_temp}. As in Section \ref{sec:global_zero_temp}, we begin by showing conservation of energy for `regular' solutions, namely those belonging to the space $\gS^{1,4}$. 

\begin{proposition}[Conservation of the relative free energy for regular solutions]\label{prop:conservation-relative-free-energy-S12}
 Let $d\ge1$ and $w\in (\ell^1(L^2)\cap L^\ii)(\R^d)$ such that $\partial^\alpha w\in (\ell^1(L^2)\cap L^\ii)(\R^d)$ for all multi-indices $|\alpha|\le 4$. Assume also $w(x)=w(-x)$ for a.e. $x\in\R^d$. Let $\gamma_f=f(-\Delta)$ with $f=(S')^{-1}$ and $S$ satisfying \eqref{eq:hyp_s_1}, \eqref{eq:hyp_s_2}, \eqref{eq:hyp_s_3} and \eqref{eq:integrability_condition} with furthermore $(1+|\cdot|^4) g(\cdot)\in\ell^1(L^2)(\R^d)$ where $g=f(|\cdot|^2)$. Let $Q_0\in\gS^{1,4}$ be such that $\gamma_f+Q_0\in\cK_f$. Let $Q\in C^0_t((-T^-,T^+),\gS^{1,4})$ be the maximal solution to the Hartree equation with initial condition $Q_0$ built in Theorem \ref{thm:local-wp-Sps}. Then, for all $t\in(-T^-,T^+)$, $\gamma_f+Q(t)\in\cK_f$ and the relative free energy is conserved: $\cF_f(\gamma_f+Q(t),\gamma_f)=\cF_f(\gamma_f+Q_0,\gamma_f)$. 
\end{proposition}

\begin{remark}\label{rk:condition-nu1}
 The assumption that $k\mapsto(1+k^4)f(k^2)\in\ell^1L^2$ is satisfied if, for instance, $(1+k^2)^{2+\delta/2}f(k^2)\in L^2(\R^d)$ for some $\delta>d/2$, by \cite[Chap. 4]{Simon-79}. A sufficient condition for this to hold is 
 $\int_0^1\lambda(1+S'(\lambda)_+)^{\delta+4}S'(\lambda)^{d/2}_+\,d\lambda<\ii$
 which, by \eqref{eq:integrability_condition}, is itself implied by the boundedness of $\lambda\mapsto\lambda^{\frac{1}{\delta+4}}S'(\lambda)$ as $\lambda\to0$. 
 In Lemma \ref{lemma:approx-gamma}, we will show that any $S$ satisfying \eqref{eq:hyp_s_1}, \eqref{eq:hyp_s_2}, \eqref{eq:hyp_s_3} and \eqref{eq:integrability_condition} can be approximated by a monotone sequence $(S_n)$ satisfying this. 
\end{remark}

\begin{proof}[Proof of Proposition \ref{prop:conservation-relative-free-energy-S12}]
Let $Q(t)$ be the solution to the Hartree equation with initial datum $Q_0$. We prove that for all $t\in(-T^-,T^+)$
\begin{equation}
\cH(\gamma_f+Q(t),\gamma_f)=\cH(\gamma_f+Q_0,\gamma_f) + \tr(-\Delta)(Q(t)-Q_0)
\label{eq:link-relative-entropy}
\end{equation}
where we recall that $(-\Delta)(Q(t)-Q_0)\in\gS^1$. The proof that the relative free energy is conserved then follows by the same method as for Theorem \ref{prop:conservation-relative-energy-cS1s}, using that $Q(t)$ belongs to $C^0_t\gS^{1,4}\cap C^1_t\gS^{1,2}$ for all times. The rest of the argument is therefore devoted to the proof~\eqref{eq:link-relative-entropy}.

Since $\rho_Q(t)$ belongs to  $L^\ii_{t,{\rm loc}}((-T^-,T^+),L^1_x(\R^d))$, then $V(t)=w*\rho_Q(t)$ belongs to $L^{p'}_{t,{\rm loc}}((-T^-,T^+),L^{q'}_x(\R^d))$ for any $2/p'+d/q'=2$. By \cite{FraLewLieSei-13}, we deduce that for all $t\in(-T^-,T^+)$, $\gamma_f+Q(t)=U_V(t)(\gamma_f+Q_0)U_V(t)^*$, where $U_V(t)-1\in C^0(\gS^{d+2})$, with $U_V(t)$ a strongly continuous family of unitary operators on $L^2(\R^d)$ with $U_V(0)=1$. In particular, $0\le\gamma_f+Q(t)\le1$ for all $t$. From the proof of Theorem~\ref{thm:local-wp-Sps}, we also know that
$(U_V(t)-1)\gamma_f(1-\Delta)\in C^0(\gS^1)$ and that $[\nabla,U_V(t)-1]\in C^0(\gS^{d+2})$. Indeed, computing the last commutator we get a sum of terms in which one of the potentials $V$ is replaced by $\nabla V$, and the estimate is the same as for $U_V(t)-1$ using our assumptions on $w$. In the following we denote $K(t):=U_V(t)-1$.

We only prove~\eqref{eq:link-relative-entropy} for short times, the proof in the general case follows by iterating the argument below. For $t\ll1$ we have the additional information that $\norm{K(t)}_{\gS^{d+2}}<1$ and this implies that we can write
$U_V(t)=\exp(iB(t))$ where $t\mapsto B(t)\in C^0(\gS^{d+2})$, with $B(t)=B(t)^*$. The operator $B(t)$ satisfies the same properties as $K(t)$, namely $B(t)\gamma_f(1-\Delta)\in C^0(\gS^1)$ and $[\nabla,B(t)]\in C^0(\gS^{d+2})$.
We now proceed in two steps. First we replace $B(t)$ by an operator $B_n$ of finite rank and prove~\eqref{eq:link-relative-entropy} with $Q_n:=e^{iB_n}(\gamma_f+Q_0)e^{-iB_n}-\gamma_f$ in place of $Q(t)$. Then we let $n\to\ii$. 

The approximation sequence $B_n=B_n^*$ is chosen such as to have 
\begin{multline*}
\lim_{n\to\ii}\norm{B_n-B(t)}_{\gS^{d+2}}=\lim_{n\to\ii}\norm{(B_n-B(t))\gamma_f(1-\Delta)}_{\gS^1}\\
=\lim_{n\to\ii}\norm{[\nabla,B_n-B(t))]}_{\gS^{d+2}}=0. 
\end{multline*}
One simple way to construct such a sequence $B_n$, is to first truncate the large momenta by looking at $B'_A:=\Pi_A^-B(t)\Pi_A^-$. Since $\gamma_f$, $\nabla$ and $\Delta$ commute with $\Pi_A^-$ it is clear that $B'_A$ converges to $B(t)$ in the above norms as $A\to\ii$. Furthermore, we have $B'_A\in\gS^1$ for the same reason as for $B(t)\gamma_f\in\gS^1$. Then $\Delta B'_A\in\gS^1$ since $\Delta$ is bounded on the range of $\Pi_A^-$. Now we can approximate $B'_A$ in $\gS^1(\Pi_A^-L^2)$ by a finite-rank operator and all the other properties follow immediately. We will use in the proof that $K_n:=e^{iB_n}-1$ is also a finite rank operator with range in $\Pi_A^-L^2$ and that it converges to $K=e^{iB}-1$ for the same norms as for $B$.

We turn to the proof of~\eqref{eq:link-relative-entropy} with $Q_n$ in place of $Q(t)$. Let $P_k$ be an increasing sequence of finite-rank orthogonal projectors whose range contains the range of $B_n$, such that $P_k\to1$ strongly. We choose $P_k$ such as to have $P_k(1-\Delta)P_k\leq C(1-\Delta)$ for all $k$ (the constant can depend on $B_n$). Since $B_n$ lives over the space $\Pi_A^-L^2$, it suffices to take projections adapted to the decomposition of the momentum space into shells, $NA\leq -\Delta\leq (N+1)A$, $N\geq0$. By definition of the relative entropy, we have
  \begin{multline*}
    \cH(\gamma_f+Q_n,\gamma_f) - \cH(\gamma_f+Q_0,\gamma_f)=\\
    \lim_{k\to\ii}\!\!\tr_{P_k\gH}\Big[S(P_k(\gamma_f+Q_0)P_k)-S(P_k(\gamma_f+Q_n)P_k)+S'(P_k\gamma_fP_k)P_k(Q_n-Q_0)P_k\Big].
  \end{multline*}
By construction of $P_k$, we have $P_ke^{iB_n}=e^{iB_n}P_k$ and hence
  \begin{eqnarray*}
    \tr_{P_k\gH}\left[S(P_k(\gamma_f+Q_n)P_k)\right] & =  &\tr_{P_k\gH}\left[e^{iB_n}S(P_k(\gamma_f+Q_0)P_k)e^{-iB_n}\right] \\
    & = & \tr_{P_k\gH}\left[S(P_k(\gamma_f+Q_0)P_k)\right].
  \end{eqnarray*}
Therefore we obtain
  \begin{equation}
    \cH(\gamma_f+Q_n,\gamma_f) - \cH(\gamma_f+Q_0,\gamma_f)=
    \lim_{k\to\ii}\tr_{\gH}\big[P_kS'(P_k\gamma_fP_k)P_k(Q_n-Q_0)\big].
\label{eq:limit_to_be_shown}
  \end{equation}
 
Using the integral representation \eqref{eq:Lowner2} for $S'$ and the inequality
$(P_k(t+\gamma_f)P_k)^{-1}\le P_k(t+\gamma_f)^{-1}P_k$ for all $t>0$ (which follows from the Schur complement formula), we get
  $$-C(1-\Delta)\leq -CP_k(1-\Delta)P_k\le S'(P_k\gamma_fP_k)\le CP_k(1-\Delta)P_k\leq C(1-\Delta).$$
Again from the representation \eqref{eq:Lowner2} for $S'$, we can then show by Lebesgue's dominated convergence theorem that 
$P_kS'(P_k\gamma_fP_k)P_k\to S'(\gamma_f)=-\Delta$
strongly as quadratic forms on $H^1(\R^d)$.
The operator 
$$Q_n-Q_0=K_nQ_0+Q_0K_n+K_nQ_0K_n+K_n\gamma_f+\gamma_fK_n+K_n\gamma_fK_n$$
has a finite rank and its eigenfunctions are in $H^1$, due to our assumption that $Q_0\in\gS^{1,4}$. 
We can therefore pass to the limit $k\to\ii$ in~\eqref{eq:limit_to_be_shown} and get 
\begin{equation*}
\cH(e^{iB_n}(\gamma_f+Q_0)e^{-iB_n},\gamma_f)=\cH(\gamma_f+Q_0,\gamma_f) +\tr(-\Delta)(Q_n-Q_0). 
\end{equation*}

In a second step we take the limit $n\to\ii$. If we can show that 
\begin{equation}
\lim_{n\to\ii}\tr(-\Delta)(Q_n-Q(t))=0,
\label{eq:limit_trace_Qt}
\end{equation}
then we find, using the lower semi-continuity of $\cH$, that 
\begin{equation*}
\cH(\gamma_f+Q(t),\gamma_f)\leq \cH(\gamma_f+Q_0,\gamma_f)  +\tr(-\Delta)\big(Q(t)-Q_0\big). 
\end{equation*}
Exchanging the roles of $t$ and $0$ gives the other inequality and this concludes the proof of~\eqref{eq:link-relative-entropy} for short times.
So it remains to show~\eqref{eq:limit_trace_Qt}. We can write (abbreviating $Q=Q(t)$)
\begin{multline*}
Q_n-Q=(K_n-K)Q_0+Q_0(K_n-K)+(K_n-K)Q_0K_n+KQ_0(K_n-K)\\+(K_n-K)\gamma_f+\gamma_f(K_n-K)+(K_n-K)\gamma_fK_n
+K\gamma_f(K_n-K).
\end{multline*}
That $\tr(-\Delta)(K_n-K)Q_0=\tr(K_n-K)Q_0(-\Delta)\to0$ is obvious, since $Q_0(-\Delta)\in\gS^1$ and $K_n\to K$ in $\gS^{d+2}$. Let us look at
\begin{align*}
\tr(-\Delta)KQ_0(K_n-K)&=\sum_{j=1}^d \tr\; p_jKQ_0(K_n-K)p_j\\
&=\sum_{j=1}^d \tr\; Kp_jQ_0p_j(K_n-K)+\tr\; [p_j,K]Q_0 [K_n-K,p_j]\\
&+\tr\; [p_j,K]Q_0 p_j(K_n-K)+\tr\; Kp_jQ_0[K_n-K,p_j].
\end{align*}
This converges to 0, since $\norm{[p_j,K_n-K]}_{\gS^{d+2}}\to0$, $\norm{K_n-K}_{\gS^{d+2}}\to0$, $Q_0p_j\in\gS^1$ and $p_jQ_0p_j\in\gS^1$. The last assertion follows from the property that $|p||Q_0|^{1/2}\in\gS^2$, by the Araki-Lieb-Thirring inequality and the fact that $Q_0\in\gS^{1,4}$. The argument is similar for all the other terms, and this finishes our proof of Proposition~\ref{prop:conservation-relative-free-energy-S12}.
\end{proof}

Global well-posedness in the energy space $\cK_f$ follows by approximating the initial datum $\gamma_0\in\cK_f$, the interaction potential $w$ and the reference state $\gamma_f$ by a sequence $(\gamma_{0,n},w_n,\gamma_{f,n})$ satisfying the assumptions of Proposition \ref{prop:conservation-relative-free-energy-S12}. The following lemma deals with the construction of $(\gamma_{0,n})$ and $(\gamma_{f,n})$.

\begin{lemma}\label{lemma:approx-gamma}
  Let $d\geq1$. We assume that $S$ satisfies~\eqref{eq:hyp_s_1},~\eqref{eq:hyp_s_2},~\eqref{eq:hyp_s_3} and~\eqref{eq:integrability_condition}. We denote as before $\gamma_f=f(-\Delta)=(S')^{-1}(-\Delta)$. Let $\gamma\in\cK_f$ and $Q:=\gamma-\gamma_f$. Then, there exist two sequences $(S_n)$ and $(\gamma_n)$ such that 
  \begin{itemize}
\item $S_n$ satisfies \eqref{eq:hyp_s_1}, \eqref{eq:hyp_s_2}, \eqref{eq:hyp_s_3} and \eqref{eq:integrability_condition} with additionally $k\mapsto(1+k^4)f_n(k^2)\in\ell^1L^2(\R^d)$, where $f_n:=(S_n')^{-1}$;

  \smallskip

\item $\gamma_n\in\cK_{f_n}$ and  $Q_n:=\gamma_n-\gamma_{f_n}\in\gS^{1,4}$;

\smallskip

  \item  $\dps\lim_{n\to\ii}\int_{\R^d}|f_n(k^2)-f(k^2)|\,dk=\lim_{n\to\ii}\int_{\R^d}k^2|f_n(k^2)-f(k^2)|^2\,dk=0$;
    
  \smallskip

  \item $\dps\lim_{n\to\ii} \cH_n(\gamma_n,\gamma_{f_n})=\cH(\gamma,\gamma_f)$;
  
  \smallskip

  \item $\dps\lim_{n\to\ii} \norm{Q_n-Q}_{\gS^{2,1}}\!=\!\lim_{n\to\ii} \norm{\rho_{e^{it\Delta}Q_ne^{-it\Delta}}-\rho_{e^{it\Delta}Qe^{-it\Delta}}}_{L^\ii_t(L^1+L^2)}=0$.
  \end{itemize}
  Here $\cH_n$ and $\cH$ are the relative entropies of $S_n$ and $S$, respectively. 
\end{lemma}

\begin{proof}[Proof of Lemma \ref{lemma:approx-gamma}]
 We begin by constructing the sequence $(S_n)$. A change of variables leads to the estimate
\begin{equation*}
 \int_{\R^d}|f_n(k^2)-f(k^2)|\,dk\le C\int_0^{f(0)}|\lambda-(S'_n)^{-1}(S'(\lambda))|S'(\lambda)^{\frac d2 -1}|S''(\lambda)|\,d\lambda.
\label{eq:estim_S_to_be_shown}
\end{equation*}
Let $\nu_1$ and $\nu_2$ be the L\"owner measures associated with $S$, as in \eqref{eq:Lowner2}. 
We define $S_n$ by the same formula \eqref{eq:Lowner2} as $S$, except that we replace $a'$ by $a_n'\le a'$ and the measure $d\nu_1$ by 
 $d\nu_{1,n}(t):=\left(t^{1-\bar\alpha}\1_{0\le t\le1/n}+\1_{t>1/n}\right)d\nu_1(t),$
 where $0\le\bar\alpha\le1$ is defined in the following fashion. First we look at 
 $$\alpha^*=\sup\left\{0\le\alpha\le1,\,\int_0^1\frac{d\nu_1(t)}{t^\alpha}<+\ii\right\},$$
 which is finite since $\int_0^1\,d\nu_1(t)<+\ii$. We distinguish two cases. If $\alpha^*>(d+6)/(d+8)$, then we set $\bar\alpha=1$, that is, $\nu_{1,n}=\nu_1$ and we claim that we already have $k\mapsto(1+k^4)f(k^2)\in\ell^1L^2$. Indeed, by definition of $\alpha^*$, there exists $\alpha>(d+6)/(d+8)$ such that 
 $\int_0^1t^{-\alpha}d\nu_1(t)<\ii$.
 Defining $\delta=\frac{1}{1-\alpha}-4$, the condition $\alpha>(d+6)/(d+8)$ implies that $\delta>d/2$. Using that $S'(\lambda)\sim C+\int_0^1(t+\lambda)^{-1}\,d\nu_1(t)$ we see that $\lambda\mapsto\lambda^{\frac{1}{\delta+4}}S'(\lambda)$ is bounded when $\lambda\to0$. By Remark \ref{rk:condition-nu1} we deduce that $k\mapsto(1+k^4)f(k^2)\in\ell^1L^2$, and there is nothing to do. 

Now, let us assume that $\alpha^*\le(d+6)/(d+8)$. This implies that $\alpha^*+2/(d+8)\le1$, and we pick any number $\bar\alpha$ such that $0\le\alpha^*<\bar\alpha<\alpha^*+2/(d+8)\le1$.
 The sequence $a_n'\le a'$ is chosen in the following fashion: we choose any sequence $a'_n\to a'$ satisfying the property that
 $a'-a_n'\ge4b\int_0^{1/n}d\nu_1(t).$
 Since for all $n$, $\nu_{1,n}\le\nu_1$ as positive Borel measures, $S_n$ satisfies conditions \eqref{eq:hyp_s_1} and \eqref{eq:hyp_s_2}. By the monotone convergence theorem, we have for all $x>0$,
 $$\int_0^\ii\frac{d\nu_{1,n}(t)}{(t+x)(1+2t)}\to\int_0^\ii\frac{d\nu_1(t)}{(t+x)(1+2t)}$$
 and hence $S'_n(x)\to S'(x)$ as $n\to\ii$. Furthermore, using only $a_n'\le a'$ and $\nu_{1,n}\le\nu_1$, we infer that $S'_n(x)\le S'(x)$ for all $0\le x\le1/2$. When $1/2\le x\le 1$, we have
 \begin{eqnarray*}
  S'(x)-S_n'(x) & = & (a'-a_n')x-b(2x-1)\int_0^{1/n}\frac{(1-t^{1-\bar\alpha})d\nu_1(t)}{(t+x)(2t+1)} \\
    & \ge & \frac{a'-a_n'}{2}-2b\int_0^{1/n}d\nu_1(t) \ge 0.
 \end{eqnarray*}
 Hence, $S_n'(x)\le S'(x)$ for all $0\le x\le 1$, and since $S_n'$ and $S'$ are decreasing, we deduce that $(S_n')^{-1}(\lambda)\le (S')^{-1}(\lambda)$ for all $\lambda$. By the monotone convergence theorem, we deduce that 
 $$\lim_{n\to\ii}\int_0^{f(0)}|\lambda-(S'_n)^{-1}(S'(\lambda))|S'(\lambda)^{\frac d2 -1}|S''(\lambda)|\,d\lambda=0.$$
 From the fact that $S'_n\le S'$, we also deduce that $S_n$ satisfies condition \eqref{eq:integrability_condition} and condition \eqref{eq:hyp_s_3} for all $n$. Because $\bar\alpha<\alpha^*+2/(d+8)$, there exists $d/2<\delta<1/(\bar\alpha-\alpha^*)-4$, and
 $$\int_0^1\frac{d\nu_{1,n}(t)}{t^{1-\frac{1}{\delta+4}}}=\int_0^{1/n}\frac{d\nu_1(t)}{t^{\bar\alpha-\frac{1}{\delta+4}}}+\int_{1/n}^1\frac{d\nu_1(t)}{t^{1-\frac{1}{\delta+4}}}<+\ii$$ 
 since $\bar\alpha-1/(\delta+4)<\alpha^*$. As before this implies that $(1+k^4)f_n(k^2)\in\ell^1L^2$.
Finally, we have the inequality
 \begin{multline*}
    \int_{\R^d}k^2|f_n(k^2)-f(k^2)|^2\,dk\\
    \le\left(\norm{k^2f(k^2)}_{L^\ii_k}+\sup_n\norm{k^2f_n(k^2)}_{L^\ii_k}\right)\int_{\R^d}|f_n(k^2)-f(k^2)|\,dk.
 \end{multline*}
 Now $\norm{k^2f(k^2)}_{L^\ii_k}=\norm{xS'(x)_+}_{L^\ii_x}<\ii$ by~\eqref{eq:limit_end_points}, and it is also clear that $\sup_n\norm{xS_n'(x)_+}_{L^\ii_x}<+\ii$ since $S_n'\le S'$. 
We now construct the sequence $(\gamma_n)$. Corollary \ref{cor:density_gS2s_rho} and Theorem \ref{thm:high_momentum_estimate} show that we can find a sequence $(\tilde{Q}_n)\subset\gS^{2,s}$ for all $s\ge0$, with furthermore $Q_n=\Pi_{A_n}^-Q_n\Pi_{A_n}^-$ for all $n$, for some sequence $A_n\to+\ii$, and such that 
 $\cH(\gamma_f+\tilde{Q}_n,\gamma_f)\to\cH(\gamma,\gamma_f)$, $\tilde{Q}_n\to Q$ in $\gS^{2,1}$, and $\rho_{e^{it\Delta}\tilde{Q}_ne^{-it\Delta}}\to\rho_{e^{it\Delta}Qe^{-it\Delta}}$ in $L^\ii(\R,L^1+L^2)$. By~\cite[Cor. 1]{LewSab-13}, we can assume that each $\tilde{Q}_n$ has finite-rank and hence belongs to $\gS^{1,4}$ since it is compactly supported in Fourier space. By the formula defining $S'_n$, we see that $S-S_n$ also satisfies \eqref{eq:hyp_s_1} and \eqref{eq:hyp_s_2}. Hence
 $$\cH(\gamma_f+\tilde{Q}_n,\gamma_f)-\cH_n(\gamma_f+\tilde{Q}_n,\gamma_f)=\cH_{s-S_n}(\gamma_f+\tilde{Q}_n,\gamma_f)\ge0.$$
 Defining $X_n=(\gamma_{f_n}/\gamma_f)^{1/2}$ which satisfies $0\le X_n\le 1$ since $(S_n)^{-1}\le (S')^{-1}$, we have by monotony of the relative entropy (see~\cite[Thm. 1]{LewSab-13})
 $$\cH_n(\gamma_f+\tilde{Q}_n,\gamma_f)\ge\cH_n(X_n(\gamma_f+\tilde{Q}_n)X_n,X_n\gamma_fX_n)=\cH_n(\gamma_{f_n}+X_n\tilde{Q}_nX_n,\gamma_{f_n}).$$
 We thus set $Q_n:=X_n\tilde{Q}_nX_n$ and $\gamma_n:=\gamma_{f_n}+Q_n$. By the previous inequality, we have
 $\limsup_{n\to\ii}\cH_n(\gamma_n,\gamma_{f,n})\le\cH(\gamma,\gamma_f)$.
 Using the weak lower semi-continuity of the relative entropy in finite dimension gives the reverse inequality with a liminf,
 and hence $\cH_n(\gamma_n,\gamma_{f,n})\to\cH(\gamma,\gamma_f)$. Since $X_n$ is a multiplication in Fourier space and since $X_n\to1$ strongly in $L^2(\R^d)$, we also deduce that $Q_n\to Q$ in $\gS^{2,1}$. Finally, Lemma \ref{lem:density_Schatten-Sobolev} gives that $\rho[e^{it\Delta}Q_ne^{-it\Delta}]-\rho[e^{it\Delta}\tilde{Q}_ne^{-it\Delta}]\to0$ in $L^\ii(\R,L^1+L^2)$.
\end{proof}

\begin{proof}[Proof of Theorem \ref{thm:global_positive_temp}]
 Let $(\gamma_0,w,\gamma_f)$ be as in the statement of the theorem. Let $(\gamma_{0,n},\gamma_{f,n})$ be given by Lemma \ref{lemma:approx-gamma}. Let also $(w_n)\subset(\ell^1L^2\cap L^\ii)(\R^d)$ be such that $\partial^\alpha w_n\in(\ell^1L^2\cap L^\ii)(\R^d)$ for all multi-indices $|\alpha|\le4$ and all $n$, with furthermore $w_n(x)=w_n(-x)$ for a.e. $x\in\R^d$, $w_n\to w$ in $(L^1\cap L^\ii)(\R^d)$ if $d=1,2$, and additionally $\nabla w_n\to\nabla w$ in $(L^1\cap L^\ii)(\R^d)^d$ when $d=3$. Let $Q_n\in\cC^0((-T^-_n,T^+_n),\gS^{1,4})$ be the unique solution to \eqref{eq:rHF_Q_Duhamel} associated to the triplet $(\gamma_{0,n},w_n,\gamma_{f,n})$ given by Proposition \ref{prop:conservation-relative-free-energy-S12}. Let also $Q\in C^0((-T^-,T^+),\gS^2)$ if $d=1,2$ and $Q\in C^0((-T^-,T^+),\gS^{2,1})$ if $d=3$ the solution to \eqref{eq:rHF_Q_Duhamel} built in Theorem \ref{thm:local-wp-S2} and Theorem \ref{thm:local-wp-S21}. By these theorems, for all $\epsilon>0$ and for all $n$ large enough, $T^\pm_n\ge T^\pm-\epsilon$ and we have 
 $V_{Q_n}\to V_Q$ in $L^{8/d}_t([-T^-+\epsilon,T^+-\epsilon],L^2_x\cap L^\ii_x)$, $Q_n(t)\to Q(t)$ in $\cC^0([-T^-+\epsilon,T^+-\epsilon],\gS^2)$. Using the weak lower semi continuity of the relative entropy in finite dimension, we infer that for all $t$
 $$\cH(\gamma_f+Q(t),\gamma_f)\le\liminf_{n\to\ii}\cH_n(\gamma_{f,n}+Q_n(t),\gamma_{f,n}).$$
The estimate
 \begin{multline*}
  \left|\iint w(x-y)\rho_{Q_n(t)}(x)\rho_{Q_n(t)}(y)\,dxdy-\iint w(x-y)\rho_{Q(t)}(x)\rho_{Q(t)}(y)\,dxdy\right|\\
  \le\norm{\rho_Q+\rho_{Q_n}}_{L^1+L^2}\norm{V_{Q_n(t)}-V_{Q(t)}}_{L^2\cap L^\ii},
 \end{multline*}
 together with H\"older's inequality give 
 $$\iint w(x-y)\rho_{Q_n(t)}(x)\rho_{Q_n(t)}(y)\,dxdy\to\iint w(x-y)\rho_{Q(t)}(x)\rho_{Q(t)}(y)\,dxdy$$
 in $L^{4/d}_t([-T^-+\epsilon,T^+-\epsilon])$, and thus almost everywhere up to a subsequence. Taking the limit $n\to\ii$ in the relative free energy of $Q_n$, we find that 
$\cF_f(\gamma_f+Q(t),\gamma_f)\le\cF_f(\gamma_0,\gamma_f)$
 almost everywhere. Putting the initial time at any $t$ for which there is convergence gives the reverse inequality.
The blow-up criterion \eqref{eq:blowup_energy_space_positive_temp} follows from the Lieb-Thirring inequality and the blow-up criterion of Theorem \ref{thm:local-wp-S2} and Theorem \ref{thm:local-wp-S21}. Assuming the defocusing conditions on $\hat{w}$, using Proposition \ref{prop:relative-free-energy-positive} and the Lieb-Thirring inequality, we deduce that the solution $Q$ is global: $T^\pm=+\ii$. This concludes the proof of our main theorem.
\end{proof}

\begin{remark}
The defocusing condition also implies that the energy is constant everywhere (and not only almost everywhere): for all $t\in\R$ and for all $n$, we have
$\cH_n(\gamma_{f,n}+Q_n(t),\gamma_{f,n})\le c_n\cF_{f_n}(\gamma_{f_n}+Q_n(t),\gamma_{f_n})$, with $(c_n)$ bounded. Thanks to Theorem \ref{thm:high_momentum_estimate} and since $Q_n(t)\to Q(t)$ in $\gS^2$, this implies that $\rho_{Q_n(t)}\to\rho_{Q(t)}$ in $L^1+L^2$ for \emph{all} $t$ (and not only on a full measure set). Hence, $\cF_f(\gamma_f+Q(t),\gamma_f)\le\cF_f(\gamma_0,\gamma_f)$ for all $t$ and we have conservation of the energy by exchanging $t$ and $0$.
\end{remark}

\begin{table}[h]
\scriptsize
\begin{tabular}{lll|lll}
$\gamma_f$ & (reference state) & Eq.~\eqref{eq:gamma_f_intro} & 
$\gS^p$ & (Schatten space) & Eq.~\eqref{eq:Schatten}\\
$\rho_Q$ & (density) & Sec.~\ref{sec:local} &
$\gS^{p,s}$ & (Schatten-Sobolev) & Eq.~\eqref{eq:def_Schatten-Sobolev} \\
$\Pi_\mu^\pm$ & (Fermi sea) & Eq.~\eqref{eq:def_Pi_mu} &
$\ell^p(L^q)$ & & Eq.~\eqref{eq:def_ell_p_L_q}\\
$\cE(A,B)$ & (relative energy) & Eq.~\eqref{eq:formal_relative_energy}, \eqref{eq:def-relative-energy} &
$\cK_\mu$ & (energy space) & Eq.~\eqref{eq:def_K_mu} \\
$\cF(A,B)$ & (relative free energy) & Eq.~\eqref{eq:formal_relative_free_energy} &
$\cX_\mu$ & & Eq.~\eqref{eq:def_X_mu}\\
$\cH(A,B)$ & (relative entropy) & Eq.~\eqref{eq:def_relative_entropy_1},~\eqref{eq:def_relative_entropy} &
$\cY_\mu$ & & Eq.~\eqref{eq:def_Y_mu} \\
$\cF_f(A,B)$ & (gen. rel. free energy)  & Eq.~\eqref{eq:def_relative_free_energy} &
$\cK_f$ & (free energy space) & Eq.~\eqref{eq:def_free_energy_space} \\
$\cW_V(s,t)$ & (wave operator) & Eq.~\eqref{eq:Dyson-series} &
$\cW_V^{(n)}(s,t)$ & ($n$th wave op.) & Eq.~\eqref{eq:nth_wave_op} 
\end{tabular}

\medskip

\caption{List of notation.\label{tab:notation}}
\end{table}

%
%
%
%


\noindent\textbf{Acknowledgements.} This work was partially done while the authors were visiting the Centre \'Emile Borel at the Institut Henri Poincar\'e in Paris. The authors acknowledge financial support from the European Research Council under the European Community's Seventh Framework Programme (FP7/2007-2013 Grant Agreement MNIQS 258023), and from the French ministry of research (ANR-10-BLAN-0101).


\end{document}